\documentclass[acmsmall,screen,nonacm]{acmart} 
\usepackage{algorithm2e} 
\usepackage{cleveref} 
\usepackage{xcolor} 

\newcommand{\CASP}{Finite-Choice Logic Program}
\newcommand{\casp}{finite-choice logic program}

\newcommand{\CASPS}{Finite-Choice Logic Programs}
\newcommand{\casps}{finite-choice logic programs}
\newcommand{\Casping}{Finite-choice logic programming}
\newcommand{\casping}{{finite-choice logic programming}}
\newcommand{\CASPING}{Finite-Choice Logic Programming}

\newcommand{\dusa}{Dusa}

\DeclareMathOperator{\is}{{\color{blue}{\underline{\textrm{\textsc{is}}}}}}
\DeclareMathOperator{\isish}{{\color{blue}{\underline{\textrm{\textsc{is?}}}}}}

\newcommand{\term}[1]{{\color{teal}\mathsf{#1}}}
\newcommand{\pred}[1]{{\color{gray}\mathit{#1}}}

\newcommand{\prop}[3]{{{\pred{#1}({#2})} \is {{#3}}}}
\newcommand{\propK}[2]{{{\pred{#1}(#2)}}}
\newcommand{\propV}[2]{{{\pred{#1}} \is {{#2}}}}

\newcommand{\propish}[3]{{{\pred{#1}({#2})} \isish {{#3}}}}
\newcommand{\propishV}[2]{{{\pred{#1}} \isish {{#2}}}}

\newcommand\ruleconc{\ensuremath{{H}}}

\newcommand\compatible\parallel
\newcommand\incompatible{\not\parallel}

\newcommand\Domain{\ensuremath{\mathscr{D}}}
\newcommand\Constraint{\ensuremath{\mathit{Constraint}}}
\newcommand\DB{\ensuremath{\mathit{DB}}}
\newcommand\Choice{\ensuremath{\mathit{Choice}}}

\newcommand\C{\ensuremath{\mathcal{C}}}

\newcommand\just[1]{{\color{blue}{\mathsf{just}}}\,{#1}}
\newcommand\none[1]{{\color{violet}{\mathsf{noneOf}}}\,{#1}}

\newcommand\immcons[1]{\ensuremath{\tau_{#1}}}
\newcommand\bigimmcons[1]{\ensuremath{\mathcal{T}_{#1}}}

\newcommand\rulechoices[1]{\ensuremath{\langle{#1}\rangle}}

\DeclareMathOperator{\setstepAarrow}{\Rightarrow}
\newcommand{\setstepA}[1]{\setstepAarrow_{#1}}
\renewcommand{\emptyset}{\varnothing}
\newcommand\of{\mathop{\in}} 
\newcommand\lfp{\mathop{\mathrm{lfp}}}

\newcommand\suggestedchange[2]{#2} 
\newcommand\suggestedchangerationale[3]{#2} 

\newcommand\reforexpanded[2]{\Cref{#1}} 

\setcopyright{rightsretained}
\acmDOI{10.1145/3704849}
\acmYear{2025}
\acmJournal{PACMPL}
\acmVolume{9}
\acmNumber{POPL}
\acmMonth{1}

\acmSubmissionID{popl25main-p44-p}
\begin{document}

\title{Finite-Choice Logic Programming}

\author{Chris Martens}
\email{c.martens@northeastern.edu}
\orcid{0000-0002-7026-0348}
\affiliation{%
  \institution{Northeastern University}
  \city{Boston}
  \state{Massachusetts}
  \country{USA}
}

\author{Robert J. Simmons}
\email{robsimmons@gmail.com}
\orcid{0000-0003-2420-3067}
\affiliation{
  \institution{Unaffiliated}
  \city{Boston}
  \state{Massachusetts}
  \country{USA}
}

\authornote{The first and second authors contributed equally to the design of {\casping} and {\dusa}. The second author performed the implementation and analysis (\Cref{sec-implementation-in-dusa}).}

\author{Michael Arntzenius}
\orcid{0009-0002-0417-5636}
\affiliation{%
    \institution{Unaffiliated}
    \city{Hamilton Township}
    \state{New Jersey}
    \country{USA}
}

\authornote{The first and third authors contributed equally to the fixed-point semantics in \Cref{sec-immcons,sec-semantics}.}

\begin{abstract}
  Logic programming, as exemplified by datalog, defines the meaning of a program as its unique smallest model: the deductive closure of its inference rules.  
However, many problems call for an enumeration of models that vary along some set of choices while maintaining structural and logical constraints---there is no single canonical model. 
The notion of {\em stable models} for logic programs with negation has successfully captured programmer intuition about the set of valid solutions for such problems, giving rise to a family of programming languages and associated solvers known as answer set programming.
Unfortunately, the definition of a stable model is frustratingly indirect, especially in the presence of rules containing free variables.

We propose a new formalism, \emph{\casping,} that uses choice, not negation, to admit multiple solutions. 
{\Casping} contains all the expressive power of the stable model semantics,
gives meaning to a new and useful class of programs, and enjoys a least-fixed-point interpretation over a novel domain. We present an algorithm for exploring the solution space and prove it correct with respect to our semantics. Our implementation, the {\dusa} logic programming language, has performance that compares favorably with state-of-the-art answer set solvers and exhibits more predictable scaling with problem size.
\end{abstract}



\maketitle

\pdfimageresolution=300

\section{Introduction}

Many important problem domains involve generating varied data according to structural and logical constraints. Examples include
property-based random testing for typed functional
programs~\cite{goldstein2023reflecting,goldstein2022parsing,paraskevopoulou2022computing,lampropoulos2017generating,seidel2015type,claessen2011quickcheck}, procedural content generation in games~\cite{short2017procedural,shaker2016procedural,dormans2011generating},
and software configuration~\cite{le2011ifdef,czarnecki2004staged,czarnecki2002generative}.
To solve these problems, we write {\em
generative programs,} characterized by {\em choice points} that create
multiple possible outcomes and {\em constraints} that eliminate undesirable outcomes.
\suggestedchangerationale{Generative programs often make use of randomness to produce different solutions when the program is run multiple times.}{}{{tightening up prose \todo{consider adding back}}}
From a
declarative perspective, these generative programs characterize {\em possibility
spaces}, and the meaning of a program is the entire space, which may be considered independently of particular methods for
sampling individual solutions. 

Using logic to describe possibility spaces has a long
history, and its expressive power has been harnessed by established
tools such as satisfiability solvers and logic programming languages.
However, different logic-based tools disagree in important ways about how
logical statements should be interpreted and thus how possibility spaces
should be expressed.

\clearpage 
Consider the problem of procedurally generating maps for a virtual world 
of connected regions where every region
has one of three terrain types---$\pred{mountain}$, $\pred{ocean}$, or
$\pred{forest}$---and where oceans never directly adjoin mountains.
We can model the possibility space of terrain maps as a Boolean satisfiability
(SAT) problem by saying that the truth of a proposition
$\propK{mountain}{r}$ means that a region $r$ has mountainous terrain (and
so on):
\begin{gather} \left(\propK{mountain}{r} \vee \propK{forest}{r} \vee
  \propK{ocean}{r}\right) \leftrightarrow  \propK{region}{r} \label{eqn-sat1}\\
  (\neg\propK{mountain}{r} \wedge \neg\propK{forest}{r}) \vee
  (\neg\propK{mountain}{r} \wedge \neg\propK{ocean}{r}) \vee
(\neg\propK{forest}{r} \wedge \neg\propK{ocean}{r}) \label{eqn-sat2}\\
(\propK{adjacent}{r_1, r_2} \wedge \propK{ocean}{r_1}) \rightarrow
  (\propK{ocean}{r_2} \vee \propK{forest}{r_2}) \label{eqn-sat3}
\end{gather}
The conjunction of these three propositions neatly captures the
desired possibility space: when augmented with a finite
universe of \suggestedchange{terms (possible regions)}{possible regions} and a collection of
input literals
characterizing the $\pred{adjacent}$ relation,
a SAT solver would return valid terrain-to-region assignments.

Now suppose we additionally wish to enumerate all regions reachable from a given starting region by traversing only forests, 
perhaps with the goal of checking whether a player's initial position can reach both mountains and ocean.
\suggestedchangerationale{Reachability is naturally characterized by introducing a recursively-defined proposition $\propK{reach}{r}$ which indicates that the region $r$ is reachable from the starting region:}{The fact that a region $r$ is reachable from a starting region is naturally characterized with a recursively-defined proposition $\propK{reach}{r}$ as follows:}{tightening up prose}
\begin{align}
  \propK{reach}{r}  & \leftarrow \propK{start}{r}\label{eqn-reach1}\\
  \propK{reach}{r_2} & \leftarrow 
  \propK{reach}{r_1}, \propK{forest}{r_1},
  \propK{adjacent}{r_1,r_2}  \label{eqn-reach2}
\end{align}
In this instance, logic programming gives the desired interpretation: from a rule $p \leftarrow q_1\ldots q_n$ (where $p$ and $q_i$ are all positive assertions) and a database containing $q_1\ldots q_n$, we deduce $p$.
Such a program has a single canonical model, the smallest set of assertions closed under the given implications;
any other assertion is deemed false~\cite{gelder-ross-schlipf91wellfounded}.

Unfortunately, these two approaches don't play well together. On one hand, logic
programming lacks a notion of \emph{choice point,} because propositions are interpreted as rules:
not merely justifications for what {\em may} appear in the model, but assertions about what {\em must} appear in the model, so the first program has no direct
analog in, say, datalog. On the other hand, defining \suggestedchangerationale{transitive closure}{reachability}{transitive closure is unnecessary jargon here} is notoriously difficult for SAT solvers:
interpreting the reachability program as a logical formula using Clark's
completion \shortcite{clark78negation}, a SAT solver would validate
solutions with spurious $\pred{reach}$ facts that have no justification
imparted by the programmer.

The programming model we desire seems to alternate between two modes
of operation: on the one hand, making {\em mutually exclusive choices}
that multiply our possibility space and rule out certain combinations
with other choices; and on the other hand, computing the 
{\em deductive consequences} of those choices. 
Unifying these perspectives
is precisely what we achieve in this paper with {\casping}, a new approach to logic programming
in which the generative constraints in \suggestedchange{\cref{eqn-sat1,eqn-sat2,eqn-sat3}}{formulas~\ref{eqn-sat1}-\ref{eqn-sat3}} can be expressed as follows:
\begin{gather}
  \prop{terrain}{r}{\{ \term{mountain}, \term{forest}, \term{ocean} \}} 
    \leftarrow \propK{region}{r} \label{eqn-dusa-demo1}\\
  \prop{terrain}{r_1}{\{\term{forest},\term{ocean}\}} 
    \leftarrow \propK{adjacent}{r_1,r_2}, \prop{terrain}{r_2}{\term{ocean}} \label{eqn-dusa-demo2}
\end{gather}

\noindent
The $\is$ syntax signals a \emph{functional dependency} from the region $r$ to its
terrain type, and the rules represent mutually exclusive choices closed
under deduction: a region must have one of three terrain types, and a
region adjacent to the ocean must contain oceans or forests. In {\casping}, such rules combine naturally with rules like \suggestedchange{\cref{eqn-reach1,eqn-reach2}}{formulas~\ref{eqn-reach1}~and~\ref{eqn-reach2}}; we would rewrite the latter rule as $\left(\propK{reach}{r_2} \leftarrow  \propK{reach}{r_1}, \prop{terrain}{r_1}{\term{forest}}, \propK{adjacent}{r_1, r_2}\right)$.

\clearpage 
\subsection{Answer Set Programming and Its Discontents}

Historically, the most successful approach for logic programming with mutually exclusive choices has been the {\em stable model semantics} introduced by Gelfond and
Lifschitz~\shortcite{gel88}. Stable models combine the Boolean-satisfiability intuitions behind Clark's completion with a rejection of circular justification, and provide the foundation for \textit{answer set programming\suggestedchange{}{ (ASP)}}.
Systems for computing the stable models of answer set
programs have fruitfully co-evolved with the advancements in Boolean
satisfiability solving, leading to sophisticated heuristics that make many
problems fast in practice~\cite{gebser2017multi}.
\suggestedchange{Answer set programming}{ASP} has seen considerable and ongoing success in procedural content generation~\cite{smith2011answer,smith2013quantifying,smith2014logical,neufeld2015procedural,summerville2018gemini,dabral2020generating},
analysis and scenario generation for legal regulations, policies, and contracts~\cite{dabral2022exploring,lim2022automating},
spatial reasoning for built structures~\cite{li2020spatial},
and distributed system reasoning~\cite{alvaro2011dedalus}.

In a modern answer set
programming language like Clingo \cite{gebser2011potassco}, the map-generation critera expressed in \suggestedchange{\cref{eqn-sat1,eqn-sat2,eqn-sat3}}{formulas \ref{eqn-sat1}-\ref{eqn-sat3}} can be expressed concisely as follows:
\begin{gather} 1 \left\{ \, \propK{mountain}{r}; \, \propK{forest}{r};
\,\propK{ocean}{r} \, \right\} 1 \leftarrow \propK{region}{r}
\label{eqn-choice-rule-ex1} \\
\leftarrow \propK{adjacent}{r_1,r_2}, \propK{ocean}{r_1}, \propK{mountain}{r_2} \label{eqn-choice-rule-ex2}
\end{gather}
\suggestedchange{\Cref{eqn-choice-rule-ex1}}{Formula~\ref{eqn-choice-rule-ex1}} contains a ``cardinality constraint'' that insists on precisely one of the three terrain types holding for each region. \suggestedchange{\Cref{eqn-choice-rule-ex2}}{Formula~\ref{eqn-choice-rule-ex2}} is a ``headless'' rule, interpreted as a constraint forbidding the conjunction of all three premises from holding simultaneously.

\suggestedchange{Unfortunately, the interpretation of answer set programs involves multiple levels of indirection that make it challenging to reason about them.}{However, the interpretation of answer set programs involves multiple levels of indirection that, in the authors' experience, make it challenging to reason about program meaning and performance.}
First, the stable model semantics that justifies answer set programming is defined only for logic programs without free variables.
Answer set programming presents clients with an interface of rules containing free variables, such as \suggestedchange{\cref{eqn-choice-rule-ex2}}{formula~\ref{eqn-choice-rule-ex2}} above;
however, in a program with sixteen regions, the stable model semantics applies after expanding this rule into 256 variable-free rules, one for each assignment of regions to $r_1$ and $r_2$.  
This is reflected in essentially all implementations of answer set programming, which involve the interaction of a \textit{solver} that only understands variable-free rules and a \textit{grounder} that generates variable-free rules, usually incorporating heuristics to minimize the number of rules the solver must deal with \cite{Kaufmann_Leone_Perri_Schaub_2016}.

The second level of indirection is that even propositional answer set programs do not directly have a semantics:
the definition of stable models involves a syntactic transformation of an answer set program into a logic program without negation (the
\emph{reduct}, see \cite{gel88}).
This transformation is defined with respect to a candidate model, and if the unique model of the reduct is the same as the candidate model, the candidate model is accepted as an actual model.
This fixed-point-like definition is what the ``stability'' of stable models refers to.

The third level of indirection is that stable models are defined in terms of \emph{negation} and not in terms of the higher-level constructs --- choice, cardinality, and headless rules --- that are foundational for essentially all modern applications of answer set programming.
Higher-level constructs are either justified by translation to ``pure'' answer set programming with only negation (as in \cite{sacca1990stable}) or by appeal to a nonstandard definition of stable models (as in \cite{SIMONS2002181}).

\clearpage 
\subsection{A Constructive Semantics Based on Mutually-Exclusive Choices}
\label{sec:mutex}

In this paper we develop {\casping} as a logic programming language with \emph{choice,} not negation, as primitive.
We present and connect two semantics for {\casping}: a set-based semantics that describes the incremental construction of solutions (\Cref{sec-definition}) and an interpretation of programs as the least fixed point of a monotonic immediate consequence operator over the domain
of possibility spaces (\Cref{sec-semantics}). 
We demonstrate that {\casping} subsumes both datalog (\Cref{sec-simulating-datalog}) and answer set programming (\Cref{sec-simulating-asp}), and present new examples of idiomatic {\casp{s}} (\Cref{sec-examples}). 
We present an algorithm for nondeterministic enumeration of solutions (\Cref{sec-implementation}) 
and describe an implementation (\Cref{sec-implementation-in-dusa}) that outperforms state-of-the-art answer set programming engines on a variety of examples, in part because our implementation avoids the
``grounding bottleneck'' encountered by ASP solvers that take a ground-then-solve approach to program execution.

\section{Defining {\CASPING}} \label{sec-definition}

At the core of {\casping} are facts, which take the form $\prop{p}{\overline{t}}{v}$, where $\propK{p}{\overline{t}}$ is the attribute and $v$ is the unique value assigned to that attribute. A set of facts (which we'll call a \textit{database}) must map each attribute to at most one value.
Concretely, a fact like $\prop{terrain}{\term{home}}{\term{forest}}$ indicates
that the region named $\term{home}$ contains forest terrain, and cannot contain any other terrain.
This is called a \emph{functional dependency}: the predicate
$\pred{terrain}$ is a partial function from region names to terrain types.

\begin{definition}[Terms]\label{def-term}
As common in logic programming settings, terms are Herbrand structures, 
either variables $x,y,z, \ldots$ or uninterpreted functions 
$\term{f}(t_1, \ldots, t_n)$ where the arguments $t_i$ are terms. Constants
are uninterpreted functions with no arguments, and as
usual we'll leave the parentheses off and just write $\term{b}$ or $\term{c}$ instead of
$\term{b}()$ or $\term{c}()$.
We'll often abbreviate sequences 
of terms $t_1,\ldots, t_n$ as
$\overline{t}$ when the indices aren't important.
\end{definition}

\begin{definition}[Facts]\label{def-fact}
A fact has the form 
\begin{math}
\prop{p}{t_1, \ldots, t_n}{v}
\end{math},
where $\pred{p}$ is a \emph{predicate} and the $t_i$ and $v$ are
variable-free (i.e.\ \emph{ground}) terms. We call $\propK{p}{t_1, \ldots, t_n}$ the fact's \textit{attribute}
and call \suggestedchange{$v$}{the term $v$} the fact's \textit{value}. We will sometimes use $a$
to stand in for a variable-free attributes $\propK{p}{\overline{t}}$.
\end{definition}

{\Casping} admits specification of possibility spaces through the interplay of two kinds of rules:

\begin{definition}[Rules]\label{def-rule}
Rules $H \leftarrow F$ have one of two forms, \textit{open} and \textit{closed}.
\begin{align}
\propish{p}{\overline{t}}{v} & \leftarrow F \tag{open form rule}\\
 \prop{p}{\overline{t}}{\{ v_1, \ldots, v_m \}} & \leftarrow F \tag{closed form rule, $m \geq 1$}
\end{align}
In both cases,
$F$ is a conjunction of \textit{premises} of the form $\prop{p}{\overline{t}}{v}$. The rule's \textit{conclusion} (or \textit{head}) $H$ is the part to the left of the 
$\leftarrow$ symbol. Both the head and the premises may contain variables, but every
variable in the head must appear in a premise.
\end{definition}

For intuition: a \textit{closed} conclusion,
     such as $\prop{terrain}{\term{port}}{\{ \term{ocean}, \term{forest} \}}$,
      \textit{requires} that the region named $\term{port}$ be either ocean or forest.
An \textit{open} conclusion, such as
      $\propish{terrain}{\term{goal}}{\term{meadow}}$, requires
      that the attribute $\propK{terrain}{\term{goal}}$ takes some value in the solution,  
      and {\em permits} that value to be meadow.
We make the semantics of these rule forms precise in the next section.

\begin{definition}[Programs]\label{def-program}
A program $P$ is a finite set of rules. 
\end{definition}

\begin{definition}[Substitutions]\label{def-substition}
A substitution $\sigma$ is a total function from variables to ground terms. Applying 
a substitution to a term ($\sigma t$) or a formula ($\sigma F$) replaces all variables $x$
in the term or formula with the term $\sigma(x)$.
\end{definition}

\subsection{Fact-Set Semantics}\label{sec-set-semantics}

In this section, we present the meaning of {\casps} by describing their construction according to a nondeterminstic semantics.

\begin{definition}[Databases]\label{def-set-database}
A \textit{database} $D$ is a set of variable-free facts $\prop{p}{\overline{t}}{v}$ that is
\textit{consistent}, meaning that each attribute $\propK{p}{\overline{t}}$ maps to at most one value $v$:
if $\prop{p}{\overline{t}}{v} \in D$ and $\prop{p}{\overline{t}}{v'} \in D$, then $v = v'$.
\end{definition}

\begin{definition}[Satisfaction] \label{def-set-satisfaction}
We say that a substitution $\sigma$ \textit{satisfies $F$ in the database $D$} when,
for each $\prop{p}{\overline{t}}{v}$ in $F$, the fact
$\prop{p}{\sigma \overline{t}}{\sigma v}$ is present in $D$.
\end{definition}

\begin{definition}[Fact-set evolution]\label{def-set-evolution}
The relation
$D \setstepA{P} S$ relates a database (a consistent set of facts) to a set of databases (a set of consistent sets of facts):
\begin{itemize}
\item If $P$ contains the closed-form rule 
$\prop{p}{\overline{t}}{\{ v_1, \ldots, v_m \}} \leftarrow F$ and
$\sigma$ satisfies $F$ in $D$, then 
$D \setstepA{P} S$, where $S$ is the set of all $D \cup \{ \prop{p}{\sigma\overline{t}}{\sigma v_i} \}$ for $1 \leq i \leq m$ where the result of the union is consistent (and therefore a database).
\item If $P$ contains the open-form rule 
$\propish{p}{t_1, \ldots, t_n}{v} \leftarrow F$ and
$\sigma$ satisfies $F$ in $D$, then 
$D \setstepA{P} S$, where $S$ contains one or two elements. $S$ always contains $D$ itself, and if
$D \cup \{ \prop{p}{\sigma\overline{t}}{\sigma v} \}$ is a consistent set of facts then $S$ contains
$D \cup \{ \prop{p}{\sigma\overline{t}}{\sigma v} \}$ as well.
\end{itemize}
To avoid a corner case with empty programs, we'll also say that $D \setstepA{P} \{ D \}$ always.
\end{definition}

\begin{definition}[Steps]\label{def-set-step-sequence}
We say that the program $P$ allows $D$ to \emph{step} to $D'$ if $D \setstepA{P} S$ and $D' \in S$. Steps gives rise to \textit{step sequences}: $D_1\ldots D_k$ is a step sequence for $P$ if $k \in \mathbb{N}$ (that is, if the sequence is finite) and if, for each $i > 1$, the program $P$ allows $D_{i-1}$ to step to $D_i$.
\end{definition}

\begin{definition}[Saturation]\label{def-set-saturation}
A database $D$ is \textit{saturated under a program $P$} 
if its {\em only possible evolution} under the $\setstepA{P}$
relation is the singleton set containing itself. In other words, $D$ is saturated under $P$
if, for all $S$ such that $D \setstepA{P} S$, it is the case that $S = \{D\}$.
\end{definition}

\begin{definition}[Solutions]\label{def-set-solutions}
A \textit{solution} to the program $P$ is a saturated database
$D$ where 
$\emptyset \ldots D$ is a  step sequence for $P$.
\end{definition}

The implications of this definition of fact-set evolution are a bit subtle. In general, for a given program $P$ and database $D$, there may be many $S$ such that $D \setstepA{P} S$, as many as there are pairs of substitutions $\sigma$ and rules $H \leftarrow F$ such that $\sigma$ satisfies $F$ in $P$. This means that each step in a step sequence resolves two levels of nondeterminism: to take a step from $D$ to $D'$, first it is necessary to pick one of the possibly many $S$ such that $D \setstepA{P} S$, and then it is necessary to pick a database $D'$ from that set $S$.

\subsection{Simulating Datalog}\label{sec-simulating-datalog}

A datalog program without negation is a set of rules (Horn clauses) where all variables in the conclusion appear somewhere in a premise.
\begin{align}
 \propK{p}{\overline{t}} & \leftarrow \propK{p_1}{\overline{t_1}}, \ldots, \propK{p_n}{\overline{t_n}} \tag{datalog rule}
\end{align}
This is a generic use of ``datalog,'' as often 
people take ``Datalog'' to specifically refer to  ``function-free''
logic programs where term constants have no arguments, a condition sufficient to ensure
that every program has a finite model. We follow many theoretical developments and practical implementations of datalog in ignoring the  function-free requirement.

{\Casping} can simulate a datalog proposition $\propK{p}{\overline{t}}$ as a fact of the form $\prop{p}{\overline{t}}{\term{unit}}$, where $\term{unit}$ is a newly introduced constant. This is consistent with the traditional interpretation of datalog so long as the predicate $\pred{p}$ only appears in premises of the form $\prop{p}{\overline{t}}{\term{unit}}$ or in conclusions of the form $\prop{p}{\overline{t}}{\{\term{unit}\}}$. This observation justifies our the use of value-free predicates in {\casps}, which we already saw with the use of $\propK{region}{r}$ in rule~\ref{eqn-dusa-demo1} and $\propK{adjacent}{r_1, r_2}$ in rule~\ref{eqn-dusa-demo2}.

\subsection{Simulating Answer Set Programming}\label{sec-simulating-asp}

\begin{figure}
\begin{align}
    \propishV{p}{\term{ff}} & \leftarrow \label{r01} \\
    \propishV{q}{\term{ff}} & \leftarrow \label{r02} \\
    \propV{p}{\{ \term{tt} \}} & \leftarrow \propV{q}{\term{ff}} \label{r03} \\
    \propV{q}{\{ \term{tt} \}} & \leftarrow \propV{p}{\term{ff}} \label{r04}
\end{align}
    \caption{The {\casp} corresponding to the two rules $\pred{p} \leftarrow \neg \pred{q}$ and $\pred{q} \leftarrow \neg \pred{p}$ in ASP.}
    \label{fig-the-littlest-asp-casp}
\Description{A finite choice logic program. In Dusa's concrete syntax, this would be written:
p is? ff.
q is? ff.
p is \{ tt \} :- q is ff.
q is \{ tt \} :- q is ff. }
\end{figure}

{\Casping} can use the interplay of open and closed rules to obtain all the expressiveness of stable models without any reference to logical negation.
An answer set program containing the two rules $\pred{p} \leftarrow \neg\pred{q}$ and $\pred{q} \leftarrow \neg\pred{p}$ corresponds to the {\casp} in \Cref{fig-the-littlest-asp-casp}. 
The open rules~\ref{r01}~and~\ref{r02} unconditionally \textit{permit} $\pred{p}$ or $\pred{q}$ to have the ``false'' value $\term{ff}$, and the closed rules~\ref{r03}~and~\ref{r04} ensure that the assignment of $\term{ff}$ to either $\pred{p}$ or $\pred{q}$ will force the other attribute to take the ``true'' value $\term{tt}$. (This asymmetry between the handling of truth and falsehood reflects their asymmetric treatment in answer set programming.)
The two solutions for this program are $\{ \propV{p}{\term{tt}}, \propV{q}{\term{ff}} \}$ and  $\{ \propV{p}{\term{ff}}, \propV{q}{\term{tt}} \}$, which correspond to the two solutions that answer set programming assigns to the source program.

Answer set programming is usually defined in terms of variable-free rules that have both non-negated premises $\pred{p_i}$ and negated premises $\neg\pred{q_i}$.
\begin{align}
\pred{p} \leftarrow \pred{p_1}, \ldots, \pred{p_n}, \neg \pred{q_{1}}, \ldots, \neg \pred{q_m} \tag{ASP rule}
\end{align}
A stable model of an answer set program takes the form of a set $X$ of variable-free propositions: any proposition in $X$ is treated as true, and any proposition not in $X$ is treated as false. (For the full definition of stable models, see Gelfond and Lifschitz \shortcite{gel88} or \reforexpanded{sec-simulating-asp-details}{Appendix A}.)

\begin{theorem} \label{thm-sound-complete-asp}
Let $P$ be a finite collection of ASP rules. There exists a translation of $P$ to a {\casp} $\langle P\rangle$ such that the following hold:
\begin{itemize}
\item
For all stable models $X$ of $P$, there is a solution $D$ to $\langle P \rangle$ such that $X = \{ \pred{p} \mid \propV{p}{\term{tt}} \in D \}$.
\item 
For all solutions $D$ of $\langle P\rangle$, the set $\{ \pred{p} \mid \propV{p}{\term{tt}} \in D \}$ is a stable model
of $P$.
\end{itemize}
\end{theorem}
The proof of \Cref{thm-sound-complete-asp} is available in supplemental materials (\reforexpanded{datalog-and-asp-semantics}{Appendix A}). 

\subsection{Example Program Execution}

Let $P$ be the four rule program from \Cref{fig-the-littlest-asp-casp}.
We demonstrate the fact-set semantics step-by-step,
starting from the database
$D_0 = \emptyset$. Two rules have satisfied premises and thus generate evolutions:
\begin{align}
    \textbf{Rules:\;\;} \propishV{p}{\term{ff}} & \leftarrow 
    &\textbf{Evolutions:\;\;} \emptyset & \setstepA{P} \{ \; \emptyset, \; \{ \propV{p}{\term{ff}}
    \} \; \} = S_1
    \notag \\
    \propishV{q}{\term{ff}} & \leftarrow
    &\emptyset & \setstepA{P} \{ \; \emptyset , \;  \{ \propV{q}{\term{ff}}
    \} \; \} = S_2
    \notag \\
    \propV{p}{\{ \term{tt} \}} & \leftarrow \propV{q}{\term{ff}} 
    && \textit{rule does not apply}
    \notag \\
    \propV{q}{\{ \term{tt} \}} & \leftarrow \propV{p}{\term{ff}} 
    && \textit{rule does not apply}\notag
    \intertext{\indent There are three sets $S$ such that $\emptyset \setstepA{P} S$: by rule~\ref{r01} we have $\emptyset \setstepA{P} S_1$, by rule~\ref{r02} we have $\emptyset \setstepA{P} S_2$, and by the trivial evolution $\emptyset \setstepA{P} \{ \emptyset \}$. \endgraf
    It requires only a slight misreading of the definitions to incorrectly conclude that $\emptyset$ is saturated: after all, in all cases that $\emptyset \setstepA{P} S$, it is the case that $\emptyset \in S$. However, $\emptyset$ is \textit{not} saturated by \suggestedchange{\Cref{def-set-satisfaction}}{\Cref{def-set-saturation}}, because $\emptyset$ can step to other databases as well. \endgraf
    If we choose to step from $\emptyset$ to
    $\{\propV{p}{\term{ff}}\}$, three rules apply:}
    \propishV{p}{\term{ff}} & \leftarrow 
    &\{ \propV{p}{\term{ff}} \} & \setstepA{P} \{ \; \{ \propV{p}{\term{ff}} \} \; \}
      = S_3
    \notag \\
    \propishV{q}{\term{ff}} & \leftarrow
    &\{ \propV{p}{\term{ff}} \} & \setstepA{P} \{ \; \{ \propV{p}{\term{ff}} \} , \;  \{
    \propV{p}{\term{ff}}, \propV{q}{\term{ff}} \} \; \} = S_4
    \notag \\
    \propV{p}{\{ \term{tt} \}} & \leftarrow \propV{q}{\term{ff}} 
    && \textit{rule does not apply}
    \notag \\
    \propV{q}{\{ \term{tt} \}} & \leftarrow \propV{p}{\term{ff}} 
    &\{ \propV{p}{\term{ff}} \} & \setstepA{P} \{ \; \{
      \propV{p}{\term{ff}}, \propV{q}{\term{tt}} \} \; \} = S_5 \notag
    \intertext{\indent As before, there are three possible evolutions and one of them is trivial. If
    we step to $\{\propV{p}{\term{ff}}, \propV{q}{\term{ff}}\}$ via $S_4$, we will find ourselves in trouble:}
    \propishV{p}{\term{ff}} & \leftarrow 
    &\{
    \propV{p}{\term{ff}}, \propV{q}{\term{ff}} \} & \setstepA{P} \{ \; \{
    \propV{p}{\term{ff}}, \propV{q}{\term{ff}} \} \; \}
    \notag \\
    \propishV{q}{\term{ff}} & \leftarrow
    &\{
    \propV{p}{\term{ff}}, \propV{q}{\term{ff}} \} & \setstepA{P} \{ \; \{
    \propV{p}{\term{ff}}, \propV{q}{\term{ff}} \} \; \}
    \notag \\
    \propV{p}{\{ \term{tt} \}} & \leftarrow \propV{q}{\term{ff}} 
    &\{
    \propV{p}{\term{ff}}, \propV{q}{\term{ff}} \}& \setstepA{P} \emptyset 
    \notag \\
    \propV{q}{\{ \term{tt} \}} & \leftarrow \propV{p}{\term{ff}} 
    &\{
    \propV{p}{\term{ff}}, \propV{q}{\term{ff}} \}& \setstepA{P} \emptyset \notag
    \intertext{\indent This is yet another case where a slight misreading of the definitions could lead one to incorrectly conclude that $\{
    \propV{p}{\term{ff}}, \propV{q}{\term{ff}} \}$ is saturated: after all, we've just demonstrated that it is a database that can only step to itself. However, $\{
    \propV{p}{\term{ff}}, \propV{q}{\term{ff}} \}$ is not saturated by \Cref{def-set-saturation}: if we look at the last rule, it would seem to require that the program derive $\propV{q}{\term{tt}}$, 
    conflicting with the existing fact $\propV{q}{\term{ff}}$. This conflict means that $\{
    \propV{p}{\term{ff}}, \propV{q}{\term{ff}} \} \setstepA{P} \emptyset$, so the the database is not saturated by \Cref{def-set-saturation}. 
    The database also cannot take a step to any other database aside from itself, and so this represents a failed search: no series of steps will lead to a solution.
    Note that this case demonstrates why we cannot define a saturated database as  ``a database that can only step to itself.'' Such a definition would---undesirably---admit databases that evolve to the empty set.
    \endgraf
    If we back up and instead let $\{ \propV{p}{\term{ff}} \}$ step to $\{\propV{p}{\term{ff}}, 
    \propV{q}{\term{tt}}\}$ via $S_5$, we can observe that the only way the database can evolve is to the singleton set containing itself, which is what \Cref{def-set-saturation} required:}
    \propishV{p}{\term{ff}} & \leftarrow 
    &\{
    \propV{p}{\term{ff}}, \propV{q}{\term{tt}} \} & \setstepA{P} \{ \; \{
    \propV{p}{\term{ff}}, \propV{q}{\term{tt}} \} \; \}
    \notag \\
    \propishV{q}{\term{ff}} & \leftarrow
    &\{
    \propV{p}{\term{ff}}, \propV{q}{\term{tt}} \} & \setstepA{P} \{ \; \{
    \propV{p}{\term{ff}}, \propV{q}{\term{tt}} \} \; \}
    \notag \\
    \propV{p}{\{ \term{tt} \}} & \leftarrow \propV{q}{\term{ff}} 
    && \textit{rule does not apply}
    \notag \\
    \propV{q}{\{ \term{tt} \}} & \leftarrow \propV{p}{\term{ff}} 
    &\{
    \propV{p}{\term{ff}}, \propV{q}{\term{tt}} \}& \setstepA{P} \{ \; \{
    \propV{p}{\term{ff}}, \propV{q}{\term{tt}} \} \; \} \notag
\end{align}
Because $\{\propV{p}{\term{ff}}, 
    \propV{q}{\term{tt}}\}$ is saturated, it is a solution. Symmetric reasoning applies to see that 
$\{\propV{p}{\term{tt}}, \propV{q}{\term{ff}}\}$
is a solution. By inspection, there are no 
solutions where both $\pred{p}$ and $\pred{q}$ are assigned the same value.

\section{{\CASPING} by Example}\label{sec-examples}

This section presents examples to demonstrate common idioms that
arise naturally in writing and reasoning about \casps.

\subsection{Spanning Tree Creation}\label{sec-rooted-spanning-tree}

\begin{figure}
    \begin{align}
    \propK{edge}{x, y} & \leftarrow \propK{edge}{y, x} \label{rspan1}\\
    \propishV{root}{x} & \leftarrow \propK{edge}{x, y} \label{rspan2}\\
    \prop{parent}{x}{\{ x \}} & \leftarrow \propV{root}{x} \label{rspan3}\\
    \propish{parent}{y}{x} & 
        \leftarrow \propK{edge}{x,y}, \prop{parent}{x}{z} \label{rspan4}
    \end{align}
\caption{Calculating a spanning tree over an undirected graph.}\label{fig-spanning-tree-prog}
\Description{A finite choice logic program. In Dusa's concrete syntax, this would be written:
edge X Y :- edge Y X.
root is? X :- edge X Y.
parent X is \{ X \} :- root is X.
parent Y is? X :- edge X Y, parent X is Z.}
\end{figure}

Seeded with
an $\pred{edge}$ relation, the {\casp} in \Cref{fig-spanning-tree-prog} will pick an arbitrary node and 
construct a spanning
tree rooted at that node. 
The structure of this program is such that it's not possible to make forward progress that indirectly
leads to conflicts: rule~\ref{rspan2} can only apply once in a series of deductions, and
rules~\ref{rspan3}~and~\ref{rspan4} cannot fire at all until some root is chosen.
A node can only be added to the tree once, with a parent that already
exists in the tree, so this is effectively a declarative description of Prim's algorithm without weights.

The creation of an arbitrary spanning tree for an undirected graph is a common first
benchmark for datalog extensions that
admit multiple solutions. 
Most previous work makes a selection greedily, either \suggestedchange{discards}{discarding} any future
contradictory selections
\cite{krishnamurthy88nondet,giannotti01nondet,grecozaniolo2001,hu21soufflechoice}, or else 
avoiding contradictory deductions by consuming linear resources \cite{simmonslla}. 

\subsection{Appointing Canonical Representatives}\label{sec-canonical-representatives}

When we want to check whether two nodes in an undirected graph are in the
same connected component,
one option is to compute the transitive closure of the $\pred{edge}$ relation. However, in a sparse
graph, that can require computing $O(n^2)$ facts for a graph
with $n$ edges.

\begin{figure}
    \begin{align}
    \propK{edge}{x, y} & \leftarrow \propK{edge}{y, x} \label{rcanon1}\\
    \propish{representative}{x}{x} & \leftarrow \propK{node}{x} \label{rcanon2}\\
    \prop{representative}{y}{\{ z \}} & \leftarrow \propK{edge}{x, y}, \prop{representative}{x}{z} \label{rcanon3}
    \end{align}
    \caption{Appointing a canonical representative for each connected component in an undirected graph.}
\Description{A finite choice logic program. In Dusa's concrete syntax, this would be written:
edge X Y :- edge Y X.
representative X is? X :- node X.
representative Y is \{ Z \} :- edge X Y, representative X is Z. }
    \label{fig-canonical-rep-prog}
\end{figure}

An alternative is to appoint an arbitrary member of each connected component as the canonical representative
of that connected component:
then, two nodes are in the same connected component if and only if they have the same canonical representative.
This is the purpose of the program in \Cref{fig-canonical-rep-prog}.
In principle, it's quite possible for this program to get stuck in dead ends: if 
nodes $\term{a}$ and
$\term{b}$ are connected by an edge, then rule~\ref{rcanon2} could appoint both 
nodes as a canonical representative, and rule~\ref{rspan3} would then prevent any extension
of that database from being a solution. In the greedy-choice languages mentioned in the previous section, this would be a problem for correctness: incorrectly firing an analogue of rule~\ref{rcanon2} would mean that a final database might contain two canonical representatives in a connected component. In {\casping}, because closed rules can lead to the outright rejection of a database, this is merely a problem of efficiency: we would like to avoid going down these dead ends.

In {\casps} like this one, we can reason about avoiding certain dead ends by assuming a mode of execution that we call \textbf{deduce, then choose}. The deduce-then-choose strategy dictates that, when picking the next step in a step sequence, we will always choose a non-trivial evolution to a singleton ($D \setstepA{P} \{ D' \}$ with $D \neq D'$) over any evolution
to a set containing two or more databases.

Endowed with the deduce-then-choose execution strategy,
the program in \Cref{fig-canonical-rep-prog} will never make deductions that indirectly lead to conflicts. First, rule~\ref{rcanon1}
will ensure that the edge relation is symmetric. Once that is complete, execution will be forced
to choose some canonical representative using rule~\ref{rcanon2} in order to make forward progress.
Once a representative is chosen, rule~\ref{rcanon3} will exhaustively assign that newly-appointed
representative to every other node in the connected component. Only when it is done may rule~\ref{rcanon2} fire again for a node in another
connected component.

\subsection{Lazy Answer Set Programming}\label{lazy-asp-in-casp}

\Cref{thm-sound-complete-asp} only describes a correspondence between variable-free {\casps} and variable-free answer set programs: given that this is how answer set programming is usually formally defined, it was difficult to do otherwise. However, the program transformation underlying \Cref{thm-sound-complete-asp} applies to the non-ground answer set programs that are usually written down in practice, and all the translations we have attempted are faithful translations of the source answer set program. We confidently conjecture that the correctness of the translation extends to non-ground programs, but we leave the formal details for future work.

\begin{figure}
    \begin{align}
    \propK{visit}{\term{z}} & \leftarrow 
    &
    \prop{visit}{\term{z}}{\{\term{tt}\}} & \leftarrow 
    \\
    \propK{visit}{\term{s}(n)} & \leftarrow \propK{more}{n}
    &
    \prop{visit}{\term{s}(n)}{\{ \term{tt} \}} & \leftarrow \prop{more}{n}{\term{tt}}
    \\
    &
    &
    \propish{more}{n}{\term{ff}} & \leftarrow \prop{visit}{n}{\term{tt}}
     \\
    \propK{stop}{n} & \leftarrow \propK{visit}{n}, \neg\propK{more}{n}
    &
    \prop{stop}{n}{\{\term{tt}\}} & \leftarrow \prop{visit}{n}{\term{tt}}, \prop{more}{n}{\term{ff}}
    \\
    &
    &
    \propish{stop}{n}{\term{ff}} & \leftarrow \prop{visit}{n}{\term{tt}} 
    \\
    \propK{more}{n} & \leftarrow \propK{visit}{n}, \neg\propK{stop}{n}
    &
    \prop{more}{n}{\{\term{tt}\}} & \leftarrow \prop{visit}{n}{\term{tt}}, \prop{stop}{n}{\term{ff}} 
    \end{align}
    \caption{At left, an answer set program with no finite grounding. At right, the translation of this program to a {\casp}.}
\Description{An answer set program and the corresponding finite-choice logic program. In Dusa's concrete syntax, the finite-choice logic program would be written:
visit Z is \{ tt \}.
visit (s N) is \{ tt \} :- more N is tt.
more N is? ff :- visit N is tt.
stop N is \{ tt \} :- visit N is tt, more N is ff.
stop N is? ff :- visit N is tt.
more N is \{ tt \} :- visit N is tt, stop N is ff. }
    \label{fig-asp-lasy-casp}
\end{figure}

Translating non-ground answer set programs is interesting in part because of answer set programs like the one in \Cref{fig-asp-lasy-casp}. That program has a well-defined set of solutions, but none of these solutions can be enumerated by mainstream answer set programming implementations because the solver invokes an initial grounding step that is forced to generate an infinite set of ground rules. The translation of this program as a {\casp}, on the other hand, has a simple operational interpretation under the deduce-then-choose strategy: starting from zero, each simple round of deduction will $\pred{visit}$ a successively larger natural number, at which point a choice will be made to either $\pred{stop}$ at that number or to continue to visit $\pred{more}$ numbers. 

This example suggests a connection between {\casping} and the strategy of using \textit{lazy grounding} in answer set programming \cite{gasp}, a strategy which similarly enables the evaluation of answer set programs with no finite grounding \cite{taupe23dshast}. We will return to this point in \Cref{sec-dusa-and-asp-perf}. 

\subsection{Satisfiability}\label{sec-example-boolsat}

The previous examples all show how the evaluation of {\casping} can entirely avoid reaching databases that are not solutions.
However, the full expressive power of {\casping} comes from the ability to represent problems where that avoidance is not possible. Since answer set programming generalizes boolean satisfiability, it is no surprise that boolean satisfiability problems can be represented 
straightforwardly in \casping.

A Boolean satisfiability problem can be written as a 
conjunction of clauses where each clause is
a disjunction of propositions $\pred{p}$
and negated propositions $\neg\pred{p}$. We represent such problems by explicitly assigning each
proposition to $\term{tt}$
or $\term{ff}$ by a closed rule, and adding a rule for
each clause that causes a value conflict for the $\pred{ok}$ predicate
if the clause's negation holds; see \Cref{fig-sat-instance-prog} for an example.
The deduce-then-choose execution strategy gives no advantages here: the only deduction is observing inconsistencies that result from already-selected choices.

\begin{figure}
\begin{align}
\propV{p}{\{ \term{tt}, \term{ff} \}} & \leftarrow \\
\propV{q}{\{ \term{tt}, \term{ff} \}} & \leftarrow \\
\propV{r}{\{ \term{tt}, \term{ff} \}} & \leftarrow \\
\propV{ok}{\{ \term{yes} \}} & \leftarrow \\
\propV{ok}{\{ \term{no} \}} & \leftarrow
\propV{p}{\term{ff}}, \propV{q}{\term{tt}} \\
\propV{ok}{\{ \term{no} \}} & \leftarrow
\propV{p}{\term{tt}}, \propV{q}{\term{ff}}, \propV{r}{\term{ff}} 
\end{align}
    \caption{A {\casp} representing the SAT instance $(p \vee \neg q) \wedge (\neg p \vee q \vee r)$.}
\Description{A finite choice logic program. In Dusa's concrete syntax, this would be written:
p is \{ tt, ff \}.
q is \{ tt, ff \}.
r is \{ tt, ff \}.
ok is \{ yes \}.
ok is \{ no \} :- p is ff, q is tt.
ok is \{ no \} :- p is tt, q is ff, r is ff. }
    \label{fig-sat-instance-prog}
\end{figure}

\section{Nondeterministic Immediate Consequences}\label{sec-immcons}

The semantics given in \Cref{sec-definition} allow us to interpret the
meaning of a {\casp} as the set of the program's solutions
(\Cref{def-set-solutions}). In some ways, though, this semantics leaves a lot to be desired.
Consider the following program $P$:
\begin{align}
\propV{p}{\{ \term{a}, \term{b} \}} & \leftarrow \label{r1} \\
\propV{p}{\{ \term{b}, \term{c} \}} & \leftarrow \label{r0}\\
\propishV{q}{\term{ff}} & \leftarrow \label{r2} \\
\propV{q}{\{ \term{tt} \}} & \leftarrow \propV{p}{x} \label{r3}
\end{align}
Under the semantics in \Cref{sec-definition}, 
there is a two-step sequence from $\emptyset$ to $\{ \propV{p}{\term{c}, \propV{q}{\term{ff}}} \}$:
\begin{align}
\emptyset & \setstepA{P} \{ \; \{ \propV{p}{\term{b}} \}, \; \underline{\{ \propV{p}{\term{c}} \}} \; \} \tag{by~rule~\ref{r0}}\\
\{ \propV{p}{\term{c}} \} & \setstepA{P}
\{ \; \{ \propV{p}{\term{c}} \}, \;  \underline{\{ \propV{q}{\term{ff}}, \propV{p}{\term{c}} \}}  \; \}
\tag{by~rule~\ref{r2}}
\end{align}
The result is not a solution, cannot be extended to a solution, and also cannot step to any other database. The
first step led us, irrevocably, to a dead end: the step sequence $\emptyset, \{ \propV{p}{\term{c}} \}, \ldots $ can never be extended to reach a solution.

Some dead ends are unavoidable when conflicting assignments only
occur down significant chains of deduction. That possibility is part of what
gives {\casping} its expressive power! In the program above, though,
the
conflict is in some sense immediate: at each step, we have enough
information to know that rule~\ref{r1} and rule~\ref{r0} both apply, and
the overlap of these closed rules means that $\pred{p}$ can only be given
the value $\term{b}$ in any solution.

In this section, we will revisit the semantics of {\casping} to present a
\textit{immediate consequence} operator that captures global information
about what is ``immediately derivable'' from a given program and database.
This development will provide the basis for a least-fixed-point semantics
of {\casping} (\Cref{sec-semantics}), as well as the foundation for our implementation of {\casping} (\Cref{sec-implementation}). 

\subsection{Bounded-Complete Posets}

This development introduces several new concepts: \textit{constraints}, \textit{constraint databases}, and \textit{choice sets}. All these are instances of the same semilattice-like structure: bounded-complete posets.

In \Cref{sec-definition}, we define databases as sets with an auxiliary
definition of consistency.  Now we
will generalize consistency to a notion of {\em compatibility:}

\begin{definition}[Compatibility]
  If $\Domain$ is a set equipped with a partial order $\le$, then a subset $X \subseteq \Domain$ is \emph{compatible} when it has an upper bound, i.e.\ $\exists y \in \Domain.\ \forall x \in X.\ x \le y$.
  We write ${\compatible} X$ to assert that $X$ is compatible and ${\incompatible}X$ for its negation.
  As a binary operator $x \compatible y = {\compatible}\{x,y\}$ and $x \incompatible y = {\incompatible}\{x, y\}$.
\end{definition}

We will frequently use this basic fact about (in)compatibility:
\begin{lemma}\label{lemma-incompatibility-monotone}
  Compatibility is anti-monotone and incompatibility is monotone: if $D \le D'$ and $E \le E'$, then $D' \compatible E' \implies D \compatible E$, and contrapositively $D \incompatible E \implies D' \incompatible E'$.
\end{lemma}

\begin{proof}
  $D' \compatible E'$ means $D', E'$ have some upper bound $D^*$; if $D \le D'$ and $E \le E'$, then $D^*$ is also an upper bound for $D, E$.
\end{proof}

\begin{definition}\label{def-domain}
  A \emph{bounded-complete poset} $(\Domain, \le_\Domain, \bot_\Domain, {\bigvee}_\Domain)$ is a poset with a least element where all compatible subsets have least upper bounds. In detail:
  \begin{enumerate}
  \item ${\le}_\Domain \subseteq \Domain \times \Domain$ is a partial order (a reflexive, transitive, antisymmetric relation).
  \item $\bot_\Domain \in \Domain$ is the least element: $\forall x \in \Domain.\ \bot \le x$.
  \item ${\bigvee}_\Domain : \{ X \subseteq \Domain : {\compatible}X \} \rightarrow \Domain$ finds the least upper bound of a compatible set of elements:
  $(\forall x \in X.\ x \le z) \implies \bigvee_\Domain X \le z$.
  As a binary operator, $x \vee_\Domain y = \bigvee_\Domain \{ x, y \}$.
  \end{enumerate}
\end{definition}

\begin{example}
    Databases as presented in \Cref{sec-set-semantics} are bounded-complete posets, where $D_1 \leq D_2$ is the subset relation and $\bot$ is the empty set $\emptyset$. The least upper bound operation is set union: if a set of databases are each individually subsets of some consistent set $D$ of facts, then their union is a subset of $D$ and so must also be a consistent set of facts.
\end{example}

\subsection{Constraints and Constraint Databases}\label{sec-constraints}

\begin{definition}
  A \emph{constraint,} $c \in \Constraint$, is either $(\just{t})$ for some ground term $t$ or $(\none{X})$ for some set $X$ of ground terms.
  Constraints form a bounded-complete poset as follows:

  \begin{enumerate}
  \item ${\le}_\Constraint$ is defined by cases:
    \begin{align*}
      \none{X} \le \none{Y} &\iff X \subseteq Y\\
      \none{X} \le \just{t} &\iff t \notin X\\
      \just{t} \le \just{t'} &\iff t = t'\\
      \just{t} \not\le \none{X} &
    \end{align*}

  \item $\bot_\Constraint = \none{\emptyset}$.

  \item $\bigvee_\Constraint\,C$ is defined by cases.
    Because the least upper bound is only defined for compatible sets, if we know $\just{t} \in C$ then $\just{t}$ is the least upper bound:
    any other upper bound must have the form $\just{t'}$ with $t = t'$.
    Otherwise, every $c_i \in C$ is of the form $\none{X_i}$, and their least upper bound is $\none{\left(\bigcup_i X_i\right)}$. By way of illustration, in the binary case:
    \begin{align*}
      \none{X} \vee \none{Y} &= \none{\left(X \cup Y\right)}
      \\
      \none{X} \vee \just{t} &= \just{t} \qquad\text{if }t \notin X
      \\
      \just{t} \vee \none{X} &= \just{t} \qquad\text{if }t \notin X
      \\
      \just{t} \vee \just{t'} &= \just{t} \qquad\text{if }t = t'
    \end{align*}

    \noindent
    If none of these cases apply, $c_1 \vee c_2$ is undefined because $c_1 \incompatible c_2$.

  \end{enumerate}
\end{definition}

\begin{definition}\label{def-db-database}
  A \emph{constraint database,} $D,E \in \DB$, is a map from ground attributes $a$ to constraints $D[a]$. Constraint databases form a bounded-complete poset with structure inherited pointwise from \Constraint\emph{:}

  \begin{enumerate}
  \item $D \le_{\DB} E \iff \forall a.\ D[a] \le E[a]$.
  \item $\bot_\DB$ is mapping that takes every attribute $a$ to $\bot_\Constraint = \none{\emptyset}$.
  \item $\bigvee_\DB\,S = a \mapsto \bigvee_\Constraint\, \{ D[a] : D \in S \}$, and so is
  defined whenever $\forall a.\ {\compatible}\{ D[a] : D \in S \}$.
  \end{enumerate}
\end{definition}

\begin{definition}\label{def-db-positive}
A constraint database $D$ is \emph{positive} if $D[a] = \none{X}$ implies $X = \emptyset$.
\end{definition}

\begin{definition}\label{def-db-finite}

    A constraint database $D$ is \emph{finite} if $\{a : D[a] \ne \none{\emptyset}\}$ is finite and all $X$ such that $D[a] = \none{X}$ are finite.
\end{definition}

We will introduce a new notation to describe finite constraint databases:
$(
\pred{p} \mapsto \just{\term{ff}},\;
\pred{q} \mapsto \none{\{ \term{ff} \}}
)$ represents the constraint database that takes 
$\pred{p}$ to $\just{\term{ff}}$,
takes $\pred{q}$ to $\none{ \{ \term{ff} \}}$, and takes every other 
attribute to $\bot_\Constraint = \none{\emptyset}$.

There is an obvious isomorphism between positive constraint databases and the databases-as-consistent-sets-of-facts as introduced in \Cref{def-set-database}. The consistent sets of facts $\emptyset$, \suggestedchangerationale{$\{ \propV{p}{\term{tt}}, \propV{q}{\term{tt}} \}$}{$\{ \propV{p}{\term{tt}}, \propV{q}{\term{ff}} \}$}{typo fix}, and $\{ \prop{edge}{\term{a},\term{b}}{\term{unit}} \}$ correspond, respectively, to the constraint databases $\bot_\DB$, $(\pred{p} \mapsto \just{\term{tt}}, \, \pred{q} \mapsto \just{\term{ff}})$, and $(\propK{edge}{\term{a}, \term{b}} \mapsto \just{\term{unit}})$.

\subsection{Why Constraint Databases Aren't Enough}\label{sec-choice-sets-needed}

The move from ``sets of facts'' to ``functions from attributes to constraints'' is an instance of a common pattern encountered in extensions of logic programming to non-Boolean values.  In the bilattice-annotated logic programming setting, Fitting \shortcite{FITTING1993197} calls the analogue of a constraint database a \textit{valuation}, and Komendantskaya and Seda \shortcite{KOMENDANTSKAYA2009141} call the analogue of a constraint database an \textit{annotation Herbrand model}. In the weighted logic programming setting, Eisner \shortcite{10.1162/tacl_a_00588} maintains a similar map $\omega$ from items to weights.

To our knowledge, all of this related work proceeds to define an immediate
consequence operator as a function from database-analogues to
database-analogues, and the meaning of a program is given as the unique
least fixed point of this operator. But in {\casping}, as in answer set
programming, programs can have multiple incompatible solutions, and so the
meaning of a program \emph{cannot} be a singular least fixed point of a
function from constraint databases to constraint databases. One way forward
is to present \suggestedchange{a function where solutions correspond to \textit{any} fixed point, or any fixed point satisfying}{an operator where solutions correspond to \textit{any} fixed point, or correspond to any fixed point of the operator that satisfies} some additional condition: this
is essentially the approach used by Fitting~\shortcite{FITTING1993197} to
bound the set of stable models for an answer set program using bilattices.
However, this formulation doesn't directly characterize 
the stable models: it is still necessary to carry out the Gelfond-Lifschitz
program transformation as a post-hoc check to see if one of Fitting's bounded
fixed points is a stable model. Our
attempts to precisely characterize solutions for {\casps} as an arbitrary
fixed point of some function from databases to databases encountered
similar difficulties.

In this work, we take a different approach, which we believe is novel: in \Cref{sec-semantics},
the least-fixed-point interpretation of {\casping} is defined to be a \emph{set of pairwise incompatible constraint databases},
which we call a {\em choice set}. \suggestedchangerationale{Choice sets are therefore the next step towards defining our immediate consequence operator.}{}{tightening up prose}

\subsection{Choice Sets}\label{sec-immcons-choice-sets}

Choice sets form not merely a bounded-complete poset, but a complete
lattice that has all least upper bounds. We will establish this in two
steps, first defining choice sets as a pointed partial order, and then
defining least upper bounds.

\begin{definition}\label{def-choice}
  A \emph{choice set} $\C \in \Choice$, is a \emph{pairwise-incompatible} set of constraint databases, meaning that $(\forall D, E \in \C.\ D \compatible E \implies D = E)$. $\Choice$ is a pointed partial order:

  \begin{enumerate}
  \item $\C_1 \le_\Choice \C_2 \iff (\forall D_2 \of \C_2.\ \exists D_1 \of \C_1.\ D_1 \le D_2) \iff (\forall D_2 \of \C_2.\ \exists! D_1 \of \C_1.\ D_1 \le D_2)$.
    Existence $\exists D_1$ implies \emph{unique} existence $\exists! D_1$ by pairwise incompatibility: if $D_1, D_1' \in \C_1$ are both $\le D_2$ they are compatible and thus equal.
  \item $\bot_\Choice = \{\bot_\DB\}$, that is, the set containing one item, the constraint database that maps every attribute to $\none{\emptyset}$.
  \item $\top_\Choice = \emptyset$.

  \end{enumerate}
\end{definition}

In a constraint database $D$, each rule and rule-satisfying substitution in a {\casp} will induce a choice set $\C_i$, and deriving the immediate consequences of $D$ involves a ``parallel composition'' of all these choices. The least upper bound $\bigvee_i \C_i$ calculates this parallel composition. We will start with a finite example to build intuition:

\begin{example}
    If we have $\C_1 = \{ D_a, D_b, D_c \}$ and $\C_2 = \{ D_x, D_y \}$, then $\C_1 \vee C_2$ contains between zero and six elements: all of the least upper bounds out of $(D_a \vee D_x)$, $(D_b \vee D_x)$, $(D_c \vee D_x)$, $(D_a \vee D_y)$, $(D_b \vee D_y)$, and $(D_c \vee D_y)$ that are actually defined.
\end{example}

For the least upper bound of $n$ choice sets, a database is in the least upper bound exactly when it is the least upper bound of $n$ compatible databases, one drawn from each of the choice sets. Formally, we define the least upper bound of a collection $\C_{i \in I}$ of choice sets indexed by some set $I$.
Let $f : I \to \DB$ be a function choosing one database from each choice set, so that $f(i) \in \C_i$.
If the set of databases thus chosen, $\mathrm{Im}(f) = \{ f(i) : i \in I \}$, is compatible, we include its least upper bound $\bigvee_i f(i)$ in the resulting choice set $\bigvee_i \C_i$.
The set of chosen databases, $\mathrm{Im}(f) = \{ f(i) : i \in I \}$, can be seen as a \emph{candidate set} for the least upper bound.
If $\mathrm{Im}(f)$ is compatible, the least upper bound of $\mathrm{Im}(f)$ exists and we include its least upper bound in the resulting choice set $\bigvee_i \C_i$.

\begin{definition}[Least upper bounds for $\Choice$]\label{def-choice-lub}
   Take any $\{ \C_i : i \in I \} \subseteq \Choice$, and let $\prod_{i \in I} \C_i$ be the set of all functions $f : I \rightarrow \DB$ such that $f(i) \in \C_i$.
   Then:
   \[
   \bigvee_{i \in I} \C_i
   =
   \left\{
   \bigvee \mathrm{Im}(f) \,:\,
   f \in {\prod_{i \in I} \C_i},\
   {\compatible} \mathrm{Im}(f)
   \right\}
   \]
\end{definition}

It is not entirely trivial to show that $\bigvee_i \C_i$ is a least upper bound in \emph{\Choice:} the proof is available in the supplemental material (\reforexpanded{sec-choice-has-lubs}{Appendix B}).  The main subtleties are ensuring that we avoid combining incompatible databases and ensuring that the resulting set remains pairwise incompatible.

The partial order $\leq_\Choice$, unlike the partial orders on \Constraint\ and \DB, has a greatest element ($\emptyset$), and so any collection of choice sets has an upper bound. Therefore, {\Choice} is a complete lattice, not merely a bounded-complete poset.

\subsection{Immediate Consequence}\label{sec-immediate-consequence-def}

Almost everything is now in place define an immediate consequence operator $\immcons{P} : \DB \rightarrow \Choice$. We only need to replay the definition of satisfaction from \Cref{sec-set-semantics} in terms of constraint databases and choice sets, and then we can define the immediate consequence as the least upper bound of the consequence of every rule that can fire.

\begin{definition}[Satisfaction]\label{def-db-satisfaction}
    We say that a substitution $\sigma$ \textit{satisfies $F$ in the constraint database $D$}
    when, for each premise $\prop{p}{\overline{t}}{v}$ in $F$, we have
    $(\just{\sigma{v}}) \leq D[\propK{p}{\sigma\overline{{t}}}]$.
\end{definition}

\begin{definition}\label{def-ruleconc}
    A ground rule conclusion {\ruleconc} defines a element of \Choice,
    which we write as $\rulechoices{\ruleconc}$, in the following way:
    \begin{itemize}
    \item $\rulechoices{\propish{p}{\overline{t}}{v}} = \{ \; ( \propK{p}{\overline{t}} \mapsto \just{v}),\;  (\propK{p}{\overline{t}} \mapsto \none{\{ v \} }) \; \}$
    \item $\rulechoices{\prop{p}{\overline{t}}{\{ v_1,  \ldots, v_n \}}} = \{ \; ( \propK{p}{\overline{t}} \mapsto \just{v_1}), \; \ldots , \; (\propK{p}{\overline{t}} \mapsto \just{v_n}) \;\}$
    \end{itemize}
\end{definition}

\begin{definition}[Immediate consequence]
\label{def-immediate-consequence}
The immediate consequence operator $\immcons{P} : \DB \to \Choice$
is the least upper bound of every head of a rule with satisfied premises:
\[\immcons{P}(D) = \{ D \} \vee \left(\bigvee \left\{ \langle \sigma H \rangle : (H \leftarrow F) \in P, \sigma F \leq {D} \right\}\right)\]
We forcibly ensure that $\{ D \} \leq \immcons{P}(D)$ by making the result the least upper bound of $\{ D \}$ itself and the least upper bound of all the satisfied rule heads. (We conjecture this is redundant for our subsequent developments, but it is also harmless.)
\end{definition}

\begin{example}
    Returning to the four rule program $P$ 
at the beginning of this section (rules~\ref{r1}-\ref{r3}), these
are all examples of how the immediate consequence operator for that program
behaves on different inputs:
\begin{align}
\immcons{P}(\bot_\DB) & = 
    \{ \; 
        ( \pred{p} \mapsto \just{\term{b}}, \; \pred{q} \mapsto \just {\term{ff}}), \; \notag 
\\ & \qquad
        ( \pred{p} \mapsto \just{\term{b}}, \; \pred{q} \mapsto \none{\{ \term{ff} \} })
    \;\}\notag
\\
\immcons{P}(\pred{p} \mapsto \just{\term{b}},\;\pred{q} \mapsto \just {\term{ff}}) &= \emptyset\notag
    \\
\immcons{P}(\pred{p} \mapsto \just{\term{b}},\; \pred{q} \mapsto \none{\{ \term{ff} \} } ) & = 
    \{ \; 
        ( \pred{p} \mapsto \just{\term{b}}, \;\pred{q} \mapsto \just{\term{tt} } )
    \;\}\notag
\\
\immcons{P}(\pred{q} \mapsto \just {\term{ff}}) & = 
    \{ \; 
        (\pred{p} \mapsto \just{\term{b}},\; \pred{q} \mapsto \just{\term{ff}} )
    \;\}\notag
\\
\immcons{P}(\pred{p} \mapsto \just {\term{b}}) &= 
    \{ \; 
        (\pred{p} \mapsto \just{\term{b}}, \;\pred{q} \mapsto \just{\term{tt} } )
    \;\}\notag
    \\
\immcons{P}(\pred{p} \mapsto \just{\term{b}}, \pred{q} \mapsto \just {\term{tt}}) & = 
    \{ \; 
        ( \pred{p}\mapsto \just{\term{b}},\; \pred{q} \mapsto \just{ \term{tt} } )
    \;\}\notag
\\
\immcons{P}(\pred{q} \mapsto \just{ \term{someOtherTerm} } ) & = 
    \{ \; 
        ( \pred{p} \mapsto \just{\term{b}},\; \pred{q} \mapsto \just{ \term{someOtherTerm} } )
    \;\}\notag
\end{align}

The immediate consequence operator on $\bot_\DB$ captures the fact that that, due to rules~\ref{r1}~and~\ref{r0} in combination, $\pred{p}$ can only be assigned the value $\term{b}$. Any other step, like a step to $\propV{p}{\term{a}}$ by rule~\ref{r1}, is a dead end. The immediate consequence operator does not, however, preclude the dead-end
step from an empty database to a database where
$\propV{q}{\term{ff}}$, since that information is not immediately available: rule~\ref{r3}
only applies in a database where $\pred{p}$ has a definite value.
\end{example}

\section{Fixed-Point Semantics}\label{sec-semantics}

We have one notion of what a {\casp} means: the meaning of a program $P$ is the set of databases that are solutions to $P$ according to \Cref{def-set-solutions}. So why do need to provide another explanation for what a program means? 

In our view, there are three satisfying ways of explaining the foundations of a logic programming language. In no particular order:
\begin{itemize}
    \item A logic program defines a set of \suggestedchange{derivable propositions}{propositions} in a constructive proof theory\suggestedchange{}{, and the meaning of a program is given by the logical consequences of those propositions}. The methodology of \textit{uniform proofs} brings with it a built-in operational semantics in terms of proof search in a focused sequent calculus \cite{miller91uniform}.
    \item A logic program defines a monotonic \textit{immediate consequence} function, and the meaning of a logic program is the least fixed point of this function. This methodology brings with it a built-in operational semantics, because one can repeatedly apply \suggestedchange{immediate}{the immediate} consequence operation to a least element in hopes of reaching a fixed point \cite{gelder-ross-schlipf91wellfounded}.
    \item A logic program \suggestedchange{is a proposition in classical logic obtained from a program}{defines a proposition in classical logic} through the program completion of Clark \shortcite{clark78negation}\suggestedchange{}{, and the meaning of a logic program is the set of satisfying models of that proposition}. This interpretation admits spurious circular justifications: the program containing \suggestedchangerationale{one rule $\pred{p} \leftarrow \pred{q}$ has the Clark completion $\pred{p} \leftrightarrow \pred{q}$, which admits both the 
expected interpretation (both propositions are false) as well as an undesirable interpretation (both propositions are true)}{one rule $\pred{p} \leftarrow \pred{p}$ has the Clark completion $\pred{p} \leftrightarrow \pred{p}$, which admits both the expected interpretation ($\pred{p}$ is false) as well as an undesirable interpretation ($\pred{p}$ is true)}{the previous version was an error: $p \leftarrow q$ will have the Clark completion $(p \leftrightarrow q) \wedge (q \leftrightarrow \bot)$}. Some programs, however, have a unique \emph{least} model --- the smallest set of true propositions --- which can be interpreted as the canonical solution.
\end{itemize}
A very nice property of datalog is that it admits all three justifications, and they all agree. Answer set programming provides a greatly desirable property --- the ability to give a program multiple, mutually-exclusive solutions --- but it does so at the cost of being able to give programs any of these foundationally satisfying semantics.

In this section, we will present the interpretation of a {\casp} as a member of $\Choice$, the least fixed point of a ``lifted'' immediate consequence operator that has choice sets as both its domain and range. \Cref{semantics-are-equivalent} establishes that solutions (\Cref{def-set-solutions}) correspond to databases in this least fixed point that are positive (\Cref{def-db-positive}) and finite (\Cref{def-db-finite}).

\begin{example}
First, let's develop some intuition about what it means for a choice set to provide the interpretation of a {\casp}. In this example only, allow yourself to squint and reinterpret consistent sets of facts as constraint databases: in this hazy light, the set of solutions for a given program looks a lot like a choice set. 

The program with no rules corresponds to the choice set
$\bot_\Choice = \{ \emptyset \}$. 

The program with one rule $(\propV{p}{ \{ \term{tt}, \term{ff} \} } \leftarrow)$
has two pairwise-incompatible solutions that form the
choice set $\{ \; \{ \propV{p}{\term{tt}} \}, \; \{ \propV{p}{\term{ff}} \} \; \}$.
\begin{align*}
\bot_\Choice& \leq_\Choice \{ \; \{ \propV{p}{\term{tt}} \}, \; \{ \propV{p}{\term{ff}} \} \; \}
\\
\intertext{\indent One way additional rules can create greater choice sets is by adding facts. A second rule $(\propV{q}{\{ \term{tt} \}} \leftarrow)$ results in a program
that still has two solutions.}
 \{ \; \{ \propV{p}{\term{tt}} \}, \; \{ \propV{p}{\term{ff}} \} \; \}
 & \leq_\Choice
 \{ \; \{ \propV{p}{\term{tt}}, \propV{q}{\term{tt}} \}, \; \{ \propV{p}{\term{ff}}, \propV{q}{\term{tt}} \} \; \}
\intertext{If we instead added a more constrained
second rule $(\propV{q}{\{\term{tt}\} \leftarrow \propV{p}{\term{tt}}})$, we would
instead have these solutions:}
 \{ \; \{ \propV{p}{\term{tt}} \}, \; \{ \propV{p}{\term{ff}} \} \; \}
 & \leq_\Choice
 \{ \; \{ \propV{p}{\term{tt}}, \propV{q}{\term{tt}} \}, \; \{ \propV{p}{\term{ff}} \} \; \}  \; \}
\intertext{\indent Another way additional rules can create greater choice sets is by \emph{removing solutions}. If we return to the program with just the rule $(\propV{p}{ \{ \term{tt}, \term{ff} \} } \leftarrow)$ and add a second rule $(\propV{p}{\{\term{tt}\}} \leftarrow \propV{{p}}{x})$, the database where $\propV{p}{\term{ff}}$ is no longer a solution:}
  \{ \; \{ \propV{p}{\term{tt}} \}, \; \{ \propV{p}{\term{ff}} \} \; \}
 & \leq_\Choice
 \{ \; \{ \propV{p}{\term{tt}} \} \; \}
 \intertext{\indent A choice set that is greater according to $\leq_\Choice$
 may alternatively contain \emph{more} constraint databases,
 which would occur if the second rule was instead $(\propV{q}{\{\term{tt}, \term{ff}\}} \leftarrow \propV{p}{\term{ff}})$:}
 \{ \; \{ \propV{p}{\term{tt}} \}, \; \{ \propV{p}{\term{ff}} \} \; \}
 & \leq_\Choice
 \{ \; \{ \propV{p}{\term{tt}} \}, \; \{ \propV{p}{\term{ff}}, \propV{q}{\term{tt}} \}, \; \{ \propV{p}{\term{ff}}, \propV{q}{\term{ff}} \} \; \}
\intertext{\indent Finally, if we have a fully contradictory program, such as one containing two rules $(\propV{p}{ \{ \term{tt}, \term{ff} \} } \leftarrow)$ and $(\propV{p}{ \{ \term{meadow} \} } \leftarrow \propV{p}{x})$, its interpretation is the set containing zero solutions. This is a valid choice set, and is in fact the greatest element of {\Choice}.}
 \{ \; \{ \propV{p}{\term{tt}} \}, \; \{ \propV{p}{\term{ff}} \} \; \}
 & \leq_\Choice
 \emptyset = \top_\Choice
\end{align*}

Running through this example, one can observe a kind of monotonicity at work: adding rules to a program always results in the meaning of that larger program being a greater choice set according to the partial order $\leq_\Choice$ from \Cref{def-choice}. 
\end{example}

\subsection{Lifted Immediate Consequence}
\label{sec-fixed-point-semantics}

The typical move when defining the meaning of a forward-chaining logic
programs is to interpret the program as the least fixed point of an immediate
consequence function. Immediate consequence presented in \Cref{sec-immediate-consequence-def} is a function $\immcons{P} : \DB \to \Choice$, and because the domain and range are different we can't take the fixed point. Recalling the definition of saturation in \Cref{def-set-saturation} gives us a limited (but important!) notion of fixed points that we call \emph{models}:

\begin{definition}\label{def-model}
  A constraint database $D$ is a \emph{model} of the program $P$ in either of these equivalent conditions:
  \[ D \in \immcons{P}(D) \iff \immcons{P}(D) = \{D\} \]
\end{definition}

We have no hope of giving a least-fixed-point interpretation of {\casps} this way, as most interesting {\casps} do not have unique minimal models. Instead, we will ``lift'' $\immcons{P}$ to a function $\bigimmcons{P} : \Choice \to \Choice$\emph{:}

\begin{definition}\label{def-lifted-immcons}
    $\bigimmcons{P}(\C) = 
    \bigcup_{D \in \C} \immcons{P}(D)$.
\end{definition}

For \Cref{def-lifted-immcons} to be a candidate for iterating to a fixed point, we must show that $\bigcup_{D \in \C} \immcons{P}(D)$ is in \Choice, i.e.\ that all the databases it contains are pairwise incompatible:

\begin{lemma}\label{lem-tp-star-in-choice}
   If $\C$ is pairwise incompatible, then so is $\bigcup_{D \in \C} \immcons{P}(D)$. That is,
   if $E_1,E_2\in \bigcup_{D \in \C} \immcons{P}(D)$ are compatible, then they are equal.
\end{lemma}
\begin{proof}
    Consider some $D_1,D_2 \in \C$ and $E_1 \in \immcons{P}(D_1)$ and $E_2 \in \immcons{P}(D_2)$.
    Suppose $E_1 \compatible E_2$. Since $D_1 \le E_1$ and $D_2 \le E_2$, by \Cref{lemma-incompatibility-monotone} we know $D_1 \compatible D_2$, thus $D_1 = D_2$ (since both are in $\C$). And since $\immcons{P}(D_1) = \immcons{P}(D_2)$ is pairwise incompatible, $E_1 = E_2$.
\end{proof}

Now we would like to define the meaning of a {\casp} as the unique least fixed point of $\bigimmcons{P}$.
Because choice sets form a complete lattice, we can  apply Knaster-Tarski~\cite{tarski1955lattice} as long as $\bigimmcons{P}$ is monotone, which it is:



\begin{lemma}[$\immcons{P}$ is monotone]\label{thm-immcons-monotone}
  If $D_1 \leq_{\DB} D_2$,
  then $\immcons{P}(D_1) \leq_{\Choice} \immcons{P}(D_2)$.
\end{lemma}

\begin{proof}
     Because $\immcons{P}(D_1)$ and $\immcons{P}(D_2)$ are the least upper bound of the rule heads satisfied in $D_1$ and $D_2$ respectively, it suffices to show that if $\sigma$ satisfies $F$ in $D_1$, then $\sigma$ satisfies $F$ in $D_2$.
     This follows from \Cref{def-db-satisfaction}: for each premise $\prop{p}{\overline{t}}{v}$ in $F$, we have $\just{\sigma{v}} \le D_1[\propK{p}{\sigma\overline{t}}] \le D_2[\propK{p}{\sigma\overline{t}}]$.
\end{proof}

\begin{lemma}[$\bigimmcons{P}$ is monotone]\label{thm-lifted-monotone}
  If $\C_1 \leq \C_2$, then $\bigimmcons{P}(\C_1) \leq \bigimmcons{P}(\C_2)$.
\end{lemma}

\begin{proof}
      We wish to show $\bigcup_{D \in \C} \immcons{P}(D) \le_\Choice \bigcup_{D' \in \C'} \immcons{P}(D')$.
  So, fixing $D' \in \C'$ and $E' \in \immcons{P}(D')$, we wish to find a $D \in \C$ and $E \in \immcons{P}(D)$ with $E \le E'$.
  Since $\C \le \C'$, for $D' \in \C'$ there exists a unique $D \in \C$ with $D \le D'$.
  By monotonicity of $\immcons{P}$ we have $\immcons{P}(D) \le_\Choice \immcons{P}(D')$.
  Thus for $E' \in \immcons{P}(D')$ we have a unique $E \in \immcons{P}(D)$ with $E \le E'$.
\end{proof}

\begin{theorem}[Least fixed points]\label{thm-least-fixed-points}
  $\bigimmcons{P}$ has a least fixed point, written as $\lfp \bigimmcons{P}$.
\end{theorem}

\begin{proof}
  Because $\bigimmcons{P}$ is monotone and $\Choice$ is a complete lattice,
  by Tarski~\shortcite{tarski1955lattice}, the set of fixed points
  of $\bigimmcons{P}$ forms a complete lattice. The least fixed point
  is the least element of this lattice, i.e.\
  $\lfp \bigimmcons{P} = \bigwedge \{\C : \bigimmcons{P}(\C) \leq \C\}$
  (where $\bigwedge X = \bigvee \{\C : \forall x \in X, \C \leq x\}$).
\end{proof}

\begin{corollary} \label{corollary-covering}
    Every model $E$ has a lower bound $D \in \lfp\bigimmcons{P}$ with $D \le E$.
\end{corollary}

\begin{proof}
    Since $\{E\}$ is a fixed point of $\bigimmcons{P}$, we know $\lfp \bigimmcons{P} \le_\Choice \{E\}$, i.e.\ $\exists D \in \lfp \bigimmcons{P}.\ D \le E$.
\end{proof}

\begin{theorem}[Minimal models] \label{thm-exactly-the-minimal-models}
  $\lfp\bigimmcons{P}$ contains exactly the minimal models of $P$, meaning models $D$ such that any model $D' \le D$ is equal to $D$.
\end{theorem}

\begin{proof}
  No model outside $\lfp\bigimmcons{P}$ is minimal, because by \Cref{corollary-covering} it has a lower bound in $\lfp\bigimmcons{P}$. And every model $D \in \lfp\bigimmcons{P}$ is minimal: given a model $E \le D$, by \Cref{corollary-covering}, there is some $D' \in \lfp\bigimmcons{P}$ with $D' \le E \le D$. But then $D' \compatible D$ and by pairwise incompatibility $D' = D = E$.
\end{proof}

\subsection{Agreement of Step Semantics and Least-Fixed-Point Semantics}

\Cref{thm-least-fixed-points} establishes that it is reasonable to say that a {\casp} has a canonical model, $\lfp\bigimmcons{P}$.
Defining the meaning of a program this way, instead of the step-by-step definition given in \Cref{sec-set-semantics}, has a number of advantages. However, we need to take care with stating the correspondence between our two ways of assigning meaning to a {\casp}.

\begin{example}\label{caveat-positive}
    The 
the one-rule program
$(\propishV{p}{\term{b}} \leftarrow)$ has one solution according to \Cref{def-set-solutions}, the set $\{ \propV{p}{\term{b}} \}$. The least-fixed-point interpretation of this program contains two models. The first, $(\pred{p} \mapsto \just{\term{b}})$, obviously corresponds to the unique solution. The second, $(\pred{p} \mapsto \none{\{ \term{b}} \})$, does not correspond to any solution.
\end{example}

This is a relatively straightforward
issue to resolve: we will only expect \emph{positive} models (\Cref{def-db-positive}) to correspond to solutions.

Another challenge for connecting the step semantics and the least-fixed-point interpretation has to do with infinite choice sets and constraint databases.

\begin{example}\label{caveat-finite}
The translated answer set program with no finite grounding that we showed in \Cref{fig-asp-lasy-casp} makes for an interesting example. All but one of the models in the least-fixed-point interpretation of this program are finite. Here is one finite model, corresponding to the visiting only 0 and then stopping:
\[
\left(
\propK{visit}{\term{z}} \mapsto \just{\term{tt}}, 
\;
\propK{more}{\term{z}} \mapsto \just{\term{ff}},
\;
\propK{stop}{\term{z}} \mapsto \just{\term{tt}} 
\right)
\]
That constraint database corresponds to the fact set $\{ \prop{visit}{\term{z}}{\term{tt}}, \prop{more}{\term{z}}{\term{ff}}, \prop{stop}{\term{z}}{\term{tt}} \}$,  which is a solution according to \Cref{def-set-solutions}.

In addition to countably infinite constraint databases that correspond to solutions, the least-fixed-point interpretation of the program in \Cref{fig-asp-lasy-casp} also includes a single infinite constraint database that does not correspond to any solution:
\[\bigvee_{i \in \mathbb{N}} \left( 
\propK{visit}{\term{s}^i(\term{z})} \mapsto \just{\term{tt}}, \;
\propK{stop}{\term{s}^i(\term{z})} \mapsto \just{\term{ff}}, \;
\propK{more}{\term{s}^i(\term{z})} \mapsto \just{\term{tt}} \right) \]
\end{example}

\Cref{caveat-finite} demonstrates an advantage the least-fixed-point interpretation: it lets us reason mathematically about the meaning of programs that cannot be reached in finitely many steps. The least-fixed-point interpretation of datalog confers analogous advantages: a simple forward-chaining interpreter cannot fully evaluate the datalog program with two rules $(\propK{p}{\term{z}} \leftarrow)$ and $(\propK{p}{\term{s}(x)} \leftarrow \propK{p}{x})$, but the least-fixed-point interpretation assigns the program a canonical interpretation as the infinite set $\{\propK{p}{\term{s}^n(\term{z})} \mid n \in \mathbb{N} \}$.

Because we defined solutions to be finite sets, we have to account for the fact that the least fixed point interpretation of a {\casp} may contain \suggestedchange{infinite constraint databases that do not}{constraint databases that are not finite, and therefore do not} correspond to solutions.

The following theorem establishes that there are no other caveats beyond the two highlighted by \Cref{caveat-finite,caveat-positive}:

\begin{theorem}\label{semantics-are-equivalent}
    For $D \in \DB$, the following are equivalent:
    \begin{enumerate}
        \item $D$ is a solution to $P$ by \Cref{def-set-solutions}.
        \item $D \in \lfp \bigimmcons{P}$ and $D$ is positive and finite.
    \end{enumerate}
\end{theorem}

\begin{proof}
    Details are in the supplemental materials (\reforexpanded{sec-algorithm-round-trip}{Appendix C}), but we will provide an outline here:
    
    Going from (1) to (2) we first establish that any step which is precluded by the immediate consequence operator (as discussed in the introduction to \Cref{sec-immcons}) is actually a dead end, which means that productive steps taken towards a solution are always suitably ``within'' the set of immediate consequences. We then observe that all solutions are models, fixed points as defined by \Cref{def-model}, and are therefore bounded below by some minimal model in $\lfp \bigimmcons{P}$. But any step sequence deriving a solution must be bounded above by a constraint database in $\lfp \bigimmcons{P}$, and this squeezes solutions into membership in $\lfp \bigimmcons{P}$.
    
    Going from (2) to (1) we establish that any finite constraint database in $\lfp \bigimmcons{P}$ belongs to a finite iteration of $\bigimmcons{P}$, and the action of each of those finitely-many iterations can be simulated with finite step sequences. 
\end{proof}
\section{Abstract Algorithm}\label{sec-implementation}

Choice sets are critical both in defining the immediate consequence operators $\immcons{P}$ and $\bigimmcons{P}$ and in defining the meaning of a {\casp} as $\lfp \bigimmcons{P}$. However, the infrastructure of choice sets is more heavyweight than was strictly required to define $\immcons{P}$. For any program $P$ and constraint database $D$, we can express $\immcons{P}(D)$ in a simplified form, as the least upper bound of \textit{singular} choice sets:

\begin{definition}
   A choice set $\C$ is \textit{singular} to an attribute $a$ 
   when, for all $D \in \C$, we have that $D[a'] = \none{\emptyset}$ whenever $a' \neq a$.
\end{definition}

\begin{figure}
\begin{align}
& \{ \; ( \pred{p_1} \mapsto \just {\term{b}} ) \; \} \; \vee
&
\{ \; & \left( \pred{p_1} \mapsto \just {\term{b}}, 
   \; \pred{p_2} \mapsto \just{\term{b}}, 
   \; \pred{p_3} \mapsto \just{\term{b}} \right),
\notag
\\
& \{ \; ( \pred{p_2} \mapsto \just{\term{b}} ), \; ( \pred{p_2} \mapsto \just{\term{c}} ) \;\} \;  \vee
&
& \left( \pred{p_1} \mapsto \just {\term{b}},
  \; \pred{p_2}  \mapsto \just {\term{b}}, 
  \; \pred{p_3} \mapsto \none {\{\term{b}\}} \right),
\notag
\\
& \{ \; ( \pred{p_3}  \mapsto \just {\term{b}} ), \; ( \pred{p_3}  \mapsto \none {\{\term{b}\}} ) \; \}
&
& \left( \pred{p_1} \mapsto \just {\term{b}}, 
  \; \pred{p_2} \mapsto \just {\term{c}}, 
  \; \pred{p_3} \mapsto \just {\term{b}} \right),
\notag \\
&&
& \left( \pred{p_1} \mapsto \just {\term{b}}, 
  \; \pred{p_2} \mapsto \just {\term{c}}, 
  \; \pred{p_3} \mapsto \none {\{\term{b}\}} \right) \; \}
\notag
\end{align}
\caption{Two views of the same choice set. On the right, the choice set is represented straightforwardly as a set of pairwise-incompatible constraint databases. On the left, the choice set is represented as the least upper bound of singular choice sets.}
\label{fig-two-views}
\Description{A choice set represented first, on the left, as the least upper bound of singular choice sets: a choice set mapping p1 to just b, a choice set mapping p2 to just b or just c, and a choice set mapping p3 to just b or noneof \{ b \}. The normal presentation of this least upper bound is shown on the right; it contains four constraint databases with mappings for each of p1, p2, and p3.}
\end{figure}

The least upper bound of multiple choice sets singular to $a$ is also a choice set singular to $a$, so we can view the output of immediate consequence as a \emph{map from attributes to singular choice sets}. This critical shift in perspective is illustrated in \Cref{fig-two-views} and is specified formally in supplemental materials (\reforexpanded{sec-algorithm-round-trip}{Appendix C}).

Note that while $\immcons{P}(D)$ can always be written as the least upper bound of singular choice sets, not all choice sets can be described this way.
A simple example is the program from \Cref{fig-the-littlest-asp-casp}, whose interpretation is given by this choice set:
\[
\{
\; (\pred{p} \mapsto \just{\term{ff}}, \; \pred{q} \mapsto \just{\term{tt}}), 
\; (\pred{p} \mapsto \just{\term{tt}}, \; \pred{q} \mapsto \just{\term{ff}}) 
\; \}
\]

\subsection{An Algorithm for Exploring the Interpretation of a {\CASP}}
\label{sec-algorithm}

This nondeterministic algorithm attempts to return a single element from $\lfp \bigimmcons{P}$.

\begin{itemize}
\item Initially, let $D$ be $\bot_\DB$.
\item While $\immcons{P}(D) \neq \{ D \}$ and $\immcons{P}(D) \neq \emptyset$:
\begin{enumerate}
    \item Interpret $\immcons{P}(D)$ as a map from attributes $a$ to choice sets singular in $a$. 
    \item Nondeterministically pick some attribute $a$ where $\immcons{P}(D)[a] \neq \{ (a \mapsto D[a] ) \}$. 
    
    (The loop guard $\immcons{P}(D) \ne \{D\}$ ensures some such $a$ exists.)
    \item Nondeterministically pick some $(a \mapsto c)$ in $\immcons{P}(D)[a]$.

    (The loop guard $\immcons{P}(D) \ne \emptyset$ ensures that
      $\immcons{P}(D)[a] \ne \emptyset$.)
    \item Modify $D$ by setting $D[a]$ to be $c$.
\end{enumerate}
\item Successfully return $D$ if $\immcons{P}(D) = \{ D \}$, and return failure if $\immcons{P}(D) = \emptyset$.
\end{itemize}

This algorithm captures the intuition discussed at the start of \Cref{sec-immcons}: we can use the immediate consequences of a database to guide the process of picking the next step in the step-by-step computation of solutions to a {\casp}.

\begin{theorem}\label{thm-algorithm-correct}
    For positive constraint databases $D$, the following are equivalent:
    \begin{enumerate}
        \item The algorithm described above may successfully return $D$.
        \item $D \in \lfp \bigimmcons{P}$ and $D$ is finite.
    \end{enumerate}
\end{theorem}

This a corollary of the lemmas used to establish \Cref{semantics-are-equivalent}, which are presented in the supplemental materials (\reforexpanded{sec-algorithm-round-trip}{Appendix C}).

The key to efficient evaluation of this algorithm is that we can incrementally update and efficiently query the portion of $\immcons{P}(D)$ where $\immcons{P}(D)[a] \neq \{( a \mapsto D[a] )\}$. This is not original to {\casping} or our implementation: maintaining a data structure capturing all immediate consequences is perhaps \textit{the} fundamental move for efficient implementation of semi-naive, tuple-at-a-time forward-chaining evaluation of logic programs. To give two examples, this is precisely the purpose of the queue $Q$ in \cite{mcallester02} and of the \textit{agenda} $A$ in \cite{10.1162/tacl_a_00588}. Our setting is novel because the immediate consequence of a partial solution is not another deterministically-defined partial solution --- as it is for McAllester, Eisner, and all other work we are aware of --- but a pairwise-incompatible set of partial solutions, which we explore nondeterministically.

\subsection{Resolving Nondeterminism}\label{sec-resolving-nondet}

There are two nondeterministic choice points in the abstract algorithm described in \Cref{sec-algorithm}: the choice of an attribute $a$ in step 2 of the loop, and the choice from $\immcons{P}(D)[a]$ in step 3 of the loop. The nondeterministic choice in step 2 is common to all similar semi-naive, tuple-at-a-time forward-chaining evaluation algorithms, but the nondeterministic choice in step 3 is not.

We \suggestedchange{claim}{conjecture} that the selection of an attribute (in step 2 of the abstract algorithm's loop) isn't ``lossy:'' when the algorithm might pick one of two attributes, the choice of one won't preclude finding any answer that the program might have returned if it had instead picked the other attribute. This means that a backtracking version of our algorithm that searches for all solutions doesn't need to reconsider the choices made in step 2. 

While each \emph{individual} attribute choice cannot cut off the path to any particular solution, a \emph{systematic} bias in the way attributes are selected can interfere with the algorithm's nondeterministic completeness. Nevertheless, we choose for our implementation to be systematically biased by the \emph{deduce-then-choose} strategy introduced in \Cref{sec-canonical-representatives}. In the context of our abstract algorithm, the deduce-then-choose strategy always picks an attribute $a$ where $\immcons{P}(D)[a]$ is a singleton when such a choice is possible. This strategy does force non-termination on the abstract algorithm where it might otherwise be able to terminate, as the following example shows:

\begin{example}\label{example-force-nontermination}
    Consider the following five-rule {\casp}:
\begin{align}
    \propV{p}{\{ \term{red} \}} & \leftarrow \label{rx1} \\
    \propV{q}{\{ \term{blue}, \term{yellow} \}} & \leftarrow \propK{num}{\term{s}(\term{s}(\term{z}))} \label{rx2} \\
    \propV{p}{\{ x \}} & \leftarrow \propV{q}{x} \label{rx3}\\
    \propK{num}{\term{z}} & \leftarrow \label{rx4} \\
    \propK{num}{ \term{s}(x) } & \leftarrow \propK{num}{x} \label{rx5}
\end{align}

\noindent
This program has no solutions --- its least-fixed-point interpretation is the empty set --- and it is possible for the abstract algorithm in \Cref{sec-algorithm} to finitely return failure or to loop endlessly, depending on how nondeterministic choices are resolved.
However, the additional constraint imposed by the deduce-then-choose strategy means that our algorithm cannot finitely return failure, and can only fail to terminate.
\end{example}

Unlike the choice of which attribute to consider in step 2 of the abstract algorithm, the choice in step 3 to pick a specific value necessarily cuts off possible solutions. However, we can turn our abstract algorithm into a partially correct algorithm for enumerating the elements of a least-fixed-point interpretation by backtracking over the choices made in step 3. 

Alongside backtracking, we add in one final bias. While our abstract algorithm enumerates arbitrary members of $\lfp \bigimmcons{P}$, we are interested only in \textit{solutions}, which by \Cref{semantics-are-equivalent} are the \textit{positive} constraint databases in $\lfp \bigimmcons{P}$. Accordingly, we restrict our backtracking strategy to prefer  constraints of the form $(\just{t})$ prior to backtracking. This means that if the algorithm returns without backtracking, the result will be a positive model.

This combination of strategies---the deduce-then-choose strategy, no backtracking over choices of attribute, and some form of backtracking over the choice of constraint while preferring positive constraints---means that we can think of an execution of the abstract algorithm as a process of expanding a (potentially infinite) tree of decisions. 

\begin{example}\label{example-explore-an-execution}
    Consider the following five-rule {\casp}:
\begin{align}
    \propV{p}{\{ \term{tt}, \term{ff} \}} & \leftarrow \label{rg1}\\
    \propV{q}{\{ \term{tt}, \term{ff} \}} & \leftarrow \label{rg2} \\
    \propishV{r}{\term{a}} & \leftarrow \label{rg3}\\
    \propV{r}{\{ \term{b}, \term{c} \}} & \leftarrow \propV{p}{\term{ff}} \label{rg4}\\
    \propV{r}{\{ x \}} & \leftarrow \propV{p}{x}, \propV{q}{x} \label{rg5}
\end{align}

Two of the many possible induced execution trees are show in Figure~\ref{fig-explore-an-execution}. The upper tree
\Cref{fig-explore-an-execution} represents
a execution where $\pred{p}$ is considered first in step 2 of the abstract algorithm's loop, and the lower tree represents an execution where $\pred{r}$ is considered first.

In the first example where $\pred{p}$ is considered first, when
the algorithm decides to give $\pred{p}$ the value $\just{\term{tt}}$, the
next attribute considered is $\pred{q}$, but when $\pred{p}$ is given the value $\just{\term{ff}}$,
the next attribute considered is $\pred{r}$. This highlights that we are not insisting on a single global ordering of attributes.
\end{example}

\begin{figure}

\begin{center}
\includegraphics[width=13.8cm]{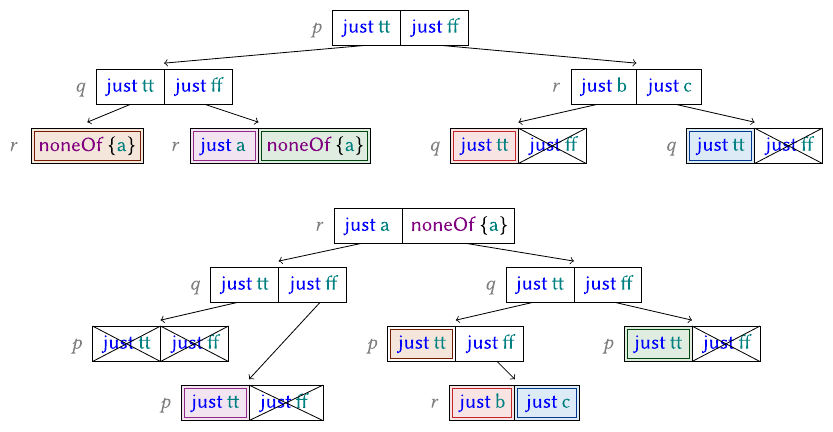}

\end{center}
\caption{Two trees of nondeterministic execution induced by different ways of selecting
attributes in course of executing the {\casp} from \Cref{example-explore-an-execution}.
Leaves representing failure-return are crossed out, and leaves representing successful terms are highlighted.     \endgraf
\indent\hspace{1em} The algorithm potentially returns the same five models regardless of the order in which attributes are selected. The related models are outlined in the same color in the two trees. For example, the leaf that corresponds to returning the model $(p \mapsto \just{\term{tt}}, \; q \mapsto \just{\term{tt}}, \; r \mapsto \none{\{ \term{a} \}})$ is highlighted in brown and can be found on the left side of the upper tree and near the center of the lower tree.}
\label{fig-explore-an-execution}
\Description{Two trees: the first tree describes an execution that first branches on the attribute "p," and the second describes an execution that first branches on the attribute "r." The trees differ in the number of routes that end up eliminated due to incompatible assignments (two in the first tree, four in the second), but both trees have leaves corresponding to the same five models; colors indicate how the solutions correspond to each other.}
\end{figure}

\subsection{Completeness of the Constrained Algorithm}
\label{sec-completeness-constrained}

\Cref{example-force-nontermination} demonstrated that the restrictions we put on the nondeterministic algorithm in this section force the non-backtracking version of our algorithm into non-termination where it might have otherwise terminated signaling failure. It is an open question whether the restrictions in this section interfere with the nondeterministic completeness of the algorithm: in other words, we don't know whether there is a program $P$ and a database $D \in \lfp \bigimmcons{P}$ such that the more restricted algorithm is unable return $D$. We conjecture that, even with the refinements we have discussed, any finite $D \in \lfp \bigimmcons{P}$ can be reached with a suitable choice of attributes in step 2 of the algorithm's loop. Even if these restrictions do sacrifice completeness, we believe they are worthwhile. The deduce-then-choose strategy makes the performance of {\casps} vastly more predictable, as demonstrated by the examples in \Cref{sec-examples}.

\section{Implementation}\label{sec-implementation-in-dusa}

Our implementation of the constrained abstract algorithm for {\casping} is called {\dusa}. The implementation supports three modes of interaction: a TypeScript API, a command-line program, and \suggestedchange{a browser-based editor.}{the browser-based editor pictured in \Cref{fig-dusa-rocks} and accessible at \url{https://dusa.rocks/}.} All modes allow the client to request incremental enumeration of all solutions to a {\casp}, and the first two modes allow the client to sample individual solutions \cite{simmons_2024_13921177}.

Despite implementing the algorithm in \Cref{sec-algorithm}, which enumerates finite models in $\lfp \bigimmcons{P}$, our implementation discards all non-positive models and only exposes solutions---the positive models in the program's interpretation---to the client. That's why, while the interpretation of the program in \Cref{example-explore-an-execution} includes five different models, the web interface in \Cref{fig-dusa-rocks} only indicates the presence of four solutions. 

\begin{figure}
\includegraphics[width=12cm]{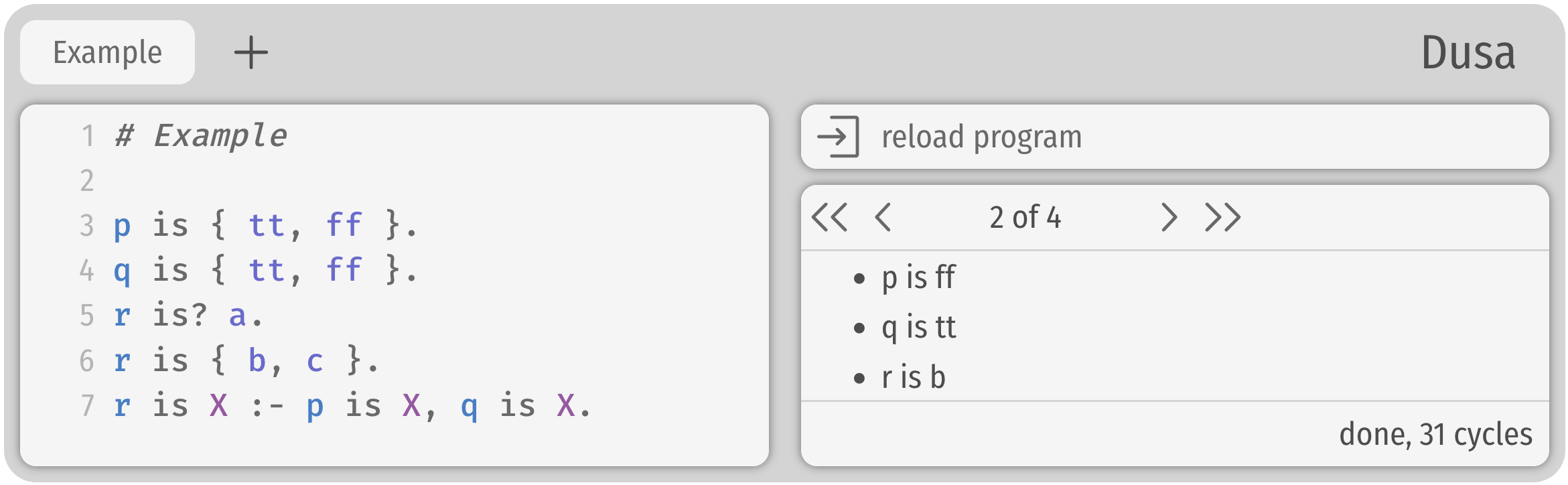}
\caption{Online editor for the {\dusa} implementation of {\casping} at \url{https://dusa.rocks/}.}
\Description{A web page with two main panels: the left panel contains the concrete syntax of a {\dusa} program, and the right panel indicates that we are looking at the second of four solutions, a solution that contains four facts: ``p is ff'', ``q is tt'', and ``r is b''.}
\label{fig-dusa-rocks}
\end{figure}

Our implementation of {\casping} was initially designed with procedural generation and possibility-space exploration in mind. This makes our design considerations very similar to those that led Horswill to design CatSAT as a way to facilitate procedural generation in Unity games \cite{horswill18catsat}. Horswill identifies four barriers to the use of answer set programming in video games, all of which are relevant to {\casping} and the {\dusa} implementation:

\paragraph{Designer Transparency} By this, Horswill refers to tools that allow \suggestedchange{``non-programmers''}{designers} to understand and manipulate programs. CatSAT does not directly address this barrier, and while {\dusa} has not been systematically tested with designers, our web interface does make it possible for non-experts to inspect, manipulate, and share programs without downloading or installing software (meeting many of the criteria for {\em casual creators} \cite{compton2015casual}).

\paragraph{Performance} CatSAT uses a simple implementation of the WalkSAT algorithm to get reasonable time performance and excellent memory performance on satisfability problems that are not excessively constrained. {\dusa} approaches performance from a different angle: our implementation is relatively memory-intensive, but we aim to provide \textit{predictable} performance, which we discuss more in \Cref{sec-perf-eval}.

\paragraph{Run-Time Integration} Whereas answer set programming tools are primarily written as standalone programs, CatSAT is implemented as a DSL in C\# to naturally integrate with programs in the Unity ecosystem. {\dusa} is implemented in TypeScript to facilitate integration with the web ecosystem. 

\paragraph{Determinism} Horswill observes that most SAT solvers and answer set programming languages make it difficult or impossible to access suitably random behavior. Like CatSAT, the {\dusa} implementation defaults to randomness. After constraining the abstract algorithm as described in \Cref{sec-resolving-nondet}, our implementation resolves all remaining nondeterministic choices uniformly at random. Furthermore, rather than using a stack-based backtracking algorithm, our implementation creates in-memory representations of the trees described in \Cref{fig-explore-an-execution}. After a solution is discovered, the algorithm prunes fully-explored parts of the tree and then returns to the root of the tree to begin random exploration again. The goal of this strategy is to avoid returning a second solution that is very similar to the first, a behavior that is commonly observed when exploring a state space with a stack-based (or depth-first) backtracking search.

While our strategy for randomized exploration does not truly sample the interpretation at random, in practice it has proven suitable for procedural generation. \suggestedchangerationale{This strategy also does an empirically good job of locating all the ``small'' solutions for programs with an infinite interpretation, like the one in \Cref{fig-asp-lasy-casp}, and will eventually return any specific finite solution for that particular program with probability 1. The Alpha implementation of ASP with lazy grounding \cite{alpha} exhibits a different behavior that is wholly unsuitable for procedural generation: when we ask Alpha to enumerate solutions to the answer set program in \Cref{fig-asp-lasy-casp}, we observe that it skips most small models and successively returns larger and larger models with large gaps in between.}{In our experience, {\dusa} also does a good job of locating ``small'' solutions for programs with an infinite interpretation like the one in \Cref{fig-asp-lasy-casp}. However, when asked to enumerate multiple solutions for this specific program, {\dusa} will start by presenting smaller solutions and will then only present successively larger solutions, never returning to smaller solutions that were skipped along the way. This behavior, which {\dusa} shares with the Alpha implementation of ASP with lazy grounding \cite{alpha}, is unsuitable for procedural generation, and we hope to address it in future work.}{The initial submission made some claims here that were incorrect, which we identified due to both the artifact evaluation process and Reviewer A's questions.} 

\subsection{Predictable Performance With Prefix Firings}\label{sec-perf-eval}

The implementation of the abstract evaluation algorithm in {\dusa} is closely modeled after the semi-naive, tuple-at-a-time forward-chaining interpreter for datalog presented by McAllester \shortcite{mcallester02}. McAllester's algorithm was designed to facilitate a \textit{cost semantics}, a way of reasoning about the run-time cost of evaluation without reasoning about the details of how the implementation works.

\begin{figure}
\includegraphics[width=13.8cm]{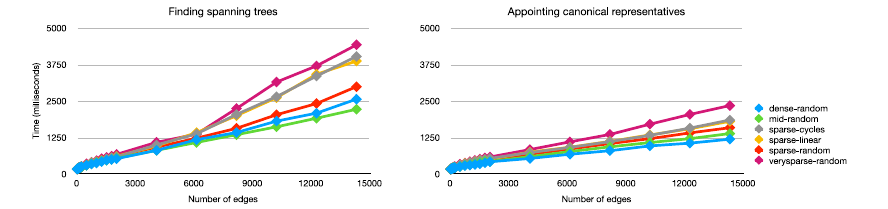}
\caption{Performance of the spanning tree program from \Cref{sec-rooted-spanning-tree} and the canonical representative algorithm from \Cref{sec-canonical-representatives} on graphs with different characteristics. All data points are the median of three runs requesting a single solution.}
\Description{Two graphs, each displaying six tightly-spaced lines that indicate approximately linear growth. The vertical axis is time in milliseconds, the horizontal axis is the number of edges in the graphs. The left graph is titled finding spanning trees, the right graph is labeled appointing canonical representatives, and the legend for the six lines describes graphs with different features: dense-random, mid-random, sparse-cycles, sparse-linear, sparse-random, and verysparse-random.}
\label{fig-its-kinda-linear}
\end{figure}

McAllester's cost semantics is based on counting the number of \textit{prefix firings} in a solution. If a rule $H \leftarrow F$ has the premises $F = \prop{p_1}{\overline{t_1}}{v_1}, \ldots, \prop{p_n}{\overline{t_n}}{v_n}$, then a prefix firing for the database $D$ is a unique variable-free instantiation of the first $i \leq n$ premises $\prop{p_1}{\sigma{\overline{t_1}}}{v_1},\ldots,\prop{p_i}{\sigma{\overline{t_i}}}{v_i}$ that is satisfied \suggestedchange{by the $D$}{in $D$}. The central result of McAllester's work is that, assuming constant-time hashtable operations, a datalog program that returns the solution $D_\textit{final}$ can be evaluated in time proportional to the prefix firings for $D_\textit{final}$. This cost semantics allows programmers to make decisions about how to write declarative programs: the McAllester cost semantics suggests defining the recursive rule for a transitive closure as $\propK{path}{x,z} \leftarrow \propK{edge}{x,y}, \propK{path}{y,z}$ rather than $\propK{path}{x,z} \leftarrow \propK{path}{x,y}, \propK{path}{y,z}$, because in a graph where the number of edges is proportional to $n$, the number of vertices, the former rule has $O(n^2)$ prefix firings and the latter has $O(n^3)$. 

\suggestedchange{We do not present a theorem about cost semantics in this development, in part because it requires careful introduction of a number of extra concepts and details, and in part because our implementation is based on functional maps instead of hash tables in order to support the very general backtracking methodology described above; this means that a variant of McAllester's result will not apply directly to our implementation.}{It would be nice to say that McAllester's cost semantics apply to {\dusa} programs that reach a solution without backtracking, but that isn't quite the case. McAllester's algorithm represents an extreme point in possible time-space tradeoffs: the result of every deduction is memoized in a hash table. (If one makes the standard assumption that hash table operations take constant time, this results in a space complexity equal to the time complexity.) In order to support our general backtracking strategy, our implementation uses AVL trees to implement the maps that McAllester's algorithm implements as hash tables. To apply a prefix-firing cost semantics to {\casping}, we will need to take one of two approaches: either we must adapt the cost semantics to account for the logarithmic factors involved in accessing functional maps, or we must describe a modified implementation of our algorithm that uses hash tables.} 

Nevertheless, \suggestedchange{the}{the prefix-firing} cost semantics \suggestedchange{has \textit{predictive} power.}{is successful at predicting the behavior of {\dusa} programs that do not backtrack.} The programs for spanning tree generation in \Cref{fig-spanning-tree-prog} and for canonical representative appointment in \Cref{fig-canonical-rep-prog} are expected, in any deduce-then-choose execution, to find some solution without backtracking, and solutions for both programs will have prefix firings proportional to the number of edges. \Cref{fig-its-kinda-linear} demonstrates that the actual behavior of {\dusa} generally conforms to the running time predicted by the cost \suggestedchange{semantics.}{semantics, possibly with some additional logarithmic factors thrown in to account for the use of functional data structures.}

\subsection{Comparing {\dusa} to Answer Set Programming}
\label{sec-dusa-and-asp-perf}

\begin{figure}
\includegraphics[width=13.8cm]{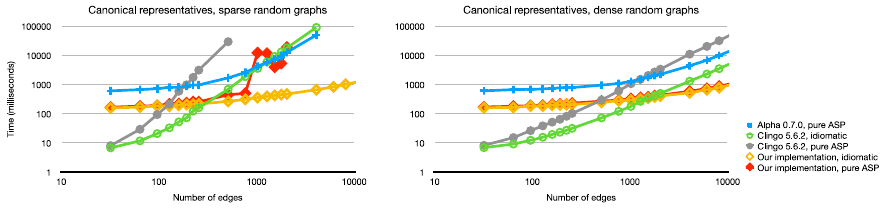}
\caption{Performance of Alpha, Clingo, and {\dusa} at finding canonical representatives. All data points are the median of three runs requesting a single solution, discarding runs exceeding 100 seconds.}
\Description{Two log-log graphs with a vertical axis representing time in milliseconds and a horizontal axis representing the number of edges in an input graph. Both graphs represent the time it takes to solve the appointment-of-canonical-representatives problem, the left one on sparse random graphs and the right one on dense random graphs. The five different lines represent different implementations, as discussed further in the text.}
\label{fig-head-to-head}
\end{figure}

There's not an obvious way to ask an answer set programming engine to run the spanning tree program from \Cref{sec-rooted-spanning-tree} or the canonical representative program from \Cref{sec-canonical-representatives}, because we have not defined a translation from {\casps} to answer set programs. However, solutions to the rooted spanning tree problem and the canonical representative problem can be expressed in answer set programming, and these answer set programs can be translated into {\casps} as discussed in \Cref{sec-simulating-asp,lazy-asp-in-casp}. 

In \Cref{fig-head-to-head} we compare the canonical representative problem from \Cref{fig-canonical-rep-prog} against an idiomatic Clingo program that uses features like cardinaity constraints, as well as comparing the performance \suggestedchange{of as single}{of a single} ``pure'' answer set program executed in Clingo, in the Alpha implementation of ASP with lazy grounding \cite{alpha}, and in {\dusa} by direct translation of ASP to {\casping}. (The yellow lines in \Cref{fig-head-to-head} repeat the red and blue lines from the right-hand side of \Cref{fig-its-kinda-linear}.) In contrast with the predictable behavior shown in \Cref{fig-its-kinda-linear}, the performance of the other four solutions varies significantly depending on characteristics of the graph. In the head-to-head comparison, {\dusa} outperforms the state of the art on dense graphs but does worse as graphs get sparser.

The poor performance of Clingo in the head-to-head comparison can be attributed to the ``grounding bottleneck'' encountered by traditional approaches to ASP. Both Alpha's lazy grounding and {\dusa}'s deduce-then-choose execution model represent ways of avoiding this grounding bottleneck, and this translates to better asymptotic performance for this example. 

\suggestedchange{An extended set of}{Compared with modern ASP implementations, the backtracking strategy used by the current {\dusa} implementation is incredibly naive, comparable to 1990s-era DPLL SAT solvers that lacked non-chronological backtracking. This lack of sophistication is almost certainly the cause of {\dusa}'s relatively poor performance in head-to-head comparisons on sparse graphs. More generally, our implementation struggles on programs like the Boolean satisfiability example in \Cref{fig-sat-instance-prog} that define a possibility space generically and then use constraints to whittle it down to a desired state. This \emph{design-space sculpting} approach to modeling \cite{smith2011answer} is a natural way to express many problems: to implement procedural generation we can let every point in a grid be land or sea and demand that there be a path between two specific points, to implement graph coloring we can assign every node in a graph to one of five colors and demand that no two edge-connected nodes have the same color, and to implement the $N$-queens problem we can let every space on a $N$-by-$N$ chessboard contain a queen or not and demand that there are $N$ queens on the board and that there are no immediate avenues of attack. Variants of all these problems are considered in our extended} benchmarking results, \suggestedchange{including the programs used in this comparison, are available in supplemental materials (\Cref{sec-benchmarking-and-examples}).}{available in supplemental materials (\reforexpanded{sec-benchmarking-and-examples}{Appendix D}). In future work, we intend to adapt {\dusa} to better support sculptural {\casping} by adapting techniques such as conflict-driven nogood learning that have proven successful in similar settings \cite{gebser2012nogood}.}
\section{Related Work}

Our least-fixed-point interpretation of {\casping} describes a domain-like structure for nondeterminism that 
draws inspiration from domain theory writ large~\cite{scott1982domains}, particularly 
powerdomains for nondeterministic lambda calculi and imperative programs~\cite{plotkin1976powerdomain,smyth1976powerdomains,kennaway1980theory}.
Our partial order on choice sets matches the one used for Smyth powerdomains in particular~\cite{smyth1976powerdomains}.
Our construction requires certain completeness criteria on posets that 
differ from these and other domain definitions we have found in the literature.

Both the least-fixed-point interpretation and the step-sequence interpretation of {\casping} represent ways of reasoning about the incremental development of solutions. 
Because we provide a translation of answer set programming into {\casping}, these ways of reasoning about incremental development of solutions are applicable to ASP as well. 
Previous work in this area includes the \emph{stable backtracking fixedpoint} algorithm described by Sacca and Zaniolo \shortcite{sacca1990stable}. (We interpret their 
algorithm as giving the first direct account for answer set programming without grounding,
though it seems to us that this significant fact was not noticed or exploited by the authors or anyone else.) There has also been a disjunctive extension to datalog~\cite{eiter1997disjunctive} and characterizations
of its stable models~\cite{przymusinski1991stable};
Leone et al.~\shortcite{leone1997disjunctive} present an algorithm for incrementally deriving stable models for disjunctive datalog. Their approach is based on initially creating a single canonical model by letting
some propositions be ``partially true.''

At the core of {\casping} is the propositional form $\prop{p}{\overline{t}}{v}$, which establishes a \emph{functional dependency} from an attribute to a value. This mirrors the proposition form $R[x_1,\ldots,x_{n-1}]=x_n$ used to enforce integrity constraints in \suggestedchange{LogicQL}{LogiQL}~\cite{aref15logicblox}. 
Souffl\'{e}'s functional dependencies allow programs to express multiple mutually contradictory conclusions \cite{hu21soufflechoice}, but as as discussed in \Cref{sec-rooted-spanning-tree}, functional dependencies in the tradition of Krishnamurthy and Naqvi \shortcite{krishnamurthy88nondet} cannot express that the violation of a functional dependency should cause a potential solution to be rejected. Roughly speaking, LogiQL works like a {\casping} language that only has closed rules without multiple choices, and Souffl\'{e} works like a {\casping} language that has only open rules. These systems are less expressive than ours in this respect, but they can completely avoid expensive backtracking. 

The partially-ordered set $\Constraint$ introduced in \Cref{sec-constraints} suggests a relationship between {\casping} and other languages that assign non-Boolean values to propositions: values drawn from a semiring in weighted logic programming \cite{eisner-etal-2005-compiling}, various interpretations of degree-of-truth in annotated logic programming \cite{KIFER1992335} and bilattice-based logic programming \cite{FITTING199191}, and probabilities in probabilistic logic programming \cite{6278236}. As discussed in \Cref{sec-choice-sets-needed}, we're unaware of work along these lines that can give an account for stable models without falling back on a variant of Gelfond and
Lifschitz's syntactic transformation. 
However, the related mathematical structures shared across these approaches points towards generalizing $\Constraint$ to accept alternate partial orders, similar to the lattice-based functional dependencies seen in datalog extensions like Bloom~\cite{conway2012logic}, Flix~\cite{madsen2020fixpoints}, and Egglog~\cite{zhang2023better}\suggestedchange{which could expand the expressiveness of {\casping} to account for monotonic data aggregation~\cite{ross1992aggregation}}{}.

Fandinno et al.~\shortcite{fandinno23} extend a translation due to Nieml{\"a} \shortcite{niemla08} to formalize the ``idea from folklore'' that ASP can be based on an ASP choice rules along with integrity constraints. This is similar to our observation that choice, rather than negation, is a suitable foundation for answer set programming, though a limited (stratified) negation is still necessary in their translation. They present a conceptual operational semantics based on making \emph{all} possible choices up-front, followed by a phase of deduction; this operational interpretation is intended as a pedagogical tool rather than the basis of a practical implementation.

\section{Conclusion}

We have introduced the theory and implementation of {\casping}, an approach to logic programming that has connections to answer set programming but that treats \emph{choice}, not \emph{negation}, as the fundamental primitive of nondeterminism. We establish that choice is a suitably expressive foundation by showing that answer set programming can be defined in terms of {\casping}.

We give a first definition of {\casping} in terms of nondeterministically augmenting a set of facts, a second definition as the unique least fixed point in a novel domain of mutually-exclusive models, and a proof that these two definitions agree on solutions. 

Our {\dusa} implementation
can enumerate solutions to {\casps}, and the runtime behavior of our
implementation can often reliably (if approximately) be predicted by 
McAllester's 
cost semantics based on prefix firings. Because our implementation is not subject to the ``grounding bottleneck'' that affects mainstream answer set programming solvers, our implementation outperforms the state-of-the-art Clingo ASP implementation in many cases despite its extremely naive backtracking strategy.

\suggestedchange{}{In future work, we hope to firmly establish the status of nondeterministic completeness for the abstract algorithm as restricted by the deduce-then-choose strategy (as discussed in \Cref{sec-completeness-constrained}), investigate fair enumeration of solutions (as discussed in \Cref{sec-implementation-in-dusa}),
and apply proven techniques from ASP solvers and Boolean satisfiability solvers to improve {\dusa}'s performance on programs where the current implementation's simplistic backtracking strategy does not perform well (as discussed in \Cref{sec-dusa-and-asp-perf}). In addition, we plan to investigate the semantics of fragments of our language,
both by giving a precise account for connecting open rules and languages with choice constructs like Souffl\'{e}'s and by seeing if closed rules can account for disjunctive datalog. We also hope to generalize {\casping} as presented here, particularly by extending the partial order on ${\Constraint}$ to richer partial orders, as has been done in related systems, which we believe
could expand the expressiveness of {\casping} to account for monotonic data aggregation~\cite{ross1992aggregation}. Finally, we hope to explore 
proof-theoretic accounts of {\casping}, which could support compositional
explanation and provenance analysis for possibility spaces.
}

\begin{acks}
  This material is based upon work supported by the National Science Foundation under Grant No. 1846122. The genesis of this work was performed at the Recurse Center, and we are grateful for the Recurse Center for nurturing a space where people can think about ideas in programming.

We want to thank Richard Comploi-Taupe, who helped us understand the state of research on answer set programming with lazy grounding, and Ben Simner, who helped acquire reference materials on disjunctive logic programming. Many folks gave helpful feedback that improved this paper, particularly Mitch Wand, Emma Tosch, Jason Reed, Philip Zucker, David Renshaw, and six anonymous reviewers.

\end{acks}

\clearpage

\bibliographystyle{ACM-Reference-Format}
\bibliography{main}


\clearpage
\appendix
\section{Datalog and Answer Set Programming Connections}
\label{datalog-and-asp-semantics}

This section includes additional details from \reforexpanded{sec-simulating-datalog,sec-simulating-asp}{Sections 2.2 and 2.3} in the main body of the paper.

\subsection{Simulating Datalog}\label{sec-simulating-datalog-details}

A datalog program is a finite collection of rules of this form:
\begin{align}
 \propK{p}{\overline{t}} & \leftarrow \propK{p_1}{\overline{t_1}}, \ldots, \propK{p_n}{\overline{t_n}} \tag{datalog rule}
\end{align}

The model of a datalog program is traditionally understood as the least fixed point of the (datalog) immediate consequence operator, which takes a set $X$ of datalog facts and returns the set of rule conclusions $\propK{p}{\sigma\overline{t}}$ such that all that rule's premises $\propK{p_i}{\sigma\overline{t_i}} \in X$. These facts can be given a total order:

\begin{lemma}\label{lem-get-sequence}
If the model of a datalog program is finite, it is possible to order the elements of that model in a sequence
$\propK{p_1}{\overline{t_1}} \ldots \propK{p_n}{\overline{t_n}}$, such that $\propK{p_i}{\overline{t_i}}$ is always
an immediate
consequence of $\propK{p_j}{t_j}$ for $j < i$.
\end{lemma}

\begin{proof}
    If a datalog model is finite, it can be derived by a finite number of
applications of the (datalog) immediate consequence operator, where 
each application adds a finite number of additional facts. We construct the
sequence by placing all the facts derived with the first application
of the immediate consequence operator first, followed by 
any additional facts derived with the second application of the immediate
consequence operator, and so on.
\end{proof}

We translate datalog rules of the form shown above to the {\casping} rules that look like this: 
\begin{align}
\prop{p}{\overline{t}}{\{ \term{unit} \}} & \leftarrow \prop{p_1}{\overline{t_1}}{\term{unit}}, \ldots, \prop{p_n}{\overline{t_n}}{\term{unit}} \notag
\end{align}
The paper glosses over the statement of correctness for this translation, but for completeness we make it explicit here: 
\begin{theorem}\label{thm-corr-datalog}
Let $P$ be a datalog program with a finite model, and 
let $\langle P \rangle$ be the interpretation of $P$ as a {\casp} as above. Then there is a unique solution $D$ to $\langle P \rangle$,
and the model of $P$ is $\{ \propK{p}{\overline{t}} \mid \prop{p}{\overline{t}}{\term{unit}} \in D \}$.
\end{theorem}

\begin{proof}
    By \Cref{lem-get-sequence} we can get some sequence
$X_\mathsf{seq} = \propK{p_1}{\overline{t_1}} \ldots \propK{p_n}{\overline{t_n}}$
containing exactly the elements in the
model of $P$, where $\propK{p_i}{\overline{t_i}}$ is always
an immediate
consequence of the set of facts that precede it in the sequence. By induction on $k$ we can derive a step sequence from $\emptyset$ to $\{ \prop{p_i}{\overline{t_i}}{\term{unit}} \mid  1 \leq i \leq k \wedge \propK{p_i}{\overline{t_i}} \in X_\mathsf{seq} \}$ for any $1 \leq k \leq n$. We will let $D = \{ \prop{p_i}{\overline{t_i}}{\term{unit}} \mid \propK{p_i}{\overline{t_i}} \in X_\mathsf{seq} \}$. There is a step sequence $\emptyset\ldots D$, so if $D$ is saturated, it is a solution.
Assume some additional fact $\prop{p}{\overline{t}}{ \term{unit}}$
can be derived in the translated program, a corresponding additional fact $\propK{p}{\overline{t}}$ can be derived in the model of $P$ that isn't present in $X_\mathsf{seq}$, a contradiction.

It remains to be shown that $D$ is the unique solution. Consider an arbitrary
solution $D'$. First, by induction on the step sequence that reaches $D'$, if 
$\prop{p}{\overline{t}}{v} \in D'$ it must be the case that $v = \term{unit}$.
If $D'$ was missing any elements in $D$ then it wouldn't be a solution, and if 
$D$ was missing any elements in $D'$ we could use that to show that $X_\mathsf{seq}$
was missing some element in $P$'s unique model, so $D = D'$.
\end{proof}

\subsection{Simulating Answer Set Programming}\label{sec-simulating-asp-details}

Answer set programs are defined in terms of the \textit{stable model semantics} \cite{gel88}. A rule in answer set programming has both non-negated (positive) premises,
which we'll write as $\pred{p_i}$, and negated (negative) premises, which 
we'll write as $\pred{q_i}$.
\begin{align}
\pred{p} \leftarrow \pred{p_1}, \ldots, \pred{p_n}, \neg \pred{q_{1}}, \ldots, \neg \pred{q_m} \tag{ASP rule}
\end{align}
Let $X$ denote a finite set of ground predicates. When modeling ASP, any 
predicate in the set is treated as true, and any predicate not in the set is false.
To explain what a {\em stable} model is, the ASP literature~\cite{gel88} 
first defines $P^X$, the \textit{reduct} of a program $P$ over $X$ obtained by:
\begin{enumerate}
    \item removing any rules where a negated premise appears in $X$
    \item removing all the negated premises from rules that weren't removed in step 1
\end{enumerate}
The reduct $P^X$ is always a regular datalog program with a unique finite model. If the model 
of $P^X$ is exactly $X$, then $X$ is a stable model for the ASP program $P$.

The translation of a single ASP rule with $m$ negated premises of the form above produces $m+1$ rules in the resulting {\casp}. These introduced
rules use two new program-wide introduced constants, $\term{tt}$ and $\term{ff}$, which
represent the answer set program's assignment of truth or falsehood, respectively, to the relevant 
attribute:
\begin{align}
\propishV{q_{1}}{\term{ff}} & \leftarrow \propV{p_1}{\term{tt}}, \ldots, \propV{p_n}{\term{tt}} \notag \\
\propishV{q_{2}}{\term{ff}} & \leftarrow \propV{p_1}{\term{tt}}, \ldots, \propV{p_n}{\term{tt}}, \propV{q_{1}}{\term{ff}} \notag \\
& \ldots \notag \\
\propishV{q_m}{\term{ff}} & \leftarrow \propV{p_1}{\term{tt}}, \ldots, \propV{p_n}{\term{tt}}, \propV{q_{1}}{\term{ff}}, \ldots, \propV{q_{m-1}}{\term{ff}} \notag\\
\propV{p}{\{ \term{tt} \}} & \leftarrow \propV{p_1}{\term{tt}}, \ldots, \propV{p_n}{\term{tt}}, \propV{q_{1}}{\term{ff}}, \ldots, \propV{q_{m-1}}{\term{ff}}, \propV{q_m}{\term{ff}} \notag
\end{align}
It is this translation that we will use to prove \reforexpanded{thm-sound-complete-asp}{Theorem 2.12}, but it is worth noting that nothing in the proof \textit{requires} the introduced open rules to have all the premises shown above. The proof works equally well if arbitrary premises are removed from the open rules, so this translation should be seen as one end of a spectrum of possible translations, the version that is the most ``hesitant'' to derive facts of the form $\propV{q}{\term{ff}}$. At the other end of the spectrum is the translation that just introduces a rule $(\propishV{q_i}{\term{ff}} \leftarrow)$ for every negated premise in the source program. 

We will prove the two directions of \reforexpanded{thm-sound-complete-asp}{Theorem 2.12} separately as \Cref{thm-sound-asp} and \Cref{thm-complete-asp}

\begin{proposition}\label{thm-sound-asp}
Let $P$ be an ASP program, and let $\langle P\rangle$ be the
interpretation of $P$ as a {\casp} as defined above. 
For all stable models $X$ of $P$, there is a solution $D$ to $\langle P \rangle$ such that $X = \{ \pred{p} \mid \propV{p}{\term{tt}} \in D \}$.
\end{proposition}

\begin{proof}
    Let $X$ be a stable model of $P$. The reduct $P^X$ is just
a datalog program, so by \Cref{lem-get-sequence} we can get some sequence 
$X_\mathsf{seq} = \pred{p_1} \ldots \pred{p_k}$ 
containing exactly the elements of $X$ 
such that each $\pred{p_i}$ is the immediate consequence 
of the preceding facts according to some rule in the reduct $P^X$.

We will construct our solution $D$ in two steps. In the first step, we establish by induction on $j$ that, for any prefix $\pred{p_1} \ldots \pred{p_j}$ of $X_\mathsf{seq}$, there 
is a database $D_j$ and step sequence $\emptyset\ldots D_j$ under such that $\langle P \rangle$ such that:
\begin{itemize}
    \item $\pred{p'} \is \term{tt} \in D_j$ if and only if
$\pred{p'} \in \pred{p_1}, \ldots, \pred{p_j}$.
    \item $\pred{p'} \is \term{ff} \in D_j$ only if $\pred{p'} \not\in X$.
\end{itemize}
Applying this inductive construction to the entire sequence $X_\mathsf{seq}$, we get a 
$D_k$ with $\pred{p'} \is \term{tt} \in D$ if and only if $\pred{p'} \in X$.
In the second step, we add additional facts of the form $\propV{q'}{\term{ff}}$
to $D_k$ until we have a solution $D$ to $\langle P \rangle$
where $\propV{p'}{\term{tt}} \in D$ if and only if $\pred{p'} \in X$.

\paragraph{First step} The base
case of our induction is immediate for $D_0 = \emptyset$. In the inductive case,
we're given a step sequence $\emptyset\ldots D_{j-1}$ under $\langle P \rangle$
where $D_{j-1}$ contains
$\propV{p'}{\term{tt}}$ for exactly the $p' \in p_1\ldots p_{j-1}$. We need to construct $D_j$ by adding one fact $\propV{p_j}{\term{tt}}$ to $D_{j-1}$. We can also add any facts of the form $\propV{p'}{\term{ff}}$ so long as $\pred{p'} \not\in X$.

When we derive $\pred{p_j}$ in the reduct, 
we know that this was due to a rule in the reduct whose positive premises all appear in $X_\mathsf{seq}$ prior to $\pred{p}$.
This rule in the reduct $P^X$ was itself derived from a rule
$\pred{p_j} \leftarrow \pred{p_1}, \ldots, \pred{p_n}, \neg \pred{q_{1}}, \ldots, \neg \pred{q_{m}}$
in the original program $P$. Having found this rule,
we let $D_j = D_{j-1} \cup \{ \propV{p}{\term{tt}}, \propV{q_1}{\term{ff}}, \ldots, \propV{q_m}{\term{ff}} \}$.
This 
immediately satisfies the first condition that $\pred{p'} \is \term{tt} \in D$ if and only if
$\pred{p'} \in \pred{p_1}, \ldots, \pred{p}$. Because the reduct survived
into $P^X$, we know none of the $q_i$ are in $X$, so the third condition that
$\pred{p'} \is \term{ff} \in D$ only if $\pred{p'} \not\in X$ is also satisfied.

We need only to construct a step sequence from $D_{j-1}$ to $D_j$. We will first take $m$ steps, one for every negative premise, using the $m$ open rules in the translation to add $\propV{q_i}{\term{ff}}$ for $i$ in order from 1 to $m$.
Having added all the relevant negative premises, we have satisfied all the premises of the one translated closed rule:
\begin{align}
\propV{p_j}{\{ \term{tt} \}} & \leftarrow \propV{p_1}{\term{tt}}, \ldots, \propV{p_n}{ \term{tt}},  \propV{q_{1}}{\term{ff}}, \ldots, \propV{q_m}{\term{ff}} \notag
\end{align}
Therefore, we can take a last step with this rule to add we can add $\propV{p_j}{\term{tt}}$. 
\paragraph{Second step}
The database $D_k$ we obtained in the first step might not be a solution,
because it might not be saturated. In 
particular, there may be open rules
that allow us to derive additional non-contradictory facts $\propV{p'}{\term{ff}}$.
This would 
be the case, for example, if our original program looked something like this:
\begin{align}
\pred{c} & \leftarrow \notag \\
\pred{b} & \leftarrow \pred{c} \notag \\
\pred{a} & \leftarrow \pred{b}, \neg \pred{d}, \neg \pred{c} \notag
\end{align}
In this case, $D_k$ would only contain $\propV{c}{\term{tt}}$
and $\propV{b}{\term{tt}}$. In order to get a solution database, 
$\propV{d}{\term{ff}}$ would need to be added.

First, we'll establish that no additional facts of the form $\propV{p}{\term{tt}}$ can be
added to $D_k$ or any consistent extension of $D_k$ that only adds facts of the form $\propV{q}{\term{ff}}$. 
We consider every rule 
\begin{align}
\propV{p}{\{ \term{tt} \}} & \leftarrow \propV{p_1}{\term{tt}}, \ldots, \propV{p_n}{ \term{tt}},  \propV{q_{1}}{\term{ff}}, \ldots, \propV{q_m}{\term{ff}} \notag
\end{align}
in the translated program.
There are three cases to consider:
\begin{enumerate}
\item If $\propV{p}{\term{tt}} \in D_k$, then the rule can't add any new facts.
\item If $\propV{p}{\term{tt}} \not\in D_k$ and the original ASP rule was
\textit{not} eliminated from the reduct $P^X$, then $\pred{p} \leftarrow \pred{p_1},\ldots, \pred{p_n}$ 
appears in the reduct. Because $\pred{p}$ is not in the model of the reduct,
there must be some premise $\pred{p_i} \not\in X$, meaning $\propV{p_i}{\term{tt}} \not\in D_k$, so the
rule doesn't apply and, furthermore, adding additional consistent facts of the form $\propV{q}{\term{tt}}$ will
not change that.
\item If $\propV{p}{\term{tt}} \not\in D_k$ and the original ASP rule \textit{was}
eliminated from the reduct $P^X$, then that rule has a premise $\propV{q_i}{\term{ff}}$ that
cannot possibly be satisfied in a consistent extension of $D_k$, because $\propV{q_i}{\term{tt}} \in D_k$. Therefore, the
rule doesn't apply and, furthermore, adding additional consistent facts of the form $\propV{q}{\term{ff}}$ will
not change that.
\end{enumerate}

Because no new facts of the form $\propV{p}{\term{tt}}$ can be derived, and every rule
that might introduce a $\propV{p}{\term{ff}}$ fact is an open rule, we can iteratively 
apply rules that derive new consistent $\propV{p}{\term{ff}}$ facts to $D_k$ until no additional
such facts can be consistently added. The result is a consistent, saturated database: a solution $D$
where $\propV{p}{\term{tt}} \in D$ if and only if $\pred{p} \in X$.
\end{proof}

\begin{proposition}\label{thm-complete-asp}
Let $P$ be an ASP program, and let $\langle P\rangle$ be the
interpretation of $P$ as a {\casp} as defined above. 
For all solutions $D$ of $\langle P\rangle$, the set $\{ \pred{p} \mid \propV{p}{\term{tt}} \in D \}$ is a stable model
of $P$.
\end{proposition}

\begin{proof}
Let $D$ be a solution to $\langle P \rangle$, and let $X$ be $\{ \pred{p} \mid \propV{p}{\term{tt}} \in D \}$. We can compute the reduct $P^X$, and that reduct $P^X$ has a unique (and finite) model $M$.
If $X = M$, then $X$ is a stable model of $P$.

\paragraph{First step: $M \subseteq X$}
$M$ is the model of $P^X$, and so by \Cref{lem-get-sequence} we can
get some $M_\mathsf{seq}$ containing exactly the elements of $M$, and where each element
in the sequence is the immediate consequence of the previous elements. We'll prove by induction
that for every prefix of $M_\mathsf{seq}$, the elements of that sequence are a subset of
$X = \{ \pred{p} \mid \propV{p}{\mathsf{tt}} \in D \}$. The base case is immediate,
because $\emptyset \subseteq X$.

When we want to extend our prefix with a new fact $\pred{p}$, we have that 
$\pred{p}$ was derived 
by a rule of the form  $\pred{p} \leftarrow \pred{p_1}, \ldots, \pred{p_n}$
in the reduct $P^X$,
which itself comes from a rule
$\pred{p} \leftarrow \pred{p_1}, \ldots, \pred{p_n}, \neg \pred{q_{1}}, \ldots, \neg \pred{q_m}$ in the program $P$.
By the induction hypothesis, $\propV{p_i}{\term{tt}} \in D$ for $1 \leq i \leq n$.

By virtue of the fact that this rule was in the reduct $P^X$, it must \textit{not} 
be the case that $\propV{q_i}{\term{tt}} \in D$ for $1 \leq i \leq m$. 
There are $m$ open rules of the form
\begin{align}
\propishV{q_{j}}{\term{ff}} & \leftarrow \propV{p_1}{\term{tt}}, \ldots, \propV{p_n}{\term{tt}},
\propV{q_{1}}{\term{ff}}, \ldots, \propV{q_{j-1}}{\term{ff}}\notag
\end{align}
for $1 \leq j \leq m$, and we can show by induction on $j$
that the presence of this rule, plus the fact that
$\propV{q_j}{\term{tt}} \not\in D$, that it must be the case that 
$\propV{q_j}{\term{ff}} \in D$, because otherwise $D$ would not be
saturated. 

This means that all the premises
of the closed rule deriving $\propV{p}{\term{tt}}$ are satisfied, so by the
 saturation of $D$ we must have $\propV{p}{\term{tt}} \in D$ and therefore $\pred{p} \in X$. 

\paragraph{Second step: $X \subseteq M$}
For every $\propV{p}{\term{tt}} \in D$, we must show that $\pred{p}$ is in the model of $P^X$.
We'll prove by induction on $j$ that, if we have a step sequence $\emptyset\ldots D_j$ and $D_j \subseteq D$, then $\{ \pred{p} \mid \propV{p}{\term{tt}} \in D_j \}$ is a subset 
of the model of $P^X$.

In the base case, we have null step sequence from $\emptyset$ to itself. It's necessarily the case that $\emptyset$ is a subset of the model of $P^X$.

When we extend the derivation by deriving a new fact
of the form $\propV{p}{\term{ff}}$,
we can directly use the induction hypothesis, because the set 
$\{ \pred{p} \mid \propV{p}{\term{tt}} \in D_j \}$ is the same as the set $\{ \pred{p} \mid \propV{p}{\term{tt}} \in D_{j-1} \}$.

When we extend the derivation with the rule
\begin{align}
\propV{p}{\{ \term{tt} \}} & \leftarrow \propV{p_1}{\term{tt}}, \ldots, \propV{p_n}{ \term{tt}},  \propV{q_{1}}{\term{ff}}, \ldots, \propV{p_m}{\term{ff}} \notag
\end{align}
then we need to show that $\pred{p}$ is in the model of $P^X$.
We know by the definition of the translation 
that a corresponding rule 
$\pred{p} \leftarrow \pred{p_1}, \ldots, \pred{p_n}, \neg \pred{q_{1}}, \ldots, \neg \pred{q_{m}}$ must appear in $P$. 

For $1 \leq i \leq m$, we have $\propV{q_i}{\term{ff}} \in D_{j-1}$ by the fact that the rule was applied, so its premises must have been satisfied.
Because $D_{j-1} \subseteq D$ we must have $\propV{q_i}{\term{ff}} \in D$, and because
$D$ is consistent, $\propV{q_i}{\term{tt}} \not\in D$, so $\pred{q_i} \not\in X$ for all $1 \leq i \leq m$. That means the rule 
$\pred{p} \leftarrow \pred{p_1}, \ldots, \pred{p_n}$ appears in the reduct $P^X$.

For $1 \leq i \leq n$, we have $\propV{p_i}{\term{tt}} \in D_{j-1}$ by the fact that the rule was applied, so its premises must have been satisfied. By 
the induction hypothesis these $\pred{p_i}$ must all be in the model of $P^X$. 

If the premises of a rule are in the model then the conclusion must be in the model, so $\pred{p}$ is in the
model of $P^X$.
\end{proof}

\section{Least Upper Bounds for Choice Sets}
\label{sec-choice-has-lubs}

The following theorems establish that the pointed partial order $\Choice$ is a complete lattice with the least upper bound defined in \reforexpanded{def-choice-lub}{Definition 4.11} as follows:
   Take any $\{ \C_i : i \in I \} \subseteq \Choice$, and let $\mathcal{F}$ be the set of functions 
   $f : I \rightarrow \DB$
   such that $f(i) \in \C_i$. (Equivalently, $\mathcal{F}$ is the direct product $\prod_{i\in I} C_i$.)
   We define
   $\bigvee \{\C_i : i \in I\}$
   as $\{ \bigvee \mathcal{D} : f \in \mathcal{F},\, \mathcal{D} = \{f(i) : i \in I\}, {\compatible} \mathcal{D}  \} $.
   
We must prove that this function that takes sets of choice
sets always produces a choice set, and that the result is actually a least
upper bound.

\begin{lemma}\label{thm-choice-lub-is-choice}
   $\bigvee_i \C_i$ is an element of $\Choice$, i.e.\ is pairwise incompatible.
\end{lemma}

\begin{proof}
    Consider databases $D_1,D_2 \in \bigvee_i \C_i$.
    These are the least upper bound of the sets of databases picked out by $f_1$ and $f_2$, respectively.
    If $D_1 \neq D_2$, then there is some index $i$ for which $f_1(i) \neq f_2(i)$, and $f_1(i)$ and $f_2(i)$ come from the same choice set $\C_i$, so they must be incompatible.
    By monotonicity of incompatibility (\reforexpanded{lemma-incompatibility-monotone}{Lemma 4.2}), $D_1$ and $D_2$ are also incompatible.
\end{proof}

\begin{lemma}\label{thm-choice-lub-is-ub}
  $\bigvee_{\Choice}$ is an upper bound in {\Choice}. That is,
  if $\C \in \mathcal{S}$ then $\C \leq \bigvee \mathcal{S}$.
\end{lemma}
\begin{proof}
  Let $\C_i \in \mathcal{S}$.
  Every element of $\bigvee \mathcal{S}$ is the join over some selection
  of compatible databases, exactly one of which is an element of $\C_i$. 
  That is, if $\mathcal{X}_j \in \bigvee \mathcal{S}$, then 
  $\mathcal{X}_j = \bigvee_i f_j(i)$ where $f_j(i) = D_i \in \C_i$.
  Thus $D_i \leq \mathcal{X}_j$, as required. 
\end{proof}

\begin{theorem}\label{thm-choice-lubs}
   $\bigvee_i \C_i$ is the least upper bound of $\{ \C_i : i \in I \}
   \subseteq \Choice$.
\end{theorem}

\begin{proof} 
    Let $\C'$ be an upper bound of 
    $\{ \C_i : i \in I \}$. 
    To show that $\bigvee_i \C_i \leq \C'$,
    it suffices to show that for every $D' \in \C'$, 
    there exists some $D \in \bigvee_i \C_i$ such that
    $D \leq D'$.
    
    Therefore, let $D'$ be an arbitrary element of $\C'$.
    For every $i \in I$, there is (by definition of $\leq_{\Choice}$)
    a unique $D_i \in \C_i$ such that 
    $D_i \leq D'$, and we can define
    $f : I \to \DB$ to pick out that $D_i$. The set of 
    databases $\{ f(i) : i \in I \}$ is compatible because $D'$ is an upper
    bound. That means that 
    $\bigvee_i \C_i$ contains $\bigvee_i f(i)$, and 
    because $\bigvee_\DB$ is already established as a least upper bound,
    $\bigvee_i f(i) \leq_\DB D'$. Letting $D = \bigvee_i f(i)$,
    establishes that $\bigvee_\Choice$ is also a least upper bound.
\end{proof}

\section{Equivalence of Semantics}
This section contains the proof that three semantics for {\casps}---the step sequences in \reforexpanded{sec-definition}{Section 2}, the least-fixed-point interpretation in \reforexpanded{sec-semantics}{Section 5}, and the algorithm in \reforexpanded{sec-implementation}{Section 6}---all agree.

These proofs relate sequences of fact sets and the least fixed-point of an immediate consequence operator on choice sets, sets of mutually exclusive constraint databases.
In the main body of the paper, we frequently conflated sets of facts and constraint databases. In this appendix, we will avoid this conflation, and will consistently use $D$ to refer to databases that are comprised of sets of facts as defined in \reforexpanded{sec-definition}{Section 2} but use $\Delta$ to refer to constraint databases as defined in \reforexpanded{sec-semantics}{Section 5}. Fact sets and constraint databases are connected by two adjoint functors, erasure and promotion.

\begin{definition}[Erasure of a Constraint Database to a Set of Facts]\label{def-erasure}
For all constraint databases $\Delta \in \DB$, define $(\Delta)^- = \{ a \is v \mid \Delta(a) = \just{v} \}$. This set is a consistent database of facts as defined in \reforexpanded{def-set-database}{Definition 2.6}.
\end{definition}

\begin{definition}[Promotion of a Set of Facts to a Constraint Database]\label{def-promotion}
If the set of facts $D$ is consistent by \reforexpanded{def-set-database}{Definition 2.6}, then for any $a$ there is at most one $v$ such that $a \is v \in D$. Therefore, we can define $(D)^+$ as a function from consistent sets of facts to constraint databases as follows:
\begin{itemize}
    \item $(D)^+(a) = \just{v}$ if $a \is v \in D$
    \item $(D)^+(a) = \none{\emptyset}$ if there does not exist a $v$ such that $a \is v \in D$
\end{itemize}
\end{definition}

Here are two examples of these definitions in action:
\begin{itemize}
    \item $(\{ \propV{p}{\term{tt}} \})^+ = (\pred{p} \mapsto \just{\term{tt}})$
    \item $(\pred{p} \mapsto \just{\term{tt}},\; \pred{q} \mapsto \none{\{ \term{a}, \term{b} \}})^- = \{ \propV{p}{\term{tt}} \}$
\end{itemize}

\begin{lemma}[Properties of Erasure and Promotion]\label{properties-of-erasure-and-promition}
The following all hold:
    \begin{itemize}
        \item For all consistent sets of facts $D$, it is the case that $(D)^+$ is positive and $((D)^+)^- = D$.
        \item For all $\Delta \in \DB$, it is the case that $((\Delta)^-)^+ \leq_\DB \Delta$.
        \item For all $\Delta \in \DB$, it is the case that $((\Delta)^-)^+ = \Delta$ if and only if $\Delta$ is positive.
    \end{itemize}
\end{lemma}
\begin{proof}
    Straightforward unfolding of definitions.
\end{proof}

With these definitions, we can precisely restate \reforexpanded{semantics-are-equivalent}{Theorem 5.12} without the imprecise punning of fact sets and constraint databases:

\begin{proposition}[Restatement of \reforexpanded{semantics-are-equivalent}{Theorem 5.12}]\label{semantics-are-equivalent-precise} 
For all $P$ and $D$, the following are equivalent:
    \begin{enumerate}
        \item There is a fact-set step sequence $\emptyset \ldots D$ and whenever $D \setstepA{P} S$ then $S = \{ D \}$.
        \item $(D)^+ \in \lfp \bigimmcons{P}$ and $D$ is finite.
    \end{enumerate}
\end{proposition}
\subsection{Correspondence of the Sequential, Algorithmic, and Least-Fixed-Point Semantics}\label{sec-algorithm-round-trip}

The algorithm described in \reforexpanded{sec-algorithm}{Section 6.1}
defines another kind of sequence, which we can call an 
\textit{algorithmic step sequence} to distinguish it from a fact-set step sequence.

The definition of algorithmic steps takes advantage of the fact that all rule consequences are singular in some attribute $a$, which lets us define {\em attribute-specific} immediate consequence $\immcons{P[a]}(\Delta)$, mapping an attribute $a$ to
a choice set, the elements of which each map $a$ to a constraint (by the definition of choice sets, those constraints are pairwise-incompatible):
\begin{definition}[Attribute-Specific Immediate Consequence]

\[
\immcons{P[a]}(\Delta) = 
 \bigvee 
 \left\{ 
    \langle \sigma H \rangle 
        : (H \leftarrow F) \in P, 
          \sigma F \leq \Delta, 
          \sigma H @ a
\right\}
\]

Where $\sigma H @ a$ means $\sigma H = \propV{a}{S}$
or $\sigma H = \propishV{a}{v}$.

\end{definition}
In \reforexpanded{sec-algorithm}{Section 6.1} we wrote $\immcons{P}(D)[a]$ for the concept we are precisely defining as $\immcons{P[a]}(D)$. The action of the algorithm described in \reforexpanded{sec-algorithm}{Section 6.1} can be characterized in terms of $\immcons{P[a]}$ as follows:

\begin{definition}[Algorithmic steps]
    \label{def-algorithmic-steps}

    Let $\Delta$ be a constraint database.  
    We say that the program $P$ allows the constraint database $\Delta$ to make an \textit{algorithmic step} to $\Delta' \in \left(\{ \Delta \} \vee \immcons{P[a]}(\Delta)\right)$ whenever $\immcons{P}(\Delta)$ is nonempty.

    Algorithmic steps give rise to algorithmic step sequences: $\Delta_0\ldots \Delta_k$ is an algorithmic step sequence for $P$ if, for each $i > 0$, the program $P$ allows $\Delta_{i-1}$ to take an algorithmic step to $\Delta_i$. (As with fact-set step sequences, algorithmic step sequences are assumed to be finite.)
\end{definition}

These definitions allow us to break our central theorem into several lemmas:

\begin{lemma}[Completeness of the Algorithm]
    \label{prop-algorithm-complete}
    If there is a fact-set step sequence $\emptyset\ldots D$ under $P$ and $D$ is
    saturated by \reforexpanded{def-set-saturation}{Definition 2.10}, then there is an algorithmic
    step sequence $\bot_\DB \ldots (D)^+$ under $P$ and $(D)^+$ is a model by~\reforexpanded{def-model}{Definition 5.2}.
\end{lemma}

\begin{proof}
    See \Cref{sec-algorithm-complete}.
\end{proof}

\begin{lemma}[Soundness of the Algorithm]
    \label{prop-algorithm-sound}
    If there is an algorithmic step sequence $\bot_\DB\ldots\Delta$  and $\Delta$ is a model of $P$ by \reforexpanded{def-model}{Definition 5.2}, then $\Delta \in \lfp \bigimmcons{P}$.
\end{lemma}

\begin{proof}
    See \Cref{sec-algorithm-sound}.
\end{proof}

\begin{lemma}[Reachability of Finite Fixed Points]
    \label{prop-finite-fixed-points-reachable}
    If $\Delta \in \lfp \bigimmcons{P}$ and $\Delta$ is finite, then there is a fact-set step sequence $\emptyset\ldots(\Delta)^-$. 
\end{lemma}

\begin{proof}
    See \Cref{sec-finite-fixed-points-reachable}.
\end{proof}

\begin{lemma}\label{prop-positive-model-is-saturated}
If $\Delta$ is a positive model, $(\Delta)^-$ is saturated (\reforexpanded{def-set-saturation}{Definition 2.10}).
\end{lemma}
\begin{proof}
    Inspection of definitions: given a positive constraint database $\Delta$, we can show that any non-trivial evolution of $(\Delta)^-$ would contradict the assumption that $\Delta$ is in $\lfp \bigimmcons{P}$. (This would not be true if $\Delta$ were not positive.)
\end{proof}

These four lemmas suffice to prove \reforexpanded{semantics-are-equivalent}{Theorem 5.12}/\Cref{semantics-are-equivalent-precise} and, in a different configuration, the correctness of the algorithm, \reforexpanded{thm-algorithm-correct}{Theorem 6.2}.

The algorithm in \reforexpanded{sec-algorithm}{Section 6.1} is really an algorithm for enumerating all finite minimal models, not just for finding the positive models that correspond to solutions. However, \reforexpanded{thm-algorithm-correct}{Theorem 6.2} only talks about positive models. We conjecture the following, which is just \reforexpanded{thm-algorithm-correct}{Theorem 6.2} with the word ``positive'' removed, but the proof requires defining a constraint-database analogue to fact-step step sequences which we have not included.

\begin{conjecture}
    For constraint databases $\Delta$, the following are equivalent:
    \begin{enumerate}
        \item The algorithm described in \reforexpanded{sec-algorithm}{Section 6.1} may successfully return $\Delta$.
        \item $\Delta \in \lfp \bigimmcons{P}$ and $\Delta$ is finite.
    \end{enumerate}
\end{conjecture}

\subsection{Completeness of the Algorithm}\label{sec-algorithm-complete}

We will start with a warm-up lemma:

\begin{lemma}\label{lem-immcons-fact-step}
If $\immcons{P}((D_0)^+) = \emptyset$ and $D_0\ldots D_n$ is a fact-step step sequence, then $D_n$ is not saturated. 
\end{lemma}
\begin{proof}
By induction on $n$, in each $D_i$ one of these two conditions apply:
\begin{itemize}
    \item There is an attribute $a$ and two rule groundings with satisfied premises in $D_i$ and conclusions $a \is \{ v_1, \ldots, v_n \}$ and $a \is \{ w_1, \ldots, w_m \}$ where the $v_i$ and $w_i$ do not overlap.
    \item $\{ a \is v \} \in D_i$ and there is a rule grounding with satisfied premises in $D_i$ and conclusion $a \is \{ w_1, \ldots, w_m \}$, where $v$ is not one of the $w_i$.
\end{itemize}
Either of these cases mean that $D_i$ cannot be saturated.
\end{proof}

Now we can proceed to prove \Cref{prop-algorithm-complete}: if there is a fact-set step sequence $\emptyset\ldots D$ under $P$ and $D$ is
    saturated by \reforexpanded{def-set-saturation}{Definition 2.10}, then there is an algorithmic
    step sequence $\bot_\DB \ldots (D)^+$ under $P$ and $(D)^+$ is a model of $P$ by ~\reforexpanded{def-model}{Definition 5.2}.

By induction, we will prove that if there is a fact-set step sequence $D_1\ldots D_k$ under $P$, then either $D_k$ cannot be extended to a saturated set of facts --- there is no solution $D'$ such that $D_k\ldots D'$ is a saturated --- or $(D_1)^+\ldots (D_k)^+$ is an algorithmic step sequence under $P$.  

The base case is immediate: given a fact-set step sequence containing no steps, provide an algorithmic step sequence containing no steps. 

In the inductive case, we have a fact-set step sequence $D_1\ldots, D_k, D_{k+1}$ and we know one of two things is true: either $D_k$ cannot be extended to a solution, or else $(D_1)^+\ldots (D_k)^+$ is an algorithmic step sequence under $P$. We need to show either that $D_{k+1}$ cannot be extended to a solution or else that $(D_1)^+\ldots (D_{k+1)})^+$ is an algorithmic step sequence under $P$.

If $D_k$ can't be extended to a solution neither can $D_{k+1}$, satisfying the first disjunctive option in the induction hypothesis.

If $D_k = D_{k+1}$ then $(D_1)^+\ldots (D_k)^+$ is the same thing as $(D_1)^+\ldots (D_{k+1})^+$, and we're done. Otherwise, by inspecting the definitions we see $D_{k+1} = \{ a \is v \} \cup D_{k}$ for some $a$ and $v$. We identify three cases:
\begin{itemize}
    \item If $\immcons{P}((D_k)^+) = \emptyset$, the first disjunctive option in the induction hypothesis is satisfied by \Cref{lem-immcons-fact-step}.
    \item If $\immcons{P}((D_k)^+) \neq \emptyset$ and $(a \mapsto \just{v}) \in \immcons{P[a]}((D_k)^+)$, then $(D_{k+1})^+ \in \{ (D_k)^+ \} \vee \immcons{P[a]}((D_k)^+)$. This means $(D_k)^k$ can take an algorithmic step to $(D_{k+1})^+$, satisfying the second disjunctive option in the induction hypothesis.
    \item If $\immcons{P}((D_k)^+) \neq \emptyset$ and $(a \mapsto \just{v}) \not\in \immcons{P[a]}((D_k)^+)$, then we must be able to identify some closed rule $H \leftarrow F$ in $P$ such that some $\sigma$ satisfies $F$ in $(D_k)^+$,  where $\sigma H = a \is \{ v_1, \ldots v_n \}$, and where $v$ is
    not one of the $v_i$. This implies that in $D_{k+1}$, or in any extension of it, 
    evolves to the empty set (\reforexpanded{def-set-evolution}{Definition 2.8}) and so therefore cannot
    be a solution or extended to a solution, satisfying the first disjunctive option in the induction hypothesis.
\end{itemize}

Applying this proof by induction to the given fact-set step sequence, we obtain from the fact-step step sequence $\emptyset\ldots D$ an algorithmic step sequence $\bot_\DB\ldots (D)^+$, because it would be an immediate contradiction to say that the saturated database $D$ could not be extended to a saturated database. We complete the proof by observing that, if $D$ is saturated by \reforexpanded{def-set-saturation}{Definition 2.10} then $(D)^+$ is a model.

\subsection{Soundness of the Algorithm}\label{sec-algorithm-sound}

Here we prove \Cref{prop-algorithm-sound}: if there is an algorithmic step sequence on $P$ from $\bot_{\DB}$ to $\Delta$ and $\Delta$ is a model of $P$, then $\Delta \in \lfp \bigimmcons{P}$.
Since $\lfp \bigimmcons{P}$ is the set of minimal models of $P$ (\reforexpanded{thm-exactly-the-minimal-models}{Theorem 5.9}), to prove \Cref{prop-algorithm-sound} it suffices to show $\Delta$ is minimal, which we show in \Cref{lemma:alg-sound-model}:



\begin{lemma}\label{lemma:alg-sound-model}
    Let $\bot_{\DB} \ldots \Delta$ be an algorithmic step
    sequence for a program $P$, with $\Delta$ a model of $P$. Then $\Delta$ is minimal: if $M$ is a model and $M \le \Delta$ then $M = \Delta$.
\end{lemma}

\begin{proof}
    We will show inductively that each step $\Delta_i$ in the sequence $\bot_{\DB} \ldots \Delta$ is below $M$, and so in particular $\Delta \leq M$.
    %
    %

    {\bf Base case:} $\Delta_0 = \bot_{\DB}$. We have $\bot_{\DB} \leq M$ by definition.

    \newcommand\flubber{\immcons{P[a]}(\Delta_i)}

    {\bf Inductive case:} 
    Given $\Delta_i \leq M$, we need to show $\Delta_{i+1} \leq M$.
    By \cref{def-algorithmic-steps}, we know $\Delta_{i+1} \in \{\Delta_i\} \vee \immcons{P[a]}(\Delta_i)$.
    As a first step, we show $\flubber \le \{M\}$:
    \begin{align*}
        \flubber
        &= \bigvee \{ \langle \sigma H \rangle :
             (H \leftarrow F) \in P,\, \sigma F \le \Delta_i,\, \sigma F @ a \}
        \\
        &\le \bigvee \{ \langle \sigma H \rangle :
             (H \leftarrow F) \in P,\, \sigma F \le \Delta_i \}
        \\
        &\le \bigvee \{ \langle \sigma H \rangle :
             (H \leftarrow F) \in P,\, \sigma F \le M \}
        && \text{since }\Delta_i \le M
        \\
        &\le \immcons{P}(M)
        && \text{definition of }\immcons{P}
        \\
        &= \{M\}
        && \text{since }M\text{ is a model}
    \end{align*}

    \noindent
    This shows $\{M\}$ upper bounds $\flubber$; it is also an upper bound of $\{\Delta_i\}$, since $\Delta_i \le M$.
    Therefore it is above their least upper bound: $\{\Delta_i\} \vee \flubber \le \{M\}$.
    By definition of $\le_\Choice$ this means there is some $\Delta' \in \{\Delta_i\} \vee \flubber$ with $\Delta' \le M$.
    To show $\Delta_{i+1} \le M$ it suffices to show that this $\Delta'$ is in fact $\Delta_{i+1}$.
    Since they are elements of the choice set $\{\Delta_i\} \vee \flubber$, this follows if they are compatible.
    And indeed, our final destination $\Delta$ upper bounds both; the first because $\Delta' \le M \le \Delta$; the second because the step sequence is monotonically increasing, $\bot_\DB \le ... \le \Delta_{i+1} \le ... \le \Delta$.
\end{proof}

\subsection{Reachability of Finite Fixed Points}\label{sec-finite-fixed-points-reachable}

Here we prove \Cref{prop-finite-fixed-points-reachable}: If $\Delta \in \lfp \bigimmcons{P}$ and $\Delta$ is finite, then there is a fact-set step sequence $\emptyset\ldots(\Delta)^-$. 

\begin{lemma}
    If $\Delta \in \lfp \bigimmcons{P}$ and $\Delta$ is finite, then there is a $n$ such that $\Delta \in \bigimmcons{P}^n(\bot)$
\end{lemma}

\begin{proof}
For all $i$, it is the case that $\bigimmcons{P}^i(\bot) \leq \lfp \bigimmcons{P}$, so we can define an infinite sequence of $\Delta_i$ such that each $\Delta_i \leq \Delta$ and $\Delta_i \in \bigimmcons{P}^i(\bot)$.

If we ever have a $i$ such that $\Delta_i = \Delta_{i+1}$, then $\Delta_i \in \lfp \bigimmcons{P}$, meaning that $i$ is the $n$ we are looking for. (Proof: $\Delta_i$ is a model, so $\Delta \leq \Delta_i$, and because $\bigimmcons{P}^i(\bot) \leq \lfp \bigimmcons{P}$, $\Delta_i \leq \Delta$.)

Therefore, to show that there exists an $n$ such that $\Delta \in \bigimmcons{P}^n(\bot)$, it suffices to show that the sequence of $\Delta_i$ cannot be be strictly ascending without limit. 

If $\Delta_{i-1}$ and $\Delta_{i}$
agree on all attributes where $\Delta_{i}[a] = \just{t}$, then $\Delta_{i} = \Delta_{i+1}$, so an infinite ascending sequence implies a that each $\Delta_{i}$ would have to contain at least $i$ attributes that map to values $\just{t}$. Because $\Delta$ is finite and $\Delta_i \leq \Delta$, this is impossible, so there must be the required finite $n$.
\end{proof}

\begin{lemma}
    If $\Delta_i$ is finite and takes an algorithmic step to $\Delta_{i+1}$, then $(\Delta_i)^-\ldots(\Delta_{i+1})^-$ is a fact-set step sequence. 
\end{lemma}

\begin{proof}
    This amounts to showing two things. First, we observe that the same substitutions satisfy the same premises according to \reforexpanded{def-set-satisfaction,def-db-satisfaction}{Definitions 2.7 and 4.12} (the extra information removed in the erasure transformation isn't relevant to which rules apply). Second, we observe that there are a finite number of relevant rule heads $\sigma H$ in any $\immcons{P}(\Delta)$, a straightforward consequence of the finite size of the database and the restriction in \reforexpanded{def-rule}{Definition 2.3} that all variables in the head of a rule appear in a premise.
\end{proof}

The proof of \Cref{prop-finite-fixed-points-reachable} follows from these two lemmas: we string together the required fact-set step sequence by concatenating finite step sequences between each of the finitely-many relevant iterations of the immediate consequence operator.

\section{Examples and Benchmarks for {\dusa}}\label{sec-benchmarking-and-examples}

We compare {\dusa} against two existing systems:
Clingo, part of the Postdam Answer Set Solving Collection (Potassco)
and generally recognized as the industry standard
 \cite{gebser2011potassco}, 
and Alpha, a implementation of answer set programming with
lazy grounding \cite{alpha}.
As discussed in \reforexpanded{lazy-asp-in-casp,sec-dusa-and-asp-perf}{Sections 3.3 and 7.2}, lazy-grounding systems are, like {\dusa}, able to handle
problems that pre-grounding systems like Clingo cannot, and are able to more efficiently compute solutions in cases where the fully-ground program is quite large (this is called the ``grounding bottleneck'').

There are three other often-cited lazy-grounding ASP solvers:
Gasp \cite{gasp}, ASPeRiX \cite{asperix}, and Omiga \cite{omiga}.
We were unable to find source code for Gasp,
we were unable to successfully
compile ASPeRiX, and we understand Alpha to be a successor to Omiga and so did not attempt to evaluate it.

All the results in this section were generated on AWS m7g.xlarge instances (64-bit 4-core virtual machines with 16 GB of memory running Ubuntu Linux) with processes limited to 10 GB of memory and 100 seconds of wall-clock time. This paper's artifact contains the benchmarking code \cite{simmons_2024_13921177}; both the raw data we generated and the spreadsheet used to analyize that data are available as well \cite{dusabenchmarkdatasets}.
We discuss three types of benchmarks:

\begin{itemize}
\item In \Cref{sec-examples-from-paper}, we cover spanning tree generation and canonical representative identification, examples from \reforexpanded{sec-examples}{Section 3} well-suited for {\casping} and the {\dusa} implementation.
\item In \Cref{sec-examples-sad-trombone}, we cover N-queens and Adam Smith's ``Map Generation Speedrun,'' pure search problems that are expressible in {\casping} but that demonstrate the kinds of problems that the current {\dusa} implementation is ill-suited for.
\item In \Cref{sec-alpha-benchmarks}, we replicate a set of benchmarks designed for the Alpha implementation of answer set programming with lazy grounding.
\end{itemize}

\paragraph{Concrete {\dusa} Syntax}

The {\dusa} examples in this section use the concrete syntax 
from the implementation, which follows the {\casping} notation from
the paper
closely. As is common in logic programming languages,
variables are distinguished by being capitalized
and the arrow $\leftarrow$ is replaced by ``\verb|:-|.'' Less commonly,
{\dusa} uses an un-curried logic programming notation, so we write
``\verb|color X Y is red|'' to express a premise we would have 
written in this paper as $\prop{color}{x,y}{\term{red}}$. In addition:

\begin{itemize}
    \item 
The implementation directly supports attributes with no values:
writing ``\verb|edge X Y|'' is equivalent to
$\prop{edge}{x,y}{\term{unit}}$ as motivated in
\reforexpanded{sec-simulating-datalog}{Section 2.2}. 
    \item 
When the
conclusion of a closed rule only has one value, we are allowed
to optionally omit the curly-braces: we can write
``\verb|parent X X is tt|''
for the head $\prop{parent}{x,x}{\{ \term{tt} \} }$, 
though $\prop{parent}{x,x}{\{ \term{tt}, \term{ff} \} }$ requires 
curly braces and will be written ``\verb|parent X X is { tt, ff }|''.
    \item
The implementation supports \verb`#forbid` rules, a list of premises which cannot be simultaneously satisfied in a valid solution. The rule ``\verb`#forbid p X, q X`'' 
in {\dusa}'s concrete syntax can be desugared into the rule $(\propV{ok}{\{ \term{no} \}} \leftarrow \propK{p}{x}, \propK{q}{x})$ with the addition a new predicate $\pred{ok}$ and rule $(\propV{ok}{\{ \term{yes} \}} \leftarrow)$.
\end{itemize}

\paragraph{Translating Answer Set Programs to {\CASPS}}

As discussed in \reforexpanded{sec-dusa-and-asp-perf}{Section 7.2}, we can perform fair head-to-head comparisons between ASP solvers and {\dusa} by using the translation from answer set 
programming to {\casping} presented in \Cref{sec-simulating-asp-details}. We will extend that translation in two ways. First, ASP features
\emph{constraints},
in the form of headless rules ($\leftarrow F$), which  
translate directly to the \verb|#forbid| rules introduced above. Second, ASP has is \emph{choice rules}, which 
allow the ASP engine to treat a predicate as either true or false.
A choice rule written in ASP as ``\verb|{p} :- q, not r|''
can be represented in {\casping} with the rule
$(\propV{p}{\{ \term{tt}, \term{ff} \}} \leftarrow \propV{q}{\term{tt}}, \propV{r}{\term{ff}})$. The general correctness of this extended translation is left for future work (as is the general correctness of the translation on open programs, as discussed in \reforexpanded{lazy-asp-in-casp}{Section 3.3}).

\subsection{Examples from the Paper}\label{sec-examples-from-paper}

Spanning tree generation and canonical representative identification as discussed in \reforexpanded{sec-examples}{Section 3} are both cases where {\casping} allows the problem to be described quite effectively. ASP implementations of the same algorithm are less consistent and less efficient, as we see in this section, which expands on the discussion in \reforexpanded{sec-perf-eval,sec-dusa-and-asp-perf}{Sections 7.1 and 7.2}.

\subsubsection{Spanning Tree Generation}

\begin{figure}
\begin{centering}
\footnotesize
\begin{minipage}[b]{0.4\textwidth}
\begin{verbatim}
edge(X,Y) :- edge(Y,X).

% Exactly one root
{root(X)} :- node(X).
:- root(X), root(Y), X != Y.
someRoot :- root(X).

:- not someRoot.

% The root has itself as a parent
parent(X,X) :- root(X).

% Any tree node can be a parent
inTree(P) :- parent(P,_).
{parent(X,P)} :- 
    edge(X,P), inTree(P).

% Only 1 parent
:- parent(X,P1), 
    parent(X,P2), P1 != P2.

% Tree covers connected component


:- edge(X,Y), inTree(X),
    not inTree(Y).
\end{verbatim}    
\end{minipage}
\begin{minipage}[b]{0.4\textwidth}
\begin{verbatim}
edge X Y :- edge Y X.

# Exactly one root
root X is { tt, ff } :- node X.
#forbid root X is tt, root Y is tt, X != Y.
someRoot is tt :- root X is tt.
someRoot is? ff.
#forbid someRoot is ff.

# The root has itself as a parent
parent X X is tt :- root X.

# Any tree node can be a parent
inTree P is tt :- parent P _ is tt.
parent X P is { tt, ff } :- 
    edge X P, inTree P.

# Only 1 parent
#forbid parent X P1 is tt,
     parent X P2 is tt, P1 != P2.

# Tree covers connected component
inTree Y is? ff :- 
    edge X Y, inTree X is tt.
#forbid edge X Y, inTree X is tt, 
    inTree Y is ff.    
\end{verbatim}
\end{minipage}
\end{centering}
\caption{Spanning-tree description in ASP, and direct translation to {\dusa}.}
\label{fig-spanning-tree-in-asp}
\end{figure}

A {\casp} for spanning tree generation is presented in \reforexpanded{sec-rooted-spanning-tree}{Section 3.1}, and a pure ASP program for solving the same problem, as well as that program's translation into {\dusa}'s concrete syntax, is given in \Cref{fig-spanning-tree-in-asp}. The best solution we know of in idiomatic Clingo ASP replaces the four rules in \Cref{fig-spanning-tree-in-asp} labeled ``Exactly one root'' with the single cardinality constraint \verb|1 {root(X) : edge(X,Y)} 1.|

\begin{figure}
\includegraphics[trim={2.5cm 16.5cm 4cm 3.1cm},clip,width=13.8cm]{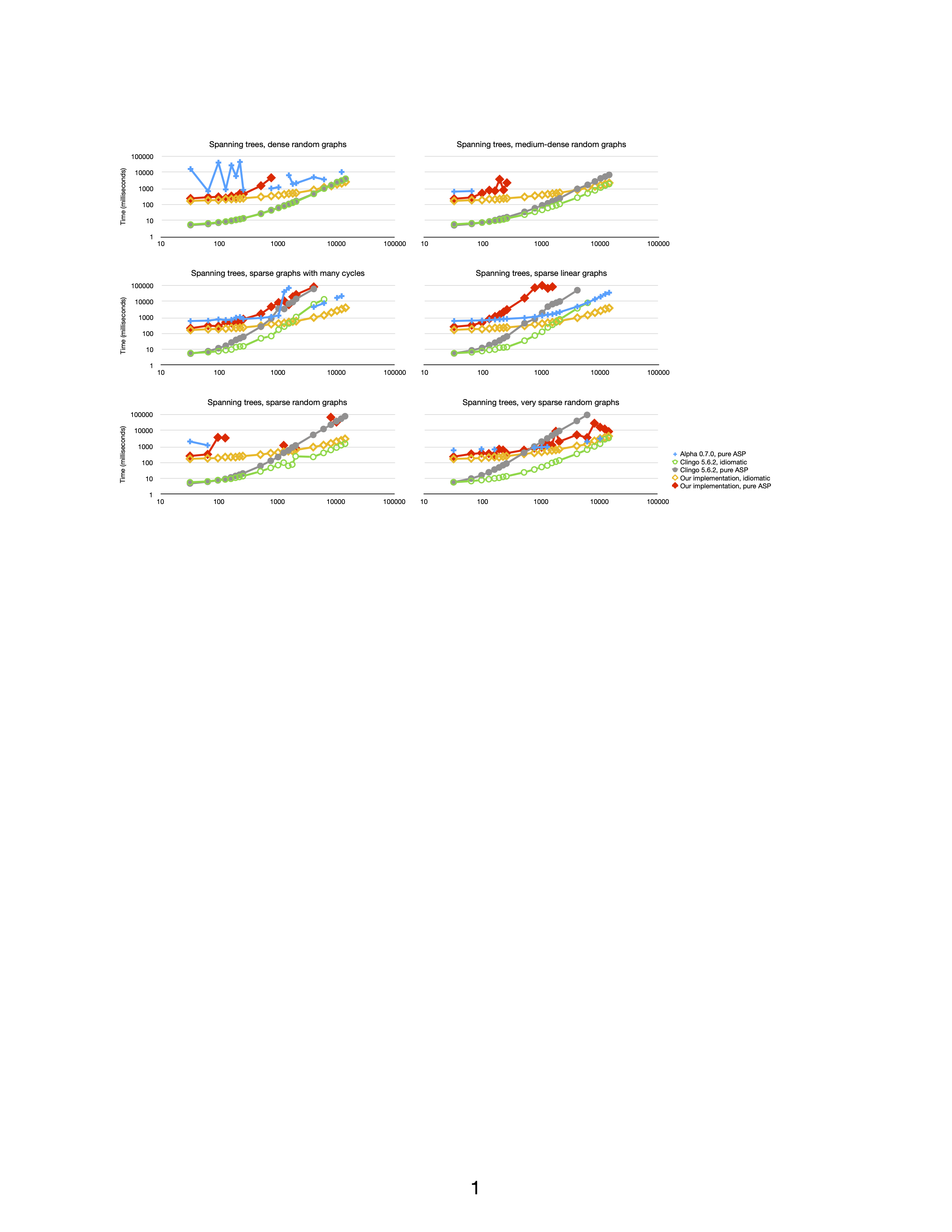}
\caption{Performance of {\dusa}, Clingo, and Alpha on finding spanning trees of graphs with different features. The $x$-axis is the number of edges in the input graph. All data points are the median of three runs requesting a single solution, discarding timeouts.}
\label{spanning-tree-analysis}
\end{figure}

These programs were tested on a graphs with different 
properties, with 
results shown in \Cref{spanning-tree-analysis}. In most cases, {\dusa} is outperformed by the idiomatic Clingo ASP solution on smaller graphs, but Clingo's performance seems to grow faster than linearly. In a fair fight, {\dusa} and Alpha outperform Clingo on some sparse graphs, which is what we would expect if Clingo's performance issues were connected to the ground-then-solve methodology.

\subsubsection{Appointing Canonical Representatives}

\begin{figure}
\begin{centering}
\footnotesize
\begin{minipage}[b]{0.4\textwidth}
\begin{verbatim}
edge(Y,X) :- edge(X,Y).

% Any node can be a rep
{representative(X,X)} :-
    node(X).

% Connected nodes have same rep
representative(Y,Rep) :-
    edge(X,Y),
    representative(X,Rep).

% Representatives must be unique
:- representative(X,R1), 
    representative(X,R2), R1 != R2.

% Every node has a rep
hasRep(X) :-
   representative(X,_).
   
:- node(X), not hasRep(X).
\end{verbatim}    
\end{minipage}
\begin{minipage}[b]{0.4\textwidth}
\begin{verbatim}
edge Y X :- edge X Y.

# Any node can be a rep
representative X X is { tt, ff } :-
    node X.

# Connected nodes have same rep
representative X Rep is tt :-
    edge X Y,
    representative Y Rep is tt.

# Representatives must be unique
#forbid representative X R1 is tt, 
    representative X R2 is tt, R1 != R2.

# Every node has a rep
hasRep X is tt :-
    representative X _ is tt.
hasRep X is? ff :- node X.
#forbid node X, hasRep X is ff.
\end{verbatim}
\end{minipage}
\end{centering}
\caption{Canonical representative identification in ASP, and direct
translation to {\dusa}.}
\label{fig-canonical-rep-in-asp}
\end{figure}

The idiomatic {\casping} approach to appointing canonical representatives for
connected components of an undirected graph is presented in 
 \reforexpanded{sec-canonical-representatives}{Section 3.2}, and a pure ASP program
 for solving the same program, as well as that program's translation into
 {\dusa}'s concrete syntax, is given in
\Cref{fig-canonical-rep-in-asp}. As in the previous section, it's possible to get a shorter and more performant Clingo program by using Clingo-supported cardinality constraints. In this case, the last three ASP rules in \Cref{fig-canonical-rep-in-asp} can be replaced by the single Clingo rule:
\begin{center}
\verb|:- node(X), not 1{representative(X,R) : node(R)}1.|
\end{center}

For this problem, idiomatic
{\dusa} is the clear winner in all cases. Both {\dusa} and Alpha outperform Clingo in almost all cases when running the same ASP program, as discussed in \reforexpanded{sec-dusa-and-asp-perf}{Section 7.2}. This is not terribly surprising:
in a graph with $v$ nodes, Clingo's grounder is almost unavoidably going to 
require $O(v^3)$ groundings of the key rule that requires connected nodes to have
the same representative. It was more unexpected that, on all but the sparsest graphs, {\dusa} 
outperformed Alpha in a head-to-head matchup despite the {\dusa} implementation exerting almost no control over backtracking. Extremely sparse random graphs were a notable exception that caused {\dusa} to time out on all but the very smallest examples.

\begin{figure}
\includegraphics[trim={2.5cm 16.5cm 4cm 3.1cm},clip,width=13.8cm]{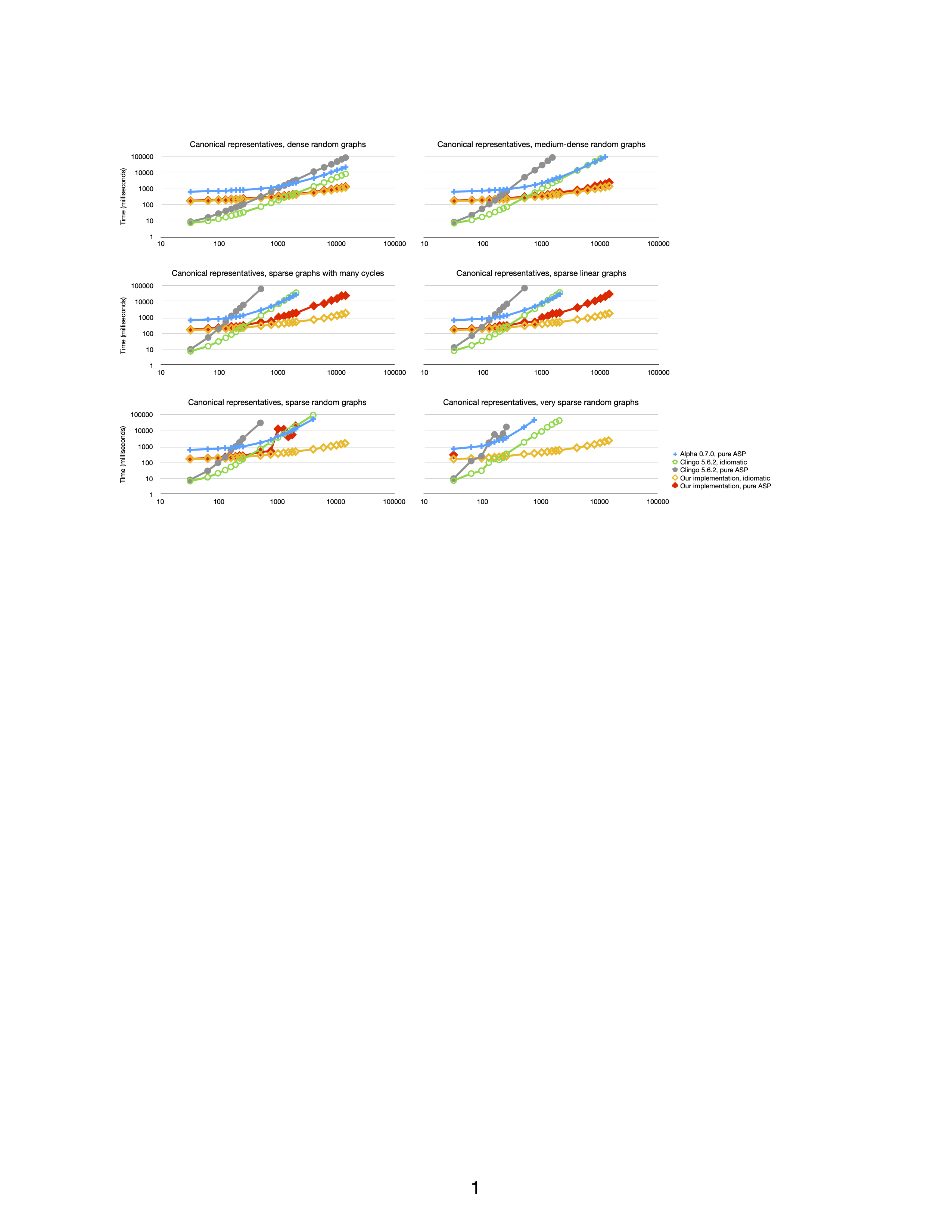}
\caption{Performance of {\dusa}, Clingo, and Alpha on finding canonical representatives for the connected components of graphs with different features. The $x$-axis is the number of edges in the input graph. All data points are the median of three runs requesting a single solution, discarding timeouts.}
\label{canonical-rep-analysis}
\end{figure}

\subsection{Examples Ill-Suited to the Current {\dusa} Implementation}\label{sec-examples-sad-trombone}

In this section we compare {\dusa} to Clingo on two problems where we expect {\dusa} to perform poorly, because they require brute-force search. These represent low-hanging fruit for optimization of {\dusa}.

\subsubsection{N-Queens}

The N-queens problem is a well-known chess puzzle. The problem is to place N queens on an N-by-N chessboard such that no two queens can immediately attack each other. We used the 2014 notes ``Basic modeling and ASP and more via the n-Queens puzzle'' by Torsten Schaub as a guide to the ASP implementations in this sections \cite{schaub2013modeling}. Schaub's solutions are not pure ASP, so we can't just turn a crank to translate them to {\dusa}. Instead, we'll consider a spectrum of two different Clingo solutions and three different {\casping} solutions.

\paragraph{Basic Clingo Implementation}

\begin{figure}
\begin{centering}
\footnotesize
\begin{minipage}[b]{0.8\textwidth}
\begin{verbatim}
{ queen(1..n,1..n) }.

 :- not { queen(I,J) } == n.
 :- queen(I,J), queen(I,JJ), J != JJ.
 :- queen(I,J), queen(II,J), I != II.
 :- queen(I,J), queen(II,JJ), (I,J) != (II,JJ), I-J == II-JJ.
 :- queen(I,J), queen(II,JJ), (I,J) != (II,JJ), I+J == II+JJ.
\end{verbatim}
\end{minipage}
\end{centering}
\caption{Basic N-queens program in Clingo's dialect of answer-set programming.}
\label{fig-basic-n-queens}
\end{figure}

The most basic Clingo implementation of N-Queens is given in \Cref{fig-basic-n-queens}. Reading the rules in order, they express:

\begin{itemize}
    \item Each spot in an N-by-N grid can have a queen or not.
    \item There must be exactly N queens.
    \item No two queens can share the same column.
    \item No two queens can share the same row.
    \item No two queens can share the same diagonal (in either direction).
\end{itemize}

This program has the characteristics of the \emph{design-space sculpting} approach described by Smith and Mateas \shortcite{smith2011answer}. In the design-space sculpting approach, a possibility space is described in the most general terms, and then constraints are used to carve out out parts of the problem space that are not wanted. 

\paragraph{An Aside on Externally-Computed Predicates in {\dusa}}

\begin{figure}
\begin{centering}
\footnotesize
\begin{minipage}[b]{0.8\textwidth}
\begin{verbatim}
#builtin INT_MINUS minus
#builtin INT_PLUS plus
dim N :- size is N.
dim (minus N 1) :- dim N, N != 1.
# Expands to: dim Nminus1 :- dim N, N != 1, minus N 1 is Nminus1. 
\end{verbatim}
\end{minipage}
\end{centering}
\caption{{\dusa} declarations for using arithmetic and making a \texttt{dim} relation containing the numbers from 1 to declared \texttt{size}, inclusive.}
\label{fig-prelim-n-queens}
\end{figure}

Our implementation doesn't support the range notation \verb|1..n| shown in \Cref{fig-basic-n-queens}, but this is simple to emulate as a relation \verb|dim| that contains all the numbers from 1 to N. This is achieved by the code in  \Cref{fig-prelim-n-queens}, which is implicitly included in subsequent {\casps} in this section. \Cref{fig-prelim-n-queens} demonstrates the use of integers and integer operations in {\dusa}, which is supported by our implementation and our theory. The theory of 
{\casping} supports treating externally-computed predicates as infinite relations. The most fundamental such relations are equality and inequality on terms and comparison of numbers, and these have built-in syntax in the {\dusa} implementation. Premises like ``\verb|4 > X|'', ``\verb|N != 1|'', or ``\verb|P == pair X Y|'' can be thought of as referring to infinite relations: in the language of the paper, we would write $\propK{gt}{\term{4}, x}$, $\propK{neq}{n, \term{1}}$, or $\propK{eq}{p, \term{pair}(x, y)}$, respectively. 

Our theory and implementation are able to handle operating on infinite relations under two conditions:
\begin{enumerate}
    \item Infinite relations are never present in the conclusion of a rule, and
    \item It is tractable to find all the substitutions that satisfy a premise. 
\end{enumerate}
In the implementation, conservative syntactic checks prevent us from writing rules like $\propK{p}{x} \leftarrow \propK{gt}{\term{4}, x}$ or $\propK{p}{x, y} \leftarrow \propK{eq}{x,y}$ that might have infinitely many satisfying substitutions.

Additional built-in operations like addition and subtraction are included with \verb|#builtin| directives. These are also treated like infinite relations. \Cref{fig-prelim-n-queens} also demonstrates an additional language convenience that we picked up from LogiQL \cite{aref15logicblox}: the language supports treating built-in relations with functional dependencies like functions, so we can write the conclusion ``\verb|dim (minus N 1)|'' instead of needing to write ``\verb|dim Nminus1|'' and including ``\verb|minus N 1 is Nminus1|'' as an additional premise.

\paragraph{N-Queens in {\dusa}}

\begin{figure}
\begin{centering}
\footnotesize
\begin{minipage}[b]{0.35\textwidth}
\begin{verbatim}
location N is? (tup X Y) :-
  dim N, dim X, dim Y.



#forbid 
  location N is (tup X _), 
  location M is (tup X _), 
  N != M.
#forbid 
  location N is (tup _ Y), 
  location M is (tup _ Y),
  N != M.
#forbid 
  location N is (tup X1 Y1), 
  location M is (tup X2 Y2), 
  N != M,
  minus X1 Y1 == minus X2 Y2.
#forbid
  location N is (tup X1 Y1),
  location M is (tup X2 Y2),
  N != M,
  plus X1 Y1 == plus X2 Y2.
\end{verbatim}    
\end{minipage}
\begin{minipage}[b]{0.35\textwidth}
\begin{verbatim}
col N is? X :-
  dim N, dim X.
row N is? Y :-
  dim N, dim Y.
  
#forbid
  col N is X, col M is X,
  N != M.
  
#forbid 
  row N is Y, row M is Y,
  N != M.
  
#forbid 
  row N is X1, col N is Y1,
  row M is X2, col M is Y2,
  N != M,
  minus X1 Y1 == minus X2 Y2.
#forbid 
  row N is X1, col N is Y1,
  row M is X2, col M is Y2,
  N != M,
  plus X1 Y1 == plus X2 Y2.
\end{verbatim}
\end{minipage}
\begin{minipage}[b]{0.2\textwidth}
\begin{verbatim}





rowFor X is? Y :-
  dim X, dim Y.

  
colFor Y is X :-
  rowFor X is Y.


posDiag (plus X Y)
  is (tuple X Y) :-
  rowFor X is Y.


negDiag (minus X Y)
  is (tuple X Y) :-
  rowFor X is Y.


\end{verbatim}
\end{minipage}
\end{centering}
\caption{Three different implementations of N-Queens in {\dusa}.}
\label{fig-n-queens-in-dusa}
\end{figure}

The notation \verb|{ queen(I, J) } == n| in \Cref{fig-basic-n-queens} is an aggregate operation: it counts the number of occurrences of the \verb|queen| predicate. Aggregates in logic programming are, in many ways, similar to negation, because they allow one part of a program to reflect on the presence \textit{and absence} of information in another part of the program. {\dusa} currently supports neither negation nor aggregation, but a future implementation of {\casping} could potentially support a stratified form of both. In order to adapt the N-queens problem to {\casping}, we need a way to ensure that there will be exactly N queens without aggregation.

Of the three programs in \Cref{fig-n-queens-in-dusa}, the first is the most sculptural: we number each of the N queens and assign a unique coordinate pair to each queen. Our prefix-firing cost semantics suggests that the first rule in this program has a complexity of $O(n^3)$ for $n$ queens, and this observation also suggests a slight improvement. The middle program in \Cref{fig-n-queens-in-dusa} is a rewrite of the first one where we switch to sixth normal form (as the literature on LogiQL suggests we should prefer \cite{aref15logicblox}). Our cost semantics predicts more manageable behavior, quadratic instead of cubic, for this middle program. Despite the fact that any polynomial effect is dwarfed by exponential backtracking, the intuition it provides seems good: the middle program performs slightly better than the left program when we test performance.

If we think a bit about the problem, we can avoid the use of \verb|#forbid| constraints entirely, as shown in the third column of Figure~\ref{fig-n-queens-in-dusa}. Each column is labeled by an $x$ coordinate, and each row is labeled by a $y$ coordinate, and there is a functional dependency in both directions: each row has a queen in one column, and each column has a queen in one row. This is expressed by the first two rules in the rightmost program. Furthermore, we can assign every cell's coordinates to a unique numbered diagonal by adding the coordinates (for the diagonals that go up and to the right) or subtracting the coordinates (for the diagonals that go down and to the right). The last two rules in the third program use functional dependencies to invalidate any assignment that places two queens on the same diagonal.

\paragraph{Better Clingo Implementation}

\begin{figure}
\begin{centering}
\footnotesize
\begin{minipage}[b]{0.5\textwidth}
\begin{verbatim}

{ queen(I,1..n) } == 1 :- I = 1..n.
{ queen(1..n,J) } == 1 :- J = 1..n.

 :- { queen(D-J,J) } >= 2, D = 2..2*n.
 :- { queen(D+J,J) } >= 2, D = 1-n..n-1.\end{verbatim}
\end{minipage}
\end{centering}
\caption{Optimized N-queens implementation in Clingo.}
\label{fig-better-n-queens}
\end{figure}

The program we worked our way towards in the previous section uses the same ideas as one of Torsten Schaub's improved Clingo programs, shown in Figure~\ref{fig-better-n-queens}. Each of the four rules in Schaub's more efficient Clingo program corresponds loosely with one of the four rules in our final program in Figure~\ref{fig-n-queens-in-dusa}, but Schaub's version uses Clingo's counting constraints to enforce the at-most-one-queen-per-diagonal property that is enforced by functional dependencies in {\casping}.

\paragraph{Performance}

\begin{figure}
\includegraphics[trim={2.5cm 21.3cm 4cm 3.2cm},clip,width=13.8cm]{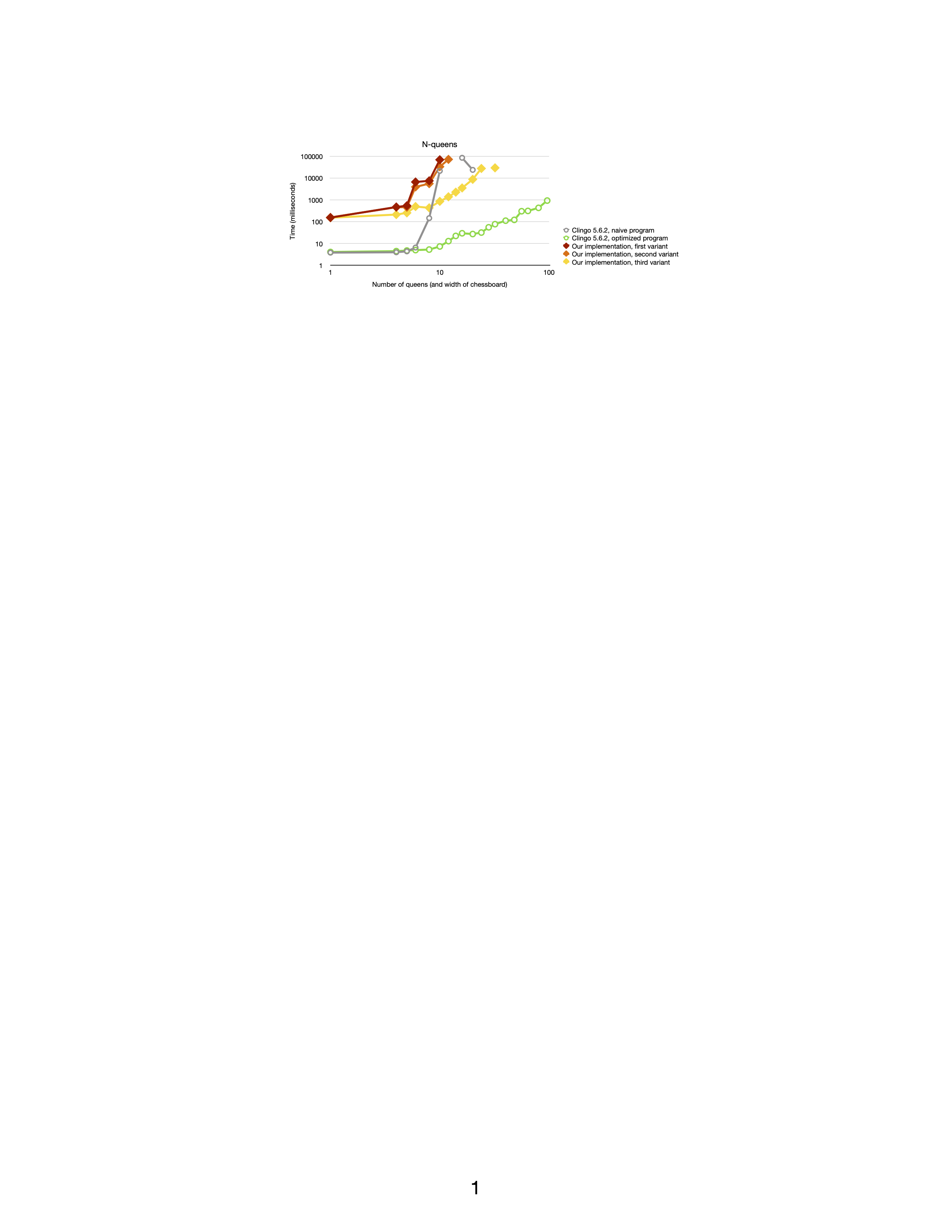}
\caption{Performance of the previous five implementations of the N-queens problem in {\dusa} and Clingo. All data points are the median of three runs requesting 10 different solutions, discarding timeouts.}
\label{queens-analysis}
\end{figure}

The performance of the two Clingo programs and three {\dusa} programs from this section is shown in Figure~\ref{queens-analysis} with the problem size scaling from 1 to 100 queens. The first Clingo program scales quite badly and starts hitting the 100-second timeout in at least some runs starting at 12 queens, whereas the optimized Clingo program in Figure~\ref{fig-better-n-queens} is far and away the best. The three successive {\dusa} programs fall in the middle: the third program from Figure~\ref{fig-n-queens-in-dusa} is the best but starts to encounter timeouts at 28 queens. 

\subsubsection{Map-Generation Speedrun}

\begin{figure}
\begin{centering}
\footnotesize
\begin{minipage}[b]{0.4\textwidth}
\begin{verbatim}
dim(1..W) :- width(W).


{ solid(X,Y) :dim(X) :dim(Y) }.


start(1,1).
finish(W,W) :- width(W).
step(0,-1).
step(0,1).
step(1,0).
step(-1,0).

reachable(X,Y) :-
    start(X,Y),
    solid(X,Y).
reachable(NX,NY) :-
    reachable(X,Y),
    step(DX,DY),
    NX = X + DX,
    NY = Y + DY,
    solid(NX,NY).

complete :- finish(X,Y),
    reachable(X,Y).
:- not complete.

at(X,Y, 0) :-
    start(X,Y),
    solid(X,Y).
at(NX,NY, T+1) :-
    at(X,Y,T),
    length(Len),
    T < Len,
    step(DX,DY),
    NX = X + DX,
    NY = Y + DY,
    solid(NX,NY).

speedrun :- finish(X,Y), at(X,Y,T).
:- speedrun.
\end{verbatim}    
\end{minipage}
\begin{minipage}[b]{0.4\textwidth}
\begin{verbatim}
#builtin INT_PLUS plus
#builtin INT_MINUS minus

dim Width :- width is Width.
dim (minus N 1) :- dim N, N != 1.

solid X Y is { tt, ff } :-
    dim X, dim Y.

start 1 1.
finish W W :- width is W.
step 0 -1.
step 0 1.
step 1 0.
step -1 0.

reachable X Y :- 
    start X Y,
    solid X Y is tt.
reachable NX NY :-
    reachable X Y,
    step DX DY,
    NX == plus X DX,
    NY == plus Y DY,
    solid NX NY is tt.

#demand finish X Y,
    reachable X Y.


at X Y 0 :- 
    start X Y,
    solid X Y is tt.
at NX NY (plus T 1) :-
    at X Y T,
    length is Len,
    T < Len,
    step DX DY,
    NX == plus X DX,
    NY == plus Y DY,
    solid NX NY is tt.

speedrun :- finish X Y, at X Y _.
#forbid speedrun.
\end{verbatim}
\end{minipage}
\end{centering}
\caption{Adam Smith's Map Generation Speedrun in Clingo's dialect of ASP, and an analogous (indirect) interpretation of the speedrun in {\dusa}.}
\label{fig-map-generation-speedrun}
\end{figure}

Adam Smith's blog post ``A Map Generation Speedrun with Answer Set Programming'' \cite{speedrun} is a beautiful example of what he calls the design-space sculpting approach to ASP \cite{smith2011answer}. On a tile grid, we stipulate that every tile can be solid or not, and we request maps where only a circuitous route through solid regions can reach the bottom-right corner of the map from the upper-left corner. 

Figure~\ref{fig-map-generation-speedrun} shows that ASP's expressiveness translates well into {\casping}, but note that this is an indirect interpretation, not a direct or mechanical translation. The fundamental choice --- is a tile solid or not? --- is naturally expressed as a closed choice rule.

\begin{figure}
\includegraphics[trim={2.5cm 21.3cm 4cm 3.15cm},clip,width=13.8cm]{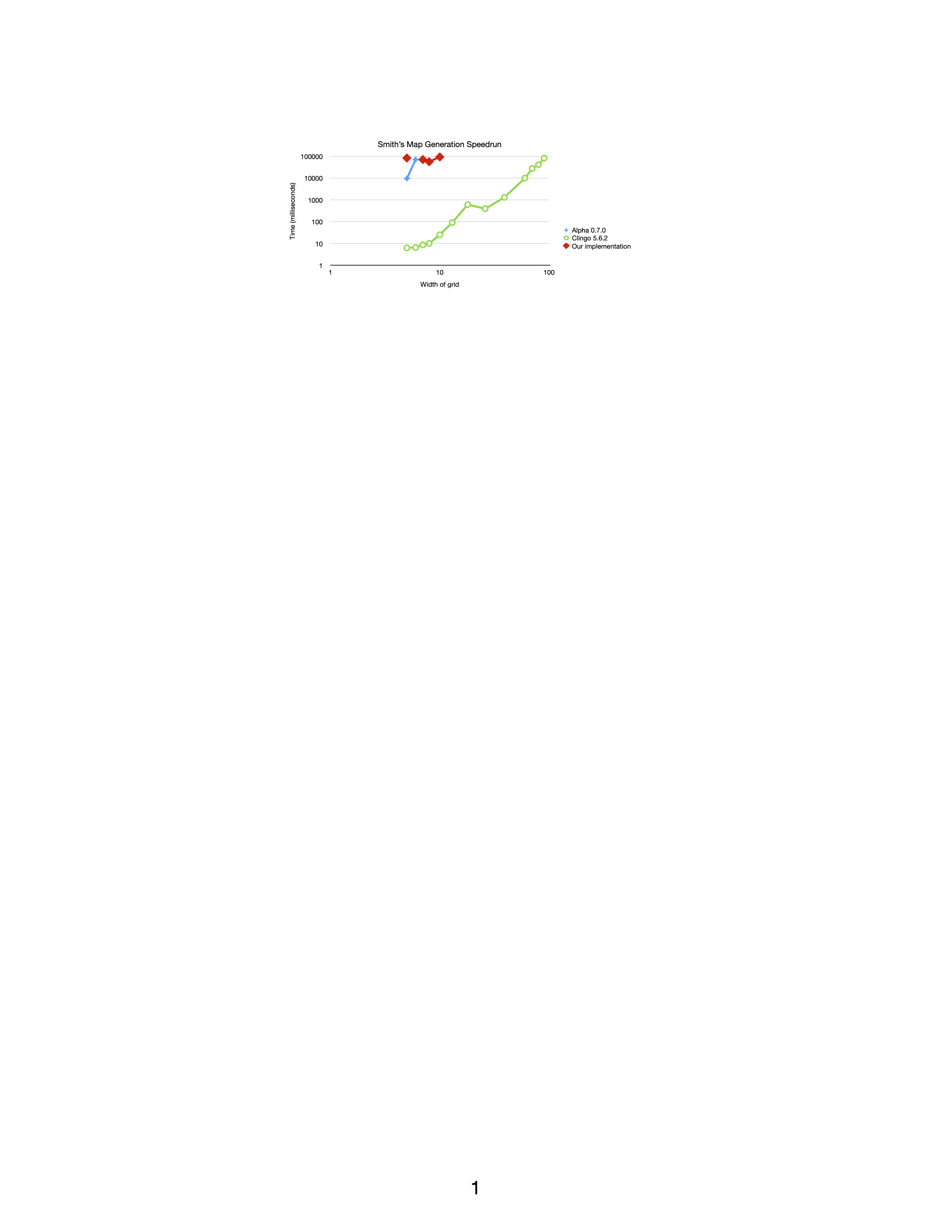}
\caption{Performance of {\dusa}, Clingo, and Alpha on the map generation programs in \Cref{fig-map-generation-speedrun}. All data points are the median of three runs requesting five different maps where the shortest path from upper-left to lower-right on an $n \times n$ grid was more than $2n + 1$ steps long, discarding timeouts.}
\label{map-generation-analysis}
\end{figure}

The map generation speedrun is an excellent demonstration of {\dusa}'s current lack of support for Smith and Mateas's design-space sculpting approach: it times out for all grids larger than $10 \times 10$, and times out at least once at every size, including $5 \times 5$ grids. We did not expect Alpha's relatively poor showing: it was unable to handle any grids larger than $6 \times 6$.

\subsection{Alpha Benchmarks}\label{sec-alpha-benchmarks}

Using a benchmark someone else wrote is a good way to keep yourself honest. We are therefore grateful to
Weinzirel \shortcite{alpha}, who published an easily accessible set of benchmarks for the Alpha implementation of answer set programming with lazy grounding \cite{alphabenchmark}. All the answer set programs in this benchmark are amenable to direct, mechanical translation to {\casps}.

Our hypothesis was that this benchmark would be somewhat kind to {\dusa}, since lazy grounding ASP and the {\casp} translation of ASP share many of the same advantages over the ground-then-solve approach. By in large, the results were similar to our other benchmarks: {\dusa} performs impressively well except on problems that obviously rely heavily on exponential backtracking search.

Weinzirel's test suite has four sections, which we will discuss in turn. Unless otherwise indicated, we evaluate {\dusa} on the direct translation of the ASP programs that are shown.

\subsubsection{Ground Explosion}

\begin{figure}
\begin{centering}
\footnotesize
\begin{minipage}[b]{0.8\textwidth}
\begin{verbatim}
p(X1,X2,X3,X4,X5,X6) :-
   select(X1), select(X2), select(X3),
   select(X4), select(X5), select(X6).

select(X) :- dom(X), not nselect(X).
nselect(X) :- dom(X), not select(X).
 :- not nselect(Y), select(X), dom(Y), X != Y.
\end{verbatim}
\end{minipage}
\end{centering}
\caption{Ground explosion benchmark.}
\label{fig-ground-explosion}
\end{figure}

The ground explosion benchmarks use the program in Figure~\ref{fig-ground-explosion} to trigger pathological behavior in ground-then-solve ASP solvers. For a domain with $n$ elements, a ground-then-solve solver will need to ground $6^n$ variants of the first rule, despite the fact that there are only $n+1$ solutions.

\begin{figure}
\includegraphics[trim={2.5cm 16.5cm 4cm 9cm},clip,width=13.8cm]{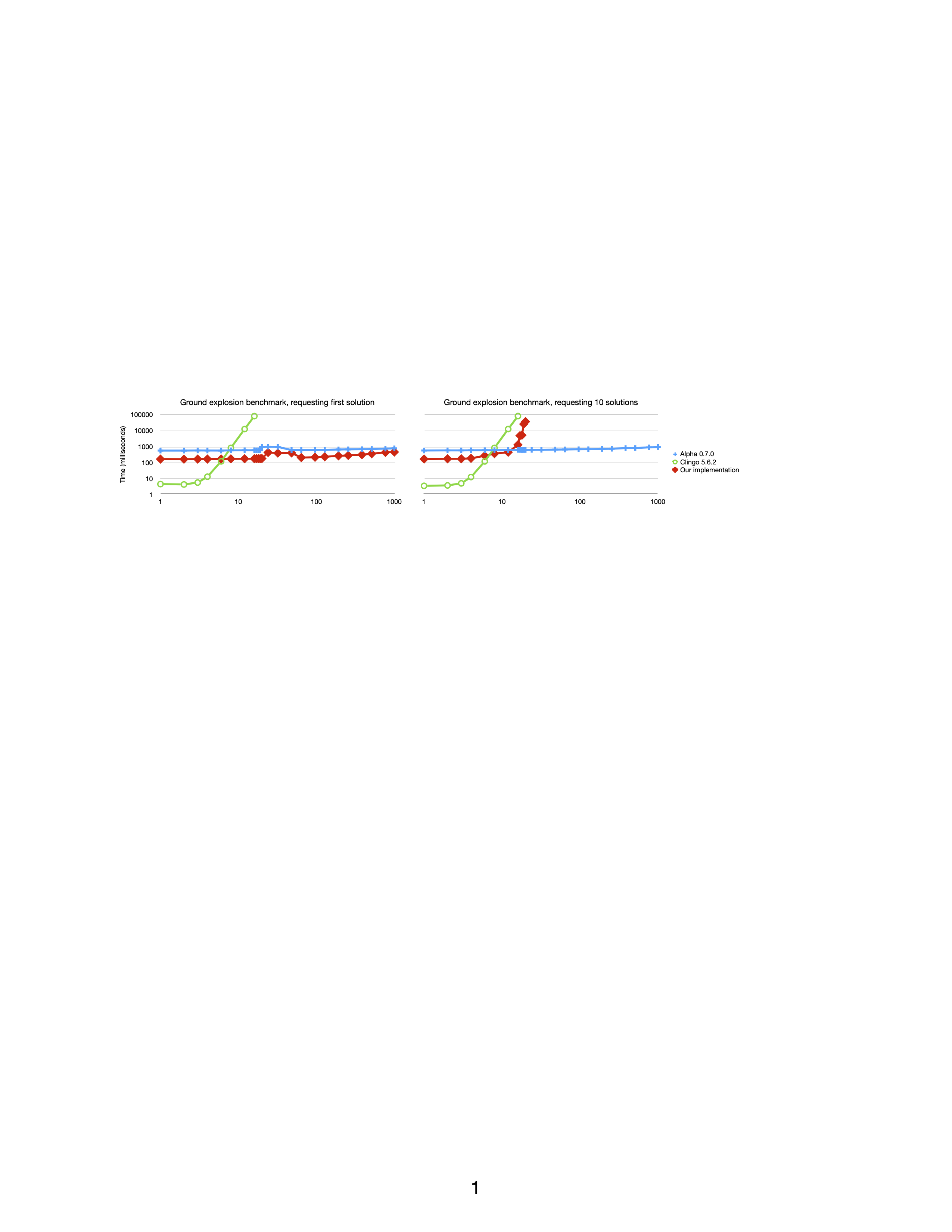}
\caption{Performance of the {\dusa}, Clingo, and Alpha solvers on the ground explosion benchmark in \Cref{fig-ground-explosion}. The $x$-axis is the problem size. All data points are the median of three runs, discarding timeouts.}
\label{ground-explosion-analysis}
\end{figure}

The results from this set of Alpha benchmarks are shown in \Cref{ground-explosion-analysis}. This benchmark shows both the strength of the {\dusa} implementation and the degree to which backtracking is done naively and inefficiently in {\dusa}. The first solution is found immediately, matching the asymptotic behavior of the Alpha implementation and improving on constant factors, but enumerating additional solutions is done in a way that results in exponential backtracking, so {\dusa} is unable to return 10 distinct solutions for most examples. We expect it to be relatively straightforward to dramatically improve {\dusa}'s behavior here.

\subsubsection{Graph 5-Colorability}

\begin{figure}
\begin{centering}
\footnotesize
\begin{minipage}[b]{0.45\textwidth}
\begin{verbatim}
chosenColour(N,C) :- node(N), colour(C),
  not notChosenColour(N,C).
notChosenColour(N,C) :- node(N), colour(C),
  not chosenColour(N,C).

:- node(X), not colored(X).
colored(X) :- chosenColour(X,Fv1).

:- node(N), chosenColour(N,C1),
  chosenColour(N,C2), C1!=C2.

colour(red0).
colour(green0).
colour(blue0).
colour(yellow0).
colour(cyan0).

:- node(X), not chosenColour(X,red0),
  not chosenColour(X,blue0),
  not chosenColour(X,yellow0),
  not chosenColour(X,green0),
  not chosenColour(X,cyan0).

:- link(X,Y), node(X), node(Y), X<Y,
  chosenColour(X,C), chosenColour(Y,C).
\end{verbatim}    
\end{minipage}
\begin{minipage}[b]{0.45\textwidth}
\begin{verbatim}
color N is {
  red,
  green,
  blue,
  yellow,
  cyan
} :- node N.






#forbid link X Y,  X < Y,
  color X is C, color Y is C.
\end{verbatim}
\end{minipage}
\end{centering}
\caption{Graph 5-colorability expressed in ASP, and an analogous (indirect) interpretation in {\dusa}.}
\label{fig-5-color}
\end{figure}

A graph is 5-colorable if each node can be assigned one of five different colors such that no edge connects nodes of the same color. Perhaps there is a dramatically clearer way to express this concept in ASP than the ASP code from the benchmark shown in \Cref{fig-5-color}, but it is remarkable how much simpler this is to express in {\casping}, as shown on the right side of \Cref{fig-5-color}.

\begin{figure}
\includegraphics[trim={2.5cm 16.2cm 4cm 9cm},clip,width=13.8cm]{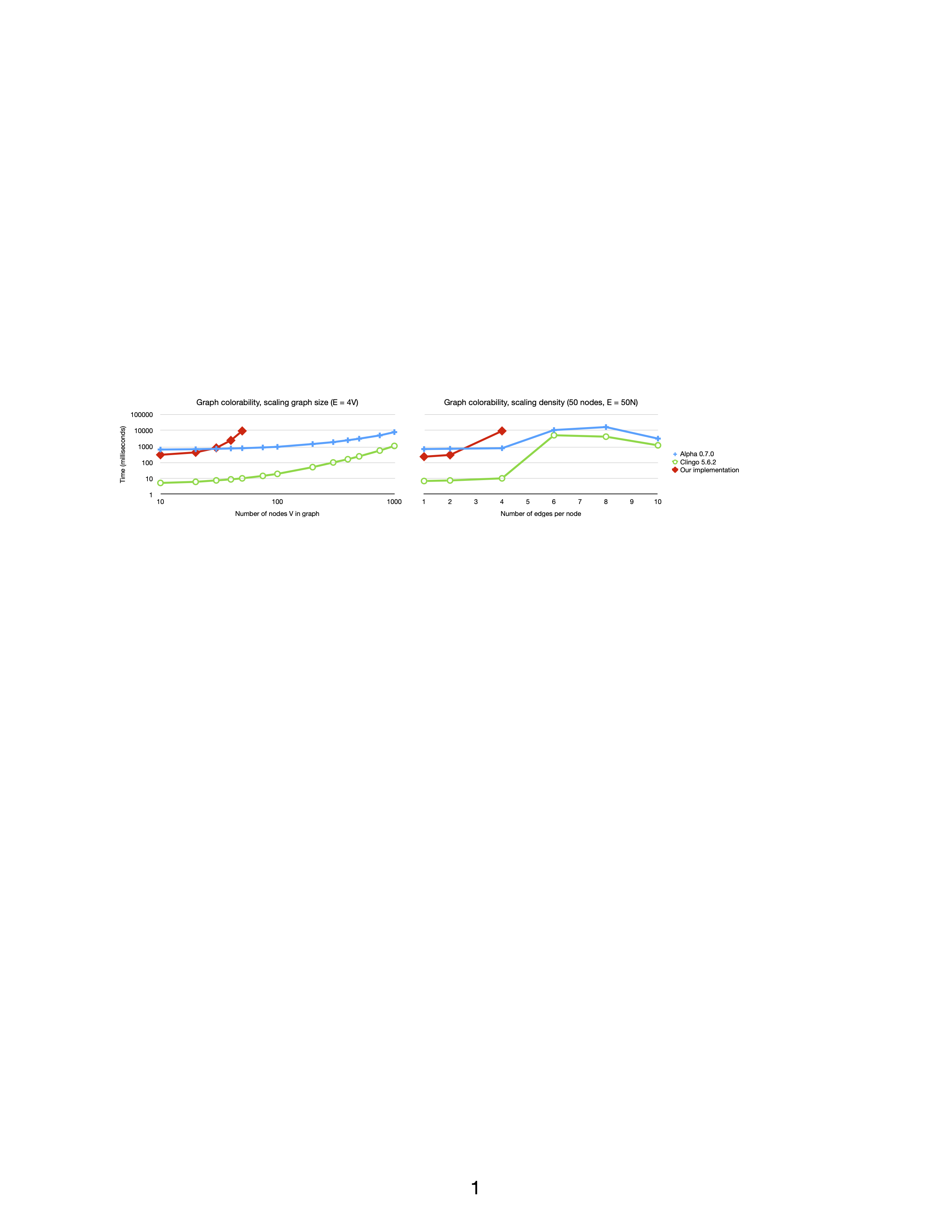}

\caption{Performance of the {\dusa}, Clingo, and Alpha solvers on 5-colorability benchmark from \Cref{fig-5-color}. All data points are the \textbf{average} of ten runs on different graphs requesting 10 different solutions each time, discarding timeouts.  (The same data points with 50 nodes and 200 edges are present in both graphs.)}
\label{5color-analysis}
\end{figure}

This is essentially a pure search problem, so it's not surprising that the {\dusa} implementation performs quite badly. We didn't bother testing the performance of the {\casp} that directly corresponds to the ASP program in \Cref{fig-5-color}, since the idiomatic program already exhibited extremely poor performance.

\subsubsection{``Cutedge'' Benchmark}

\begin{figure}
\begin{centering}
\footnotesize
\begin{minipage}[b]{0.8\textwidth}
\begin{verbatim}
delete(X,Y) :- edge(X,Y), not keep(X,Y).
keep(X,Y) :- edge(X,Y), delete(X1,Y1), X1 != X.
keep(X,Y) :- edge(X,Y), delete(X1,Y1), Y1 != Y.
reachable(X,Y) :- keep(X,Y).
reachable(X,Y) :- special(Y),reachable(X,Z),reachable(Z,Y).
\end{verbatim}
\end{minipage}
\end{centering}
\caption{``Cutedge'' benchmark.}
\label{fig-cutedge}
\end{figure}

The ``cutedge'' benchmark shown in \Cref{fig-cutedge} removes an arbitrary edge from a graph and then calculates reachability from one endpoint of the removed edge and another specially designated node. The benchmark seems to be intended to force bad performance on ground-then-solve approaches without being as obviously unfair as the ground explosion benchmark. \Cref{cutedge-analysis} shows the very similar performance of {\dusa} and Alpha. Clingo only able to handle the smallest examples in under 100 seconds.

\begin{figure}
\includegraphics[trim={2.5cm 16.5cm 4cm 8.75cm},clip,width=13.8cm]{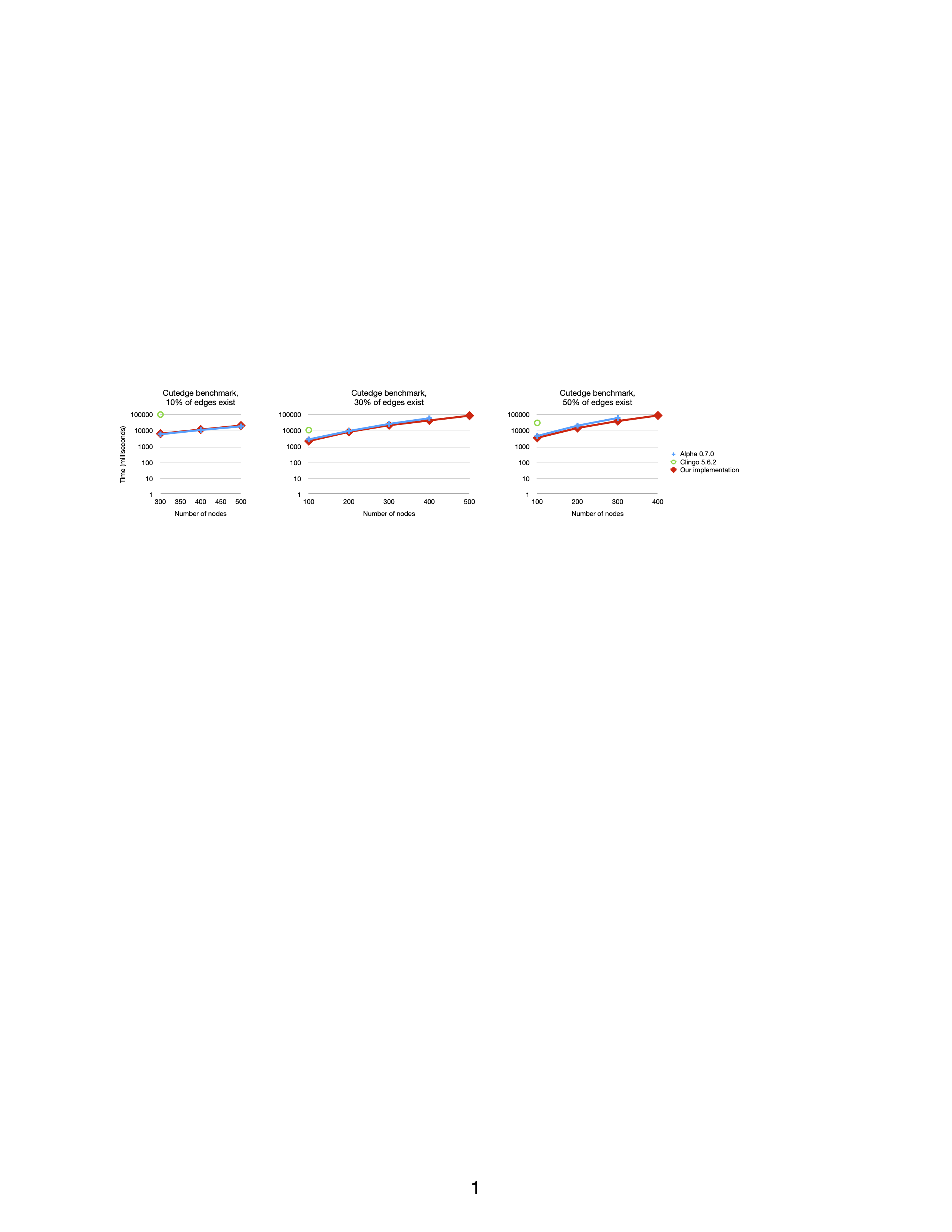}
\caption{Performance of the {\dusa} and Alpha solvers on the ``cutedge'' benchmark from \Cref{fig-cutedge} on random graphs with different edge probabilities. All data points are the \textbf{average} of ten runs on different graphs requesting 10 different solutions, discarding timeouts.}
\label{cutedge-analysis}
\end{figure}

\subsubsection{Graph Reachability Benchmark}

The graph reachability benchmark (\Cref{fig-reachability}) is included for completeness, but the test cases and results weren't particularly interesting.
It would be more interesting to rerun this test on the graphs used to test the spanning tree and canonical representative programs in order to tease apart different asymptotic behaviors.

\begin{figure}
\begin{centering}
\footnotesize
\begin{minipage}[b]{0.8\textwidth}
\begin{verbatim}
reachable(X, Y) :- edge(X, Y).
reachable(S, Y) :- start(S), reachable(S, X), reachable(X, Y).
\end{verbatim}
\end{minipage}
\end{centering}
\caption{Graph reachability.}
\label{fig-reachability}
\end{figure}

\begin{figure}
\begin{center}
\begin{tabular}{c  c c c} 
\toprule
 Instance size & Our implementation & Clingo & Alpha \\ [0.5ex] 
\midrule
 1000/4 & 0.4 & 0.05 & 0.8 \\ 
 \hline
 1000/8 & 0.5 & 0.1 & 0.9 \\
 \hline
 10000/2 & 1.9 & 1.1 & 1.2 \\
 \hline
 10000/4 & 3.6 & 2.3 & 1.6 \\
 \hline
 10000/8 & 2.3 & 5.1 & 2.1 \\ 
\bottomrule
\end{tabular}
\end{center}\caption{Performance of the {\dusa}, Clingo, and Alpha solvers on the graph reachability program from \Cref{fig-reachability}. Instance size is number of nodes / number of edges per node, running time in \textbf{seconds}, and all data points are the \textbf{average} of ten runs on different graphs requesting 10 different solutions each time (though there is in fact only one unique solution in each case).}
\label{reachability-analysis}
\end{figure}

\end{document}